\documentclass[a4paper,onesided,12pt]{report}
\usepackage{fbe_tez}
\usepackage[utf8x]{inputenc} % To use Unicode (e.g. Turkish) characters

\usepackage{amsmath, amsthm, amssymb}
\usepackage[bottom]{footmisc}
\usepackage{cite,comment}
\usepackage{graphicx}
\usepackage{enumerate}
\usepackage{longtable}
\graphicspath{{figures/}} % Graphics will be here
\usepackage{comment}
\usepackage{multirow}
\usepackage{subfigure}
\usepackage{algorithm}
\usepackage{algorithmic}

\newcommand{\dollar}[0]{\$}
\newcommand{\modm}{\mathtt{MOD_m}}

\newcommand{\Leq}{\mathtt{EQ}}
\newcommand{\modtwothree}{\mathtt{MOD23}}

\newcommand{\mymatrix}[2]{\left( \begin{array}{#1} #2\end{array} \right)}

\newcommand{\mypar}[1]{\left( #1 \right)}

\newtheorem{thm}{Theorem}[chapter]
\newtheorem{prop}[thm]{Proposition}
\newtheorem{lem}[thm]{Lemma}
\newtheorem{cor}[thm]{Corollary}
\newtheorem{fact}[thm]{Fact}

\newtheorem{ex}[thm]{Example}

\title{EXTENDED MODELS OF FINITE AUTOMATA}
\turkcebaslik{GÜÇLENDİRİLMİŞ SONLU DURUMLU MAKİNE MODELLERİ}
\degree{B.S., Mathematics, Boğaziçi University, 2011\\
	M.S., Computer Engineering, Boğaziçi University, 2013}
\author{Özlem Salehi Köken}
\program{Computer Engineering}
\subyear{2019}

\supervisor{Prof. A. C. Cem Say}
\examineri{Assoc. Prof. Özlem Beyarslan} 
\examinerii{Prof. Sema Fatma Oktuğ}
\examineriii{Assoc. Prof. Tolga Ovatman}
\examineriv{Prof. Can Özturan}
\dateofapproval{12.06.2019}

\begin{document}

\pagenumbering{roman}
%\makemstitle % M.S. thesis
\makephdtitle 
\makeapprovalpage
\newpage

\begin{flushright}

\vspace*{\fill}
To the memory of my grandfather Hasan Salehi,
\\
my great-grandmother Gülizar Burhanetin, and 
\\
my grandmother Mader Burhanettin...
	
\end{flushright}

\begin{acknowledgements}

First of all I would like to express my gratitude to my supervisor Prof. A. C. Cem Say for his guidance and support throughout my journey at Boğaziçi University. He has always been a source of inspiration for me and it is a great pleasure to be his student. I would also like to thank Abuzer Yakaryılmaz for his significant contributions to my thesis. His enthusiasm has always motivated and encouraged me.

I would like to thank my thesis committee members Prof. Can Özturan and Assoc. Prof. Özlem Beyarslan for their valuable feedbacks over the years and devoting their time to this research. I would like to thank Prof. Sema Fatma Oktuğ and Assoc. Prof. Tolga Ovatman for kindly accepting to be in my thesis jury and for their helpful comments.

I am grateful to Dr. Flavio D'Alessandro for our joint works and for sharing his knowledge. It was a great experience to work with him. I would like to thank professors Ryan O'Donnell, Alexandre Borovik and Igor Potapov for their helpful answers to my questions. 

My gratitudes go to all members of the Department of Computer Engineering for providing such a nice environment. I would like to especially thank Prof. Pınar Yolum Birbil and Prof. Cem Ersoy for their kindness and support. I want to thank all past and current members of BM26 with whom I had the chance of working with. 

I would like to thank Pelin Gürel for her valuable friendship all through these years.

I would like to offer my deepest gratitude to my grandfather İbrahim Burhanettin, my aunt Perihan Burhanettin, and my parent-in-law Nebahat Köken for their endless love and faith in me. I want to thank my sister Meltem Salehi for always being cheerful and motivating me. I am indebted to my parents Feryal and Mecid Salehi for all their efforts and generosity which made this thesis possible. Finally I would like to express my love and gratitude to my husband Oktay Köken for his patience, guidance and encouragement and to my lovely daughter Eda for bringing joy to my life. Without them, it would have no meaning.

This work is supported by Boğaziçi University Research Fund Grant Number 11760 and by Scientific and Technical Research Council of Turkey (TÜBİTAK) BİDEB Scholarhip Program.

\end{acknowledgements}
\begin{abstract}

Many of the numerous automaton models proposed in the literature can be regarded as a finite automaton equipped with an additional storage mechanism. In this thesis, we focus on two such models, namely the finite automata over groups and the homing vector automata.

A finite automaton over a group $ G $ is a nondeterministic finite automaton equipped with a register that holds an element of the group $ G $. The register is initialized to the identity element of the group and a computation is successful if the register is equal to the identity element at the end of the computation after being multiplied with a group element at every step. We investigate the language recognition power of finite automata over integer and rational matrix groups and reveal new relationships between the language classes corresponding to these models. We examine the effect of various parameters 
%such as the growth rate of the groups and the run-time of the machines 
on the language recognition power. We establish a link between the decision problems of matrix semigroups and the corresponding automata.  We present some new results about 
%computational models which are closely related to finite automata over groups, namely the 
valence pushdown automata and context-free valence grammars.

We also propose the new homing vector automaton model, which is a finite automaton equipped with a vector that can be multiplied with a matrix at each step. The vector can be checked for equivalence to the initial vector and the acceptance criterion is ending up in an accept state with the value of the vector being equal to the initial vector. We examine the effect of various restrictions on the model by confining the matrices to a particular set and allowing the equivalence test only at the end of the computation. We define the different variants of the model and compare their language recognition power with that of the classical models. \vspace{-5in}
\end{abstract}

\begin{ozet}
Literatürde ortaya sürülmüş olan pek çok makine, bir sonlu durumlu makinenin ek bir hafıza ünitesi ile güçlendirilmiş hali olarak düşünülebilir. Bu tezde, bu makine-\newline lerden ikisine, gruplar üzerinde tanımlı sonlu durumlu makinelere ve eve dönen vektör makinelerine odaklanılmıştır.

$ G $ grubu üzerinde tanımlı bir makine, ek hafıza ünitesinde $ G $ grubundan bir elemanı tutma hakkına sahip, belirlenimci olmayan bir sonlu durumlu makinedir. Başlangıçta hafıza ünitesinin değeri $ G $ grubunun birim elemanıdır. Bir hesaplamanın başarılı sayılabilmesi için hafıza ünitesinin değeri, her adımda grubun bir elemanıyla çarpıldıktan sonra bitimde grubun birim elemanına eşit olmalıdır. Bu çalışmada tam sayılı ve rasyonel sayılı matris grupları üzerinde tanımlanan sonlu durumlu makinelerin tanıdıkları dil sınıfları incelenmiştir. Çeşitli parametrelerin makinelerin tanıma gücünü nasıl etkilediği araştırılmıştır. Matris yarıgruplarının karar verme problemleri ile ilintili makinelerinki arasında bir bağ kurulmuştur. Grup üzerinde tanımlı makinelerle ilişkili olan bazı modellerle ilgili yeni sonuçlar elde edilmiştir.

Yeni tanımladığımız eve dönen vektör makinesi, bir sonlu durumlu makinenin bir vektörle güçlendirilmesi ve bu vektöre her adımda bir matrisle çarpılma hakkı verilmesiyle ortaya çıkmıştır. Vektörün başlangıç vektörüne eşit olup olmadığı kontrol edilebilir ve makinenin kabul şartı, hesaplama bittiğinde vektörün başlangıç vektörüne eşit olması ve kabul durumlarından birinde bulunulmasıdır. Kullanılan matris kümesi sınırlanarak ve vektörün eşitlik kontrolünün sadece sonda gerçekleşmesine izin verilerek, farklı kısıtlamaların makineye olan etkisi incelenmiştir. Makinenin çeşitli sürümlerinin dil tanıma gücüyle klasik modellerin dil tanıma gücü karşılaştırılmıştır. 
%Matris grupları üzerinde tanımlı sonlu durumlu makinelerle, bir yönlü belirlenimci olmayan kör eve dönen vektör maki-\newline neleri arasında ilişki kurularak eve dönen vektör makineleri ile ilgili yeni sonuçlar elde edilmiştir. Gerçek zamanlı eve dönen vektör makineleri üzerinde özellikle durulmuş, bazı kapalılık özellikleri gösterilmiş ve bu makinelerin tek durumlu versiyonları analiz edilmiştir. Stern-Brocot ağacından faydalanarak dizgeleri vektörlere kodlamaya yarayan bir metod geliştirilmiştir.
\vspace{-5in} 
\end{ozet}
\tableofcontents
\listoffigures
\listoftables
\begin{symbols}
% The title will be typeset as "LIST OF SYMBOLS".
%
% Use a separate \sym command for each symbols definition.
% First Latin symbols in alphabetical order

\sym{$ 1 $}{The identity element of a group/monoid}
\sym{$ |A| $}{The cardinality of a set $ A $}
\sym{$ \mathcal{A} $}{A pushdown automaton}
\sym{$ A[i,j] $}{The entry in the $ i $'th row and the $ j $'th  column of the matrix $ A $}
\sym{$ A_L(n) $}{The maximum $ k $ such that there exist $ k $ distinct strings that are pairwise $ n $-dissimilar for $ L $ }
\sym{B}{The bicyclic monoid}
\sym{$ B^A_G(n) $}{The set of all elements in 
$ G $ which can be represented by a word of length at most $ n $}
\sym{$ BS(m,n) $}{The Baumslag Solitar group}
\sym{$ \mathcal{C} $}{A counter automaton}
\sym{$ \mathcal{D} $}{The set of head directions}
\sym{$ e $}{The identity element of a group/monoid}
\sym{$ e_m(w) $}{The base-10 number encoded by $ w $ in base-$ m $}
\sym{$ \mathcal{E} $}{An extended finite automaton}
\sym{$ \mathcal{F} $}{A finite automaton}
\sym{$ \mathbf{F}_r $}{The free group of rank $ r $}
\sym{$ G $}{A group}
\sym{$ \mathcal{G} $}{A context-free grammar}
\sym{$ \mathfrak{G} $}{A context-free valence grammar}
\sym{$ g_G(n) $}{The growth function of a group $ G $}

\sym{$ GL(n,\mathbb{Q}) $}{The general linear group of order $ n $ over the field of rationals}
\sym{$ GL(n,\mathbb{Z}) $}{The general linear group of order $ n $ over the field of integers}
\sym{$ \mathbf{H} $}{The discrete Heisenberg group}
\sym{$ {I} $}{Identity matrix}
\sym{$ \mathtt{I} $}{An ideal}

\sym{$ L $}{A language}
\sym{$\bar{ L} $}{The complement of the language $ L $}
\sym{$ L(.) $}{The language generated by some grammar/recognized by some machine}
\sym{$ \mathfrak{L}(G)$}{The class of languages recognized by $ {G} $-automata}
\sym{$ \mathfrak{L}(G)_{t(n)}^s $}{The class of languages recognized by strongly $ t(n) $-time bounded $ G $-automata}
\sym{$ \mathfrak{L}(G)_{t(n)}^w $}{The class of languages recognized by weakly $ t(n) $-time bounded $ G $-automata}
\sym{$ \mathfrak{L}(\cal{M})$}{The class of languages recognized by the machine type $ \cal{M} $}
\sym{$\mathfrak{L}(\textup{Val}, \textup{CF},M)$}{The class of languages generated by context-free valence grammars over $ M $}
\sym{$\mathfrak{L}(\textup{Val}, \textup{NFA},M)$}{The class of languages recognized by valence automata over $ M $}
\sym{$\mathfrak{L}(\textup{Val}, \textup{PDA},M)$}{The class of languages recognized by valence pushdown automata over $ M $}
\sym{$\mathfrak{L}(\textup{Val}, \textup{REG},M)$}{The class of languages generated by regular valence grammars over $ M $}
\sym{ M}{The set of matrices multiplied with the vector of a homing vector automaton}
\sym{$ M$}{A monoid}
\sym{$ M_n(\mathbb{Z})$}{The set of $ n \times n $ matrices with integer entries }
\sym{$ \mathbb{N} $}{The set of natural numbers}
\sym{$o(.)$}{Little-o notation}
\sym{$O(.)$}{Big-O notation}
\sym{$\mathfrak P$}{A valence pushdown automaton}
\sym{$\mathcal{P}(A)$}{The power set of a set $ A $}
\sym{P$_2$}{The polycyclic monoid of rank 2}
\sym{$P(X)$}{The polycyclic monoid on $ X $ }
\sym{$ \mathbb{Q} $}{The set of rational numbers}
\sym{$ \mathbb{Q}^+ $}{The set/group of nonzero rational numbers}
\sym{$ Q $}{The set of states}.
\sym{$ q_1 $}{The initial state}
\sym{$ Q_a $}{The set of accept states}
\sym{$ S $}{A semigroup}
\sym{$ S/\mathtt{I} $}{The Rees quotient semigroup}
\sym{$ S_k(m) $}{The set of $ k \times k $ matrices whose entries belong to the set $ \{-m,\dots,m\} $ for some positive integer $ m $}
\sym{$ SL(n,\mathbb{Q}) $}{The special linear group of order $ n $ over the field of rationals}
\sym{$ SL(n,\mathbb{Z}) $}{The special linear group of order $ n $ over the field of integers}
\sym{$ \mathcal{T} $}{A Turing machine}
\sym{$ U_L(n) $}{The maximum $ k $ such that there exist $ k $ distinct strings that are uniformly $ n $-dissimilar for $ L $ }
\sym{$ v $}{A vector}
\sym{$ \cal V $}{A vector automaton}
\sym{$ v[i] $}{The $ i $'th entry of the vector $ v $}
\sym{$ \mathcal{W} $}{A finite automaton with multiplication}
\sym{$ |w| $}{The length of the string $ w $}
\sym{$ W(G) $}{The word problem language of the group $ G $}
\sym{$ |w|_{\sigma} $}{The number of occurrences of $ \sigma $ in $ w $}
\sym{$ w[i] $}{The $ i $'th symbol of the string $ w $}
\sym{$ w^r$}{The reverse of the string $ w $}
\sym{$ \mathbb{Z} $}{The set of integers}
\sym{$ \mathbb{Z}^k $}{The set/group of $ k $-dimensional integer vectors}
\\

\sym{$ \delta $}{A transition function}
\sym{$ \varepsilon $}{The empty string}
\sym{$ \Gamma $}{A tape/stack alphabet}
\sym{$ \Gamma_{\varepsilon} $}{$ \Gamma \cup \{\varepsilon\} $}
\sym{$ \Lambda $}{The set of multipliers of a finite automaton with multiplication}
\sym{$ \omega $}{A register status}
\sym{$ \Omega $}{The set of register status}
\sym{$ \phi(L) $}{The Parikh image of a language $ L $}
\sym{$ \phi(w) $}{The Parikh image of a string $ w $}
\sym{$ \sigma $}{A symbol}
\sym{$ \Sigma $}{An alphabet}
\sym{$ \Sigma^* $}{The set of all strings over $ \Sigma $}

\sym{$ \Sigma_{\dollar} $}{$ \Sigma \cup \{\dollar \}$}
\sym{$ \Sigma_{\varepsilon} $}{$ \Sigma \cup \{\varepsilon \}$}
\sym{$ \theta $}{A counter status}
\sym{$ \Theta $}{The set of counter status}
\\

\sym{$ \simeq $}{Isomorphic}
\sym{$ \downarrow $}{Stay}
\sym{$ \rightarrow $}{Right}
\sym{$ \leftarrow $}{Left}
\sym{$ \sqcup $}{Blank symbol}
\sym{$ \dollar $}{End-marker}
\end{symbols}

\begin{abbreviations}
 % Abbreviations in alphabetical order
\sym{0-1DFAMW}{Stateless one-way deterministic finite automaton with multiplication without equality}
\sym{0-DBHVA($ k $)}{Stateless real-time $ k $-dimensional deterministic blind homing vector automaton}
\sym{0-DFAMW}{Stateless real-time deterministic finite automaton with multiplication without equality}
\sym{0-DHVA($ k $)}{Stateless real-time $ k $-dimensional deterministic homing vector automaton}
\sym{0-NBHVA($ k $)}{Stateless real-time $ k $-dimensional nondeterministic blind homing vector automaton}
\sym{0-NFAMW}{Stateless real-time nondeterministic finite automaton with multiplication without equality}
\sym{0-NHVA($ k $)}{Stateless real-time $ k $-dimensional nondeterministic homing vector automaton}
\sym{1DBHVA($ k $)}{One-way $ k $-dimensional deterministic blind homing vector automaton}
\sym{1DFA}{One-way deterministic finite automaton}
\sym{1DFAM}{One-way deterministic finite automaton with multiplication}
\sym{1DFAMW}{One-way deterministic finite automaton with multiplication without equality}
\sym{1DHVA($ k $)}{One-way $ k $-dimensional deterministic homing vector automaton}
\sym{1D$ k $BCA}{One-way deterministic blind $ k $-counter automaton}
\sym{1D$ k $CA}{One-way deterministic $ k $-counter automaton}
\sym{1NBHVA($ k $)}{One-way $ k $-dimensional nondeterministic blind homing vector automaton}
\sym{1NHVA($ k $)}{One-way $ k $-dimensional nondeterministic homing vector automaton}
\sym{1NFA}{One-way nondeterministic finite automaton}
\sym{1NFAM}{One-way nondeterministic finite automaton with multiplication}
\sym{1NFAMW}{One-way nondeterministic finite automaton with multiplication without equality}
\sym{1N$ k $BCA}{One-way nondeterministic blind $ k $-counter automaton}
\sym{1N$ k $CA}{One-way nondeterministic $ k $-counter automaton}
\sym{1NPDA}{One-way nondeterministic pushdown automaton}
\sym{DBHVA($ k $)}{Real-time $ k $-dimensional deterministic blind homing vector automaton}
\sym{DFA}{Real-time deterministic finite automaton}
\sym{DFAM}{Real-time deterministic finite automaton with multiplication} 
\sym{DFAMW}{Real-time deterministic finite automaton with multiplication without equality}
\sym{DHVA($ k $)}{Real-time $ k $-dimensional deterministic homing vector automaton}
\sym{D$ k $BCA}{Real-time deterministic blind $ k $-counter automaton}
\sym{D$ k $CA}{Real-time deterministic $ k $-counter automaton}
\sym{DVA($ k $)}{Real-time $ k $-dimensional deterministic vector automaton}

\sym{CA}{Counter automaton}
\sym{$ \mathsf{CF} $}{The class of context-free languages}
\sym{FA}{Finite automaton}
\sym{FAM}{Finite automaton with multiplication}
\sym{HVA($ k $)}{$ k $-dimensional homing vector automaton}
\sym{$ k $BCA}{blind $ k $-counter automaton}
\sym{$ k $CA}{$ k $-counter automaton}
\sym{NBHVA($ k $)}{Real-time $ k $-dimensional nondeterministic blind homing vector automaton}
\sym{NHVA($ k $)}{Real-time $ k $-dimensional nondeterministic homing vector automaton}
\sym{N$ k $BCA}{Real-time nondeterministic blind $ k $-counter automaton}
\sym{N$ k $CA}{Real-time nondeterministic $ k $-counter automaton}
\sym{NFA}{Real-time nondeterministic finite automaton}
\sym{NFAM}{Real-time nondeterministic finite automaton with multiplication}
\sym{NFAMW}{Real-time nondeterministic finite automaton with multiplication without equality}
\sym{PDA}{Pushdown automaton}
\sym{$ \mathsf{RE} $}{The class of recursively enumerable languages}
\sym{$ \mathsf{REG} $}{The class of regular languages}
\sym{$  \mathsf{TISP}(t(n),s(n) $}{The class of languages that can be decided by a deterministic Turing machine within $ t(n) $ time and $ s(n) $ space where $ n $ is the length of the input }
\sym{TM}{Turing machine}
\end{abbreviations}

\chapter{INTRODUCTION}
\label{chapter:introduction}
\pagenumbering{arabic}
\section{Automata Theory}\label{sec: int-aut}
The theory of computation aims to investigate the computation process and tries to answer the question ``What are the fundamental capabilities and limitations of computers?'' as stated by Sipser \cite{Si06}. To study the computation process, we have to first formalize the notions of computational problems and computational models. 

At the heart of computational problems, lie the decision problems, the problems whose answers are either yes or no. Any decision problem can be represented by the set of instances which have the answer yes. Consider the problem of checking whether a number is prime. This problem can be represented by the set $ \{2,3,5,7,\dots\} $, where the members of the set are the prime numbers. Furthermore, the elements of the set can be represented as \textit{strings}, a finite sequence of \textit{symbols} belonging to a finite set called the \textit{alphabet}. For instance, the set of prime numbers can be expressed by the set $\{11,111,11111,1111111,\dots\}  $, where the alphabet contains the single symbol $ 1 $. 

Automata theory is the study of computational models that solve decision problems. An automaton is as an abstract machine which processes an input string and makes the decision of yes or no, more technically called as acceptance or rejection. Automata allow us to formalize the notion of computation and serve as mathematical models for computing devices. The set of all accepted input strings is called the \textit{language} recognized by the machine. If there exists a machine whose language is the set of yes instances of a problem, than we have a machine solving the problem. 

The most basic model which is known as the \textit{finite automaton} or \textit{finite state machine} is an abstract model for computation with finite memory. It is assumed that the input string is written on a tape and the machine has a tape-head reading the string from left to right. There exist finitely many states and a set of rules governing the transitions between these states. Computation starts from a designated initial state and one input symbol is consumed at each step. An input string is accepted if after reading the string, computation ends in a special state called the accept state and otherwise rejected.

One of the founders of the theory of formal languages, Noam Chomsky, defined four classes of languages, namely the classes of regular languages, context-free languages, context-sensitive languages and recursively enumerable languages, forming the Chomsky Hierarchy. Finite automata recognize exactly the class of regular languages.

\section{Extended Models of Finite Automata} \label{sec: int-efa}
Throughout the literature, a variety of automaton models has been proposed. Many different models of automata that have been examined can be regarded as a finite automaton augmented with some additional memory. The type of the memory, restrictions on how this memory is accessed, computation mode and the conditions for acceptance determine the expressiveness of the model in terms of language recognition. One can list pushdown automata \cite{Ch62}, counter machines \cite{FMR67} and Turing machines \cite{Tu37} among the many such proposed models. 

\textit{Pushdown automata} are finite automata augmented with a stack, a memory which can be used in the last-in-first-out manner. Their \textit{nondeterministic} variants, in which there may be more than one possible move at each step, and the acceptance condition is the existence of at least one computational path that ends in an accept state, recognize exactly the class of context-free languages. Note that a language is context-free if it is generated by a \textit{context-free grammar}, which is a collection consisting of a set of variables and terminals, and a set of production rules. Starting from the start variable, the rules describe how to generate a string of terminals, by replacing each variable with a string of variables and terminals. The set of all generated strings is the language of the grammar.  

The foundations of the theory of computation have been established by Alan Turing, who proposed the \textit{Turing machine} as a model for universal computation \cite{Tu37}. A Turing machine is a finite automaton equipped with an infinite read-write tape which is allowed to move in both directions. The Turing machine is a model for our computers and for the computation process performed by the human mind. A language is recognized by a Turing machine if for every string in the language, the computation on the string ends in the accept state. The class of languages recognized by Turing machines is known as the \textit{recursively enumerable} or \textit{Turing recognizable} languages. 

\textit{Counter machines} are finite automata equipped with additional registers that are initialized to zero at the beginning and can be incremented or decremented based on the current state and the status of the counters (zero or nonzero) throughout the computation. A finite automaton with 2 counters is as powerful as a Turing machine, which led researchers to add various restrictions to the definition. For instance, in a \textit{blind counter automaton}, the counters cannot be checked until the end of the computation and the next move depends only on the current state and the scanned symbol. The class of languages recognized by nondeterministic blind counter automata are incomparable with the class of context-free languages.

Another variant is the \textit{extended finite automaton} (finite automaton
over a group, group automaton, $ G$-automaton), which is a nondeterministic finite automaton equipped
with a register that holds an element from a group \cite{MS97}. The
register is initialized to the identity element of the group, and a
computation is deemed successful if the register is equal to the
identity element at the end of the computation after being multiplied with a group element
at every step. The computational power of a $ G$-automaton is determined by the group $ G $. This setup generalizes various models such as pushdown automata, Turing machines,
nondeterministic blind counter automata and finite
automata with multiplication \cite{ISK76}. When a monoid is used instead of a group, then the model is also called monoid automaton or $ M $-automaton. The same model also appears under the name of valence automata in various papers.

The notion of extended finite automata is also strictly related to that of valence grammar introduced by P\v{a}un in \cite{Pa80}. A \textit{valence grammar} is a formal grammar in which every rule
of the grammar is equipped with an element of a monoid called the valence of the rule. Words generated by the grammar are defined by successful derivations. A successful
derivation is a derivation that starts from the start symbol of the grammar such
that the product of the valences of its productions (taken in the obvious order) is the
identity of the monoid. Valence pushdown automata, which are pushdown automata equipped with a register that is multiplied by elements from a monoid at each step, and context-free valence grammars are discussed in \cite{FS02}.

We have also introduced a new model called vector automaton in \cite{SYS13}. A \textit{vector automaton} is a finite  automaton which is endowed with a
vector, and which can multiply this vector with an appropriate matrix
at each step. One of the entries of this vector can be tested for
equality to a rational number. The machine accepts an input string if
the computation ends in an accept state, and the test for equivalence
succeeds. 

In order to incorporate the notion of the computation being successful if the register returns to
its initial value at the end of the computation as in the case of extended finite automata to this setup, we
propose the new \textit{homing vector automaton} (HVA) model. A homing vector automaton can multiply its vector with an appropriate matrix at each step and can check the entire vector for
equivalence to the initial value of the vector. The acceptance
criterion is ending up in an accept state with the value of the vector
being equal to the initial vector.

\section{Contributions and Overview}
The aim of this thesis is to investigate extended models of finite automata focusing mainly on finite automata over groups and homing vector automata. We examine the classes of languages that
can be recognized by different variants of these models and compare them with the classes of languages recognized by the classical models. We prove separation results based on the different restrictions imposed on the models.

Much of the current literature on extended finite automata pays particular attention to finite automata over free groups and free Abelian groups. This study makes a major contribution to the research on extended finite automata by exploring finite automata over matrix groups for the first time. Most of these results were published in \cite{SDS16,SDS18,SS18}. 

Matrices play an important role in many areas of computation and many important models of probabilistic and quantum computation
\cite{Tu69,LR14} can be viewed in terms of vectors being multiplied by
matrices. The motivation behind analyzing homing vector automata is the matrix
multiplication view of programming, which abstracts the remaining features
of such models away. We investigate homing vector automata under several different regimes, which helps us to determine whether different parameters confer any additional recognition power. Our results on homing vector automata have previously appeared in \cite{SS15, SSD16, SYS19}.

We also present some results on context-free valence grammars and valence pushdown automata. These results are mainly some generalizations of the previously established results for the theory of extended finite automata and appeared in \cite{SDS17}.

The rest of the thesis is structured as follows:

Chapter \ref{chap: back} contains definitions of basic terminology and formal definitions of some of the classical models. It provides a framework for the rest of the thesis. A background on algebra is also presented.

In Chapter \ref{chap: efa}, we investigate the language classes recognized by finite automata over
matrix groups. For the case of $2 \times 2 $ matrices, we prove that
the corresponding finite automata over rational matrix groups are more
powerful than the corresponding finite automata over integer matrix
groups. Finite automata over some special matrix groups, such as the
discrete Heisenberg group and the Baumslag-Solitar group are also
examined.  We also introduce the notion of time complexity for group
automata and demonstrate some separations among related classes. The
case of linear-time bounds is examined in detail throughout our
repertory of matrix group automata. Furthermore, we look at the connection between decision problems for matrix groups and finite automata over matrix groups.

Chapter \ref{chap: hva} defines the homing vector automaton and introduces the various limited versions that we will use to examine the nature of the contribution of different aspects of the definition to the power of the machine. A generalized version of the
Stern-Brocot encoding method, suitable for representing strings on
arbitrary alphabets, is developed. The computational power and properties of deterministic, nondeterministic,
blind, non-blind, real-time and one-way
versions of these machines are examined and compared to
various related types of automata. We establish a connection between one-way nondeterministic version of homing vector automata and extended finite automata. As one-way versions are too powerful even in the case of low dimensions, we pay special attention to real-time homing vector automata. Some closure properties of real-time homing vector automata and their stateless (one state) versions are investigated.

In Chapter \ref{chapter: valence}, we focus on pushdown valence automata and context-free valence grammars. We investigate valence pushdown automata, and prove that they are only as powerful as valence automata. We observe that certain results proven for monoid automata can be easily lifted to the case of context-free valence
grammars. 

Chapter \ref{chap: con} is the conclusion of the thesis. We summarize the results and list some open questions which form a basis for future research.
\chapter{BACKGROUND}\label{chap: back}

\section{Basic Notation and Terminology}
\subsection{Sets} 

Let $ A $ be a set. $ |A| $ denotes the cardinality of $ A $. The power set of $A$ is denoted by $ \mathcal{P}(A) $. A subset $ A \subseteq \mathbb{N}^n  $ is a \textit{linear} set if there exist vectors  $ v_0,v_1,\dots,v_k \in\mathbb{N}^n  $ such that
$$ A=\{v|v=v_0 + \Sigma_{i=1}^k c_iv_i, c_i \in \mathbb{N}\}. $$ 
A \textit{semilinear} set is a finite union of linear sets.

The \textit{Cartesian product of sets} $ A_1,A_2,\dots,A_n $ is the set of all ordered $ n $-tuples $ (a_1,a_2,\dots,a_n) $ where $ a_i \in A_i $ for $ i=1,2,\dots,n $. The Cartesian product is denoted by either $ A_1\times A_2\times \dots \times A_n $ or by $ \prod_{i=1}^{n} A_i$.

\subsection{Strings and Languages} An \textit{alphabet} is a finite set of symbols and usually it is denoted by $ \Sigma $. A \textit{string} (\textit{word}) over $\Sigma  $ is obtained by concatenating zero or more symbols from $ \Sigma $. The string of length zero is called the \textit{empty string} and denoted by $ \varepsilon $. The set $ \Sigma \cup \{\varepsilon\}$ is denoted by $  \Sigma_{\varepsilon} $ in short. We denote by $ \Sigma^* $ the set of all words over $ \Sigma $.

For a string $w \in \Sigma^*  $, $w^r$ denotes its reverse, $ |w| $ denotes its length, $ w[i] $ denotes its $i$'th symbol, $ |w|_{\sigma} $ denotes the number of occurrences of $ \sigma \in \Sigma $ in $ w $. 

A \textit{language} $ L \subseteq \Sigma^* $ is a set of strings over $ \Sigma $. For a given language $L$, its complement is denoted by $\bar{L}$. For a given string $ w$, $ {L}_w $ denotes the singleton language containing only $ w $.    

For a string $ w \in \Sigma^* $ where $ \Sigma=\{\sigma_1,\sigma_2,\dots,\sigma_k\} $ is an ordered alphabet, the \textit{Parikh image of} $ w $ is defined as $$ \phi(w)=(|w|_{\sigma_1}, |w|_{\sigma_2},\dots, |w|_{\sigma_k}). $$ For a language $ L $, its Parikh image is defined as $$ \phi(L) = \{\phi(w)| w \in L\}. $$  A language is called \textit{semilinear} if $ \phi(L) $ is semilinear.

A language ${L}\subseteq \Sigma^*$ is said to be {\em bounded} if there exist words $w_1, \ldots, w_n \in \Sigma^+$ such that
${L} \subseteq w_1 ^* \cdots w_n^*$. A bounded language is said to be (\textit{bounded}) \textit{semilinear} if there exists a semilinear set $A$
of $\mathbb{N}^n$ such that $${L} = \{w_1 ^{a_1} \cdots w_n ^{a_n} : (a_1, \ldots, a_n)\in {A}.\}$$ 

\subsection{Vectors and Matrices} For a given row vector $ v $, $ v[i] $ denotes its $i$'th entry. Let $ A_{k \times l} $ be a $ k \times l $ dimensional matrix. $ A[i,j] $ denotes the entry in the $ i $'th row and $ j $'th column of $ A $. 

The \textit{identity matrix} $ I $ of size $ n $ is the $ n\times n  $ matrix with ones on the main diagonal and zeros elsewhere.

\section{Automata and Computation}
In this section, we are going to talk about basic notions of computation. We refer to the finite automaton model which may be extended with a storage mechanism as discussed in Section \ref{sec: int-aut} and Section \ref{sec: int-efa}. 

For a machine $ \mathcal{M} $, there exist a finite set of \textit{states} $Q = \{ q_1,\ldots,q_n \}$ where $ q_1 $ is the \textit{initial state} unless otherwise specified and a set of \textit{accept state(s)} $Q_a \subseteq Q$. An input string $ w \in \Sigma^* $ is written on a one-way infinite tape starting from the leftmost tape square. The transition function denoted by $ \delta $ describes the next move of $ \mathcal{M} $ upon reading some input symbol.  Initially, the tape-head is placed on the leftmost tape square. Starting from the initial state, the sequence of transitions performed by $ \mathcal{M} $ on any input string is called a \textit{computation}. 

The acceptance criteria of an input string depend on the type of the machine which we will discuss in detail in Section \ref{sec: back-c}. A computation is called \textit{accepting} if the computation results in the acceptance of the string, and \textit{rejecting} otherwise. The set of all strings accepted by $ \mathcal{M} $ is called the \textit{language recognized by} $ \mathcal{M} $ and denoted by $ L(\mathcal{M}) $.

\subsection{Deterministic Computation} A computation is \textit{deterministic} if there is only one possible move at each step. When an input string is read, there is only a single computation. 

\subsection{Nondeterministic Computation} In a \textit{nondeterministic} computation, there may be more than one possible move at each step. When an input string is read, the computation looks like a tree since there may be more than one computation path.

\begin{figure}[h]
	\centering
	\includegraphics[width=1.1\linewidth]{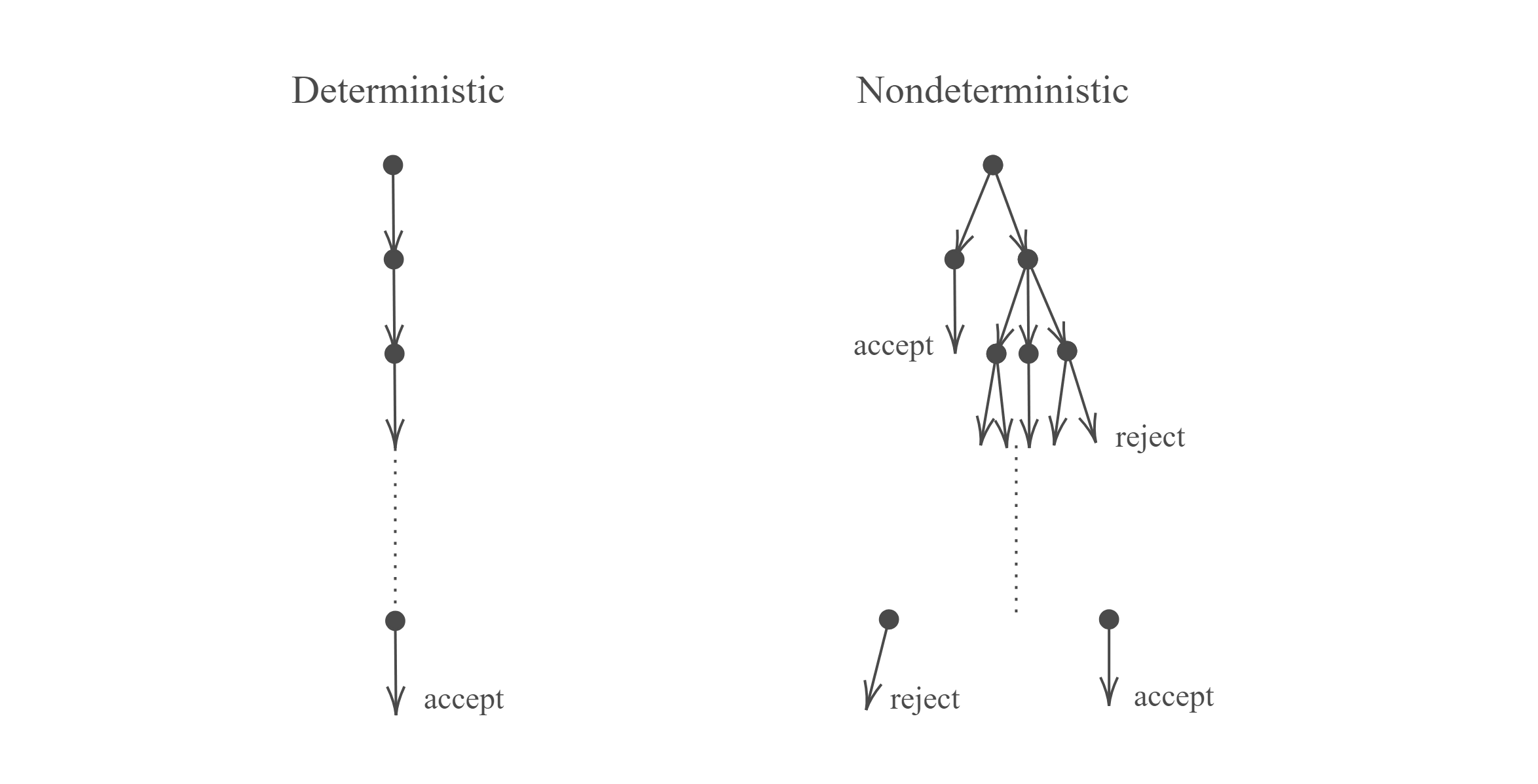}
	\caption{Deterministic and nondeterministic computation}
	\label{fig:dnd}
\end{figure}

\subsection{Blind Computation} For the machine types which have additional storage mechanisms like registers or counters, computation is called \textit{blind} if the status of the storage mechanism cannot be checked until the end of the computation. The next move of the machine is not affected by the current status of the storage mechanism. When the computation ends, the status of the storage mechanism is checked and it determines whether the input string will be accepted or not.

\subsection{Empty String Transitions} When a machine makes an \textit{empty string} ($ \varepsilon $) transition, it moves without consuming any input symbol. In deterministic machines, $ \varepsilon $ moves should be defined carefully as they may lead to nondeterminism. 

\subsection{Real-time, One-way and Two-way Computation} A computation is called \textit{real-time} if the tape-head moves right at each step. A computation is called \textit{one-way} if the tape-head is allowed to stay on the input tape while moving from left to right. This can be accomplished by adding an additional direction component to the transition function which dictates the movement of the tape-head. A computation is called two-way if the tape-head can move both left and right. Tape head directions will be specified by a subset of the set $\mathcal{D}= \{\leftarrow, \downarrow, \rightarrow\} $ where $ \leftarrow, \downarrow $, and $ \rightarrow $ stand for left, stay and right respectively.

Note that a machine making $ \varepsilon $-transitions does not operate in real-time. A one-way nondeterministic computation may be also defined without specifying the head directions but allowing $ \varepsilon $-moves instead. In that case, it is assumed that the machine moves right as long as it consumes an input symbol. 
 
\subsection{End-marker} In some models, the input string is written on the tape in the form $ w\dollar $ and the machine is allowed to make transition(s) after finishing reading the input string, upon scanning the end-marker $ \dollar $, which we call postprocessing. Postprocessing may add additional power depending on the model.

\section{Classical Models}\label{sec: back-c}

\subsection{Finite Automaton} \label{sec: back-c-f}

A \textit{finite automaton} (FA) is a 5-tuple $$\mathcal{F}= (Q,\Sigma,\delta,q_1,Q_a), $$ where $ Q $ is the set of \textit{states}, $ \Sigma $ is the \textit{input alphabet}, $ \delta $ is the \textit{transition function}, $ q_1 \in Q $ is the \textit{initial state} and $ Q_a \subseteq Q  $ is the set of \textit{accept states}. The transitions of $ \mathcal{F} $ depend only on the current state and the input symbol. 

Formally, the transition function of a one-way deterministic finite automaton (1DFA) is defined as follows:
\[ \delta: Q \times \Sigma \rightarrow Q \times \mathcal{D},\]
where $ \cal{D}=\{\downarrow,\rightarrow\} $ is the set representing the possible moves of the tape-head, $ \downarrow $ denoting stay and $ \rightarrow $ denoting right. 

$\delta(q,\sigma)=(q',d)$ means that $ \mathcal{F} $ moves to state $ q' \in Q $ moving its tape-head in direction $ d \in \mathcal{D }$ upon reading $ \sigma \in \Sigma $ in state $ q \in Q$. We assume that the last move of the machine always moves the tape-head right.

In a \textit{real-time deterministic finite automaton} (DFA), the tape-head moves right at every step and the direction component for the tape-head is omitted from the transition function:
\[ \delta: Q \times \Sigma \rightarrow Q . \]

Let us define the nondeterministic variants of finite automata. The transition function of a \textit{one-way nondeterministic finite automaton} (1NFA) is defined as
\[ \delta: Q \times \Sigma_{\varepsilon} \rightarrow \mathcal{P}(Q), \]
so that there may be more than one possible move at each step and the machine is allowed to make $ \varepsilon $-transitions. 

A \textit{real-time nondeterministic finite automaton} (NFA) is not allowed to perform any $ \varepsilon $-transitions. The transition function of an NFA is defined as follows:
\[ \delta: Q \times \Sigma \rightarrow \mathcal{P}(Q). \]

An input string $ w $ of length $ n $ is accepted by a finite automaton if there is a computation in which the machine enters an accept state with the tape-head on the $ n+1 $'st tape square.

1NFAs, 1DFAs, NFAs and DFAs recognize the same class of languages known as the class of \textit{regular} languages, abbreviated by $ \mathsf{REG} $.

\subsection{Pushdown Automaton}
A \textit{pushdown automaton} (PDA) is a finite automaton equipped with a stack. Stacks are one-way infinite storage mechanisms working in last-in-first-out fashion. The operations applied on a stack are called the \textit{pop} and \textit{push} operations, which stand for removing the topmost symbol from the stack and adding a new symbol on top of the stack by pushing down the other symbols in the stack, respectively. Note that the stack alphabet may be different than the input alphabet. At each step of the computation, the machine may pop the topmost symbol from the stack, move to a new state, and push a new symbol onto the stack, depending on the current state and the input symbol.

Formally, \textit{a one-way nondeterministic pushdown automaton} (1NPDA) is a 6-tuple $$ \mathcal{A}= (Q, \Sigma, \Gamma, \delta, q_1, Q_a), $$ where $ \Gamma $ is the stack alphabet. The transition function of a 1NPDA is defined as 
\[ 
Q \times  \Sigma_{\varepsilon} \times \Gamma_{\varepsilon} \rightarrow \mathcal{P}(Q \times \Gamma_{\varepsilon}),
\]
where $  \Gamma_{\varepsilon} = \Gamma \cup \{\varepsilon\}$. $ (q_2,\gamma_2) \in \delta(q_1,\sigma,\gamma_1) $ means when in state $ q_1 $ reading $ \sigma \in \Sigma_{\varepsilon}$, $ \mathcal{A} $ pops $ \gamma_1 \in \Gamma_{\varepsilon} $ from the stack, moves to state $ q_2 $ and pushes $\gamma_2 \in \Gamma  $ onto the stack. Note that $  \gamma_1 $ and $ \gamma_2 $ can be $ \varepsilon $ in which case nothing is popped from or pushed onto the stack. If $ \gamma_1 $ is not on top of the stack, then the transition cannot take place.

A string of length $ n $ is accepted by a 1NPDA if there exists a computation in which the machine enters an accept state with the tape-head on the $n+1 $'st tape square and the stack is empty. There are also alternative definitions for acceptance which do not require an empty stack. By using $ \varepsilon $-transitions, it is easy to see that the stack may be emptied in an accept state to satisfy the additional empty stack requirement and the two definitions correspond to the same class of languages, as long as the computation is one-way and $ \varepsilon $-moves are possible.  

1NPDAs recognize the class of context-free ($ \mathsf{CF} $) languages.

\subsection{Turing Machine}

A \textit{Turing machine} is a finite state automaton with the following properties:
\begin{itemize}
	\item The tape-head can move in both directions.
	\item The tape-head can read from the tape and modify the tape content by writing on the tape.
\end{itemize} 

At the beginning of the computation, the input is written on the tape starting from the first tape square and the rest of the tape contains the special blank symbol. If the tape-head tries to move left on the leftmost square, its position does not change.

Formally, a (two-way deterministic) \textit{Turing machine} (TM) is a 7-tuple
\[ \mathcal{T}=(Q, \Sigma, \Gamma, \delta, q_1, q_a, q_r), \]
where $ \Sigma $ is the input alphabet not containing the blank symbol $ \sqcup $, $ \Gamma  $ is the tape alphabet where $\sqcup  \in \Gamma$ and $ \Sigma \subseteq \Gamma $, $ q_a \in Q $ and $ q_r \in Q$ are the accept and reject states respectively. The transition function of a TM is defined as follows:
\[ \delta: Q \times \Gamma \rightarrow Q \times \Gamma \times \cal{D}  \]
where $\cal{D}= \{\leftarrow,\rightarrow\} $, meaning that $ \cal{T} $ moves to state $ q' \in Q $, updating the tape square under the tape-head by $ \gamma_2 \in \Gamma $, moving the tape-head in the direction $ d \in \cal{D} $, when in state $ q $ and reading $ \gamma_1 \in \Gamma $ from the tape, specified by the transition $ \delta(q,\gamma_1)=(q',\gamma_2,d) $.  

 While processing an input string, the computation may end at any point (before finishing reading the input string) in the designated accept state resulting in the acceptance of the input string, in the designated reject state resulting in the rejection of the input string or the computation may go on forever without ever entering the accept state or the reject state.
 
Turing machines recognize the class of \textit{recursively enumerable} languages ($\mathsf{RE}  $).

\subsection{Counter Automaton}
A \textit{counter automaton} (CA) is a finite automaton equipped with one or more counters, a storage mechanism holding an integer which can be incremented, decremented, and checked for equivalence to zero. Formally a CA is a 6-tuple 
$$ \mathcal{C}= ( Q,\Sigma,\delta,q_1,Q_a) $$  
and a CA with $ k $ counters is abbreviated as $k$CA. At the beginning of the computation, counters are initialized to 0. At each step of the computation, depending on the current state and the status of the counters, $ \mathcal{C} $ moves to another state and updates its counters. . 

The transition function of a \textit{one-way deterministic $ k $-counter automaton} (1D$ k $CA) is defined as follows:
 \[ \delta:Q \times \Sigma \times \Theta \rightarrow Q \times \{-1,0,1\}^k \times \mathcal{D} \]
where $ \Theta = \{=,\neq\}^k $ and $ \mathcal{D}=\{\downarrow,\rightarrow\} $. A transition of the form $ \delta(q,\sigma,\theta) = (q',c,d)$ means that upon reading $ \sigma \in \Sigma $ in state $ q \in Q $, $ \cal{C} $ moves to $ q' $ updating its counters by $ c  \in  \{-1,0,1\}^k $ and updates the tape-head with respect to $ d \in \mathcal{D} $, given that the status of the counters is $\theta  \in \{=,\neq\}^k $, where $ = $ and $ \neq  $ denote whether the corresponding counter values equal zero or not, respectively.

By restricting  1D$ k$CAs so that the status of the counters cannot be checked until the end of the computation, we obtain \textit{one-way deterministic blind $ k $-counter automaton} (1D$ k $BCA). The transition function of a 1D$ k $BCA is formally defined as follows:
\[  
  \delta:Q \times \Sigma \rightarrow Q \times \{-1,0,1\}^k \times \mathcal{D}. 
\]
In a blind counter automaton, the next move of the machine does not depend on the status of the counters.

For the one-way deterministic machines, we assume that the last move of the machine always moves the tape-head to the right.

The real-time versions, \textit{real-time deterministic $ k $-counter automaton} (D$ k $CA) and \textit{real-time deterministic blind $ k $-counter automaton} (D$ k $BCA) are defined analogously, by omitting the direction component from the transition functions. The ranges of the transition functions take the following form: $  Q \times \{-1,0,1\}^k  $.

Let us define the nondeterministic variants of counter automata. The transition function of a \textit{one-way nondeterministic $ k $-counter automaton} (1N$ k $CA) is defined as follows:
\[  
\delta:Q \times \Sigma_{\varepsilon} \times \Theta \rightarrow \mathcal{P}(Q \times \{-1,0,1\}^k). 
\]
A \textit{one-way nondeterministic blind $ k $-counter automaton} (1N$ k $BCA) is a restricted 1N$ k $CA which cannot check the value of the counters until the end of the computation. Transition function of a 1N$ k $BCA is defined as
\[  
\delta:Q \times \Sigma_{\varepsilon} \rightarrow \mathcal{P}(Q \times \{-1,0,1\}^k ). 
\]
The real-time versions, \textit{real-time nondeterministic $ k $-counter automaton} (N$ k $CA) and  \textit{real-time nondeterministic blind $ k $-counter automaton}  (N$ k $BCA) are defined analogously by not allowing $ \varepsilon $-moves. The domains of the transition functions of these models are replaced with $ Q \times \Sigma \times \Theta $ and $ Q \times \Sigma $ respectively.

An input string $ w $ of length $ n $ is accepted by a $ k $CA if there exists a computation in which the machine enters an accept state with the tape-head on the $ n+1 $'st square. An input string is accepted by a $ k $BCA with the further requirement that all of the counters- values are equal to 0 .

The abbreviations used for counter automata variants discussed so far are given in Table \ref{table: counter}.

\begin{table}[!h] 
	\centering
	\vskip\baselineskip 
	\caption{The abbreviations for CA variants.}
	\label{table: counter}
	\vspace{0.1in}
	\begin{tabular}{|p{0.3\textwidth}|p{0.15\textwidth}|p{0.15\textwidth}|}
		\hline 
		& Real-time & One-way \\
		\hline 
		Deterministic & D$  k$CA & 1D$  k$CA \\
		\hline 
		Deterministic blind & D$  k$BCA & 1D$  k$BCA \\
		\hline 
		Nondeterministic & N$  k$CA & 1N$  k$CA \\
		\hline 
		Nondeterministic blind & N$  k$BCA & 1N$  k$BCA \\
		\hline
	\end{tabular}
	
\end{table}

A 1D2CA can simulate a Turing Machine \cite{Mi60} and therefore $ \bigcup_k \mathfrak{L}(\textup{1D$ k $CA}) = \mathsf{RE}$. As two counters are enough to recognize any recursively enumerable language, a considerable amount of literature has been published on language recognition power of counter automata under various restrictions. Preliminary work on counter automata was undertaken by Fischer \textit{et al.}, who investigated real-time deterministic counter automata \cite{FMR67}, and one-way deterministic counter automata \cite{FMR68}. A hierarchy based on the number of the counters for real-time deterministic counter automata is  demonstrated in \cite{FMR67}. In \cite{FMR68}, the state set is separated into polling states from which the machine moves to another state and updates the counters by consuming an input symbol and autonomous states which allow machine to update the counters and change the state without reading anything. It is easy to show that both definitions are equivalent and correspond to machines with the same language recognition power. 

A remarkable result from \cite{FMR68} states the following:

\begin{fact}\label{fact: counter}\textup{\cite{FMR68}}
	Given any $ k $CA  with the ability to alter the contents of each counter independently by any integer between $ +c  $ and $ -c $ in a single step (for some fixed integer $ c $), one can effectively find a time-equivalent (ordinary) $ k $CA. 
\end{fact}
 
The proof involves using some additional states to simulate counter updates from the set $ \{-c,-c+1,\dots,c-1,c\}^k $ with an ordinary counter automaton, without increasing the time complexity and the number of the counters. Hence, we can assume that any $ k $CA can be updated by arbitrary integers at each step.

One-way deterministic and nondeterministic counter automata working under time restrictions formed the central focus of study of Greibach in \cite{Gr76}, in which the author proved various separation results between deterministic and nondeterministic models. Let us note that in \cite{Gr76} and \cite{Pe11} where one-way deterministic counter automata are investigated, it is assumed that the counter automata can process the end-marker.

In another major study by Greibach, one-way nondeterministic blind counter automata are examined \cite{Gr78} and some other restricted versions like partially blind counters and reversal bounded counters are introduced as well.

\subsection{Finite Automata with Multiplication}

A \textit{finite automaton with multiplication} (FAM) \cite{ISK76} is 6-tuple
\[\mathcal{W}=(Q,\Sigma,\delta,q_1,Q_a,\Lambda), \] where the additional component $\Lambda$ is a finite set of rational numbers
(multipliers). A FAM is a finite automaton equipped with a register holding a positive rational number. The register is initialized to 1 at the beginning of the computation and multiplied with a positive rational number at each step, based on the current state, input symbol and the status of the register determined by whether the register is equal to 1 or not. The input string of a FAM is given in the form $ w\dollar $ and FAMs are allowed to perform post-processing. 

The original definition of FAMs is given for \textit{one-way} machines where the tape-head is allowed to stay on the same input symbol for more than one step. 

A \textit{one-way deterministic finite automaton with multiplication} (1DFAM) is defined by the transition function
$$ \delta: Q \times \Sigma_{\dollar} \times \Omega
\rightarrow  Q\times \mathcal{D} \times \Lambda, $$ where $ \Sigma_{\dollar}= \Sigma \cup \{\dollar\} $, $\Omega$ is the set $\{=,\neq\}$ denoting the whether the register is equal to 1 or not respectively and $ \cal{D}= \{\downarrow, \rightarrow\} $ is the set of head directions. It is assumed that $ \delta(q,\dollar,\omega) = \emptyset $ for all $ q \in Q_a $ so that the computation ends once the end-marker $ \dollar $ is scanned in an accept state. Reading symbol $\sigma \in \Sigma_{\dollar}$ in state $q \in Q$, $\mathcal{W}$ compares the current value of the register with 1, thereby calculating
the corresponding value $\omega \in \Omega$, and switches its state to $q' \in Q$, moves its head in direction
$d \in \cal{D}$, and multiplies the register by $\lambda \in \Lambda$, in accordance with the transition function value 
$\delta(q,\sigma,\omega) =(q',d,\lambda)$.

A \textit{one-way deterministic finite automaton with multiplication without equality} (1DFAMW) is a model obtained by restricting 1DFAM so that the register can be checked only at the end of the computation. The transition function of a 1DFAMW is defined as follows: 
$$ \delta: Q \times \Sigma_{\dollar} \rightarrow  Q\times \mathcal{D} \times \Lambda, $$
where the next move of the machine does not depend on the current status of the register. The 1DFAMW can be seen as the blind version of the 1DFAM model.

We define the real-time versions, \textit{real-time deterministic finite automaton with multiplication} (DFAM), and \textit{real-time deterministic finite automaton with multiplication without equality} (DFAMW), by removing the direction component from the transition functions and assuming that the tape-head moves right at each step. The ranges of the transition functions are updated with $ Q \times \Lambda $.

A \textit{one-way nondeterministic finite automaton with multiplication} (1NFAM) is a model that extends the 1DFAM with the ability to make nondeterministic moves. The transition function of a 1NFAM is defined as 
$$ Q \times \Sigma_{\dollar} \times \Omega \rightarrow \mathcal{P}( Q\times \mathcal{D} \times \Lambda). $$ 

 A \textit{one-way nondeterministic finite automaton with multiplication without equality} (1NFAMW) is the blind version of the 1NFAM model which cannot check whether or not the register has value 1 during computation. Transition function of a 1NFAMW is defined as follows: 
$$ \delta: Q \times \Sigma_{\dollar} \rightarrow \mathcal{P}( Q\times \mathcal{D} \times \Lambda), $$
so that the next move of the machine does not depend on the current status of the register.

We also define real-time versions \textit{real-time nondeterministic finite automaton with multiplication} (NFAM) and \textit{real-time nondeterministic finite automaton with multiplication without equality} (NFAMW) by removing the direction component from the transition functions and assuming that the tape-head moves right at each step. In that case, the ranges of the transition functions are updated with  
$ \mathcal{P}( Q\times \Lambda) $.

For an input string $ w $ of length $ n $, $ w $ is accepted by a FAM or a FAMW if there exists a computation in which the machine enters an accept state with the input head on the end-marker $ \dollar $ and the register is equal to 1.

The abbreviations used for finite automata with multiplication variants discussed so far are given in Table \ref{table: fam}.

\begin{table}[!h] 
	\centering
	\vskip\baselineskip 
	\caption{The abbreviations for FAM variants.}
	\label{table: fam}
	\vspace{0.1in}
	\begin{tabular}{|p{0.3\textwidth}|p{0.15\textwidth}|p{0.15\textwidth}|}
		\hline 
		& Real-time & One-way \\
		\hline 
		Deterministic & DFAM & 1DFAM \\
		\hline 
		Deterministic blind & DFAMW & 1DFAMW \\
		\hline 
		Nondeterministic & NFAM & 1NFAMW \\
		\hline 
		Nondeterministic blind & NFAMW & 1NFAMW \\
		\hline
	\end{tabular}
	
\end{table}

The following characterization of the class of
languages recognized by 1NFAMWs for the case where the alphabet is unary is given in \cite{ISK76}.

\begin{fact}\textup{\cite{ISK76}}\label{fact: unary}
	All \textup{1NFAMW}-recognizable languages over a unary alphabet are regular.
\end{fact}
 
Furthermore, bounded languages recognized by 1NFAMWs have also been examined.

\begin{fact}\textup{\cite{ISK76}}\label{fact: bounded}
	A bounded language is recognized by a \textup{1NFAMW} iff it is semilinear.
\end{fact}

1NFAMWs are defined by the tape-head directions and they process the end-marker. In the next lemma, we show that modifying the definition of 1NFAMW slightly does not change its recognition power.

\begin{lem}\label{lem: 1nfamw}
Let $ \cal W $ be a \textup{1NFAMW}. There exists a \textup{1NFAMW} $ \cal{W} '$ which does not process the end-marker and defined using $ \varepsilon $-transitions that recognizes the same language as $ \cal W $.
\end{lem}
\begin{proof}
Given a 1NFAMW $ \mathcal{W}_1=(Q,\Sigma,\delta,q_1,Q_a,\Lambda) $, we construct $ \mathcal{W}_2 $ from $ \mathcal{W}_1 $ by first removing the transitions which are traversed upon reading $ \sigma \in \Sigma $ and which do not move the tape-head, by using some additional states and $ \varepsilon $-transitions as follows: Let $ Q_s $ be the set of state-symbol pairs of $ \mathcal{W}_1 $ such that $ (q,\sigma) \in Q_s $ if there is no incoming transition to $ q $ that does not move the tape-head and $ \delta(q,\sigma) =(q',\lambda,\downarrow)$ for some state $ q' \in Q$ and $ \lambda \in \Lambda $. For each $ (q,\sigma) $ pair, let $ G_{q,\sigma} $ be the graph obtained from the state transition diagram of $ \mathcal{W}_1 $ by removing all transitions except the ones of the form $ \delta(q,\sigma) =(q',\lambda,\downarrow)$. Let $ r_{q,\sigma} $ be the subgraph of $ G_{q,\sigma} $ induced by the set of reachable vertices from $ q $ in $G_{q,\sigma} $. We create a copy of each $ r_{q,\sigma} $ which we denote by  $ r_{q,\sigma}^c$ and connect it to $ \mathcal{W}_2 $ as follows: From the state $ q$ in the original copy, we add an $ \varepsilon $-transition to the state $ q $ in $r_{q,\sigma}^c $. In $ r_{q,\sigma} $, there should be some states $ t $ satisfying $ \delta(t,\sigma)= (u,\lambda,\rightarrow) $ for some $u \in Q $ and $ \lambda \in \Lambda $, since otherwise the rest of the string cannot be scanned. We remove those transitions from $ \mathcal{W}_2 $ and add a $ \sigma $-transition to each original $ u $ from the copy of the state $ t $ in $ r_{q,\sigma}^c $. The inherited transitions from $ \mathcal{W}_1 $ which do not move the tape-head are removed from $ \mathcal{W}_2 $. 

Next we remove the $ \dollar $-transitions from $ \mathcal{W}_2$. Let $ Q_{\dollar} $ be the set of states of $\mathcal{W}_1$ that don't have any incoming $ \dollar $-transition and have an outgoing $ \dollar  $-transition. After finishing reading the string, $ \mathcal{W}_1 $ should enter a state from $ Q_{\dollar} $, read the $\dollar  $ symbol and possibly make some transitions without changing the tape-head and eventually end in an accept state, to accept any string. Let $ G_{\dollar} $ be the graph obtained from the transition diagram of $ \cal W$, by removing all transitions except the $\dollar  $-transitions. Let $ r_q $ be the subgraph of $G_{\dollar}$, induced by the set of reachable vertices from $q $ in $ G_{\dollar} $, for each $q \in Q_{\dollar}  $. We create a copy $ r_q^c $ of each subgraph $r_q  $, replace the $ \dollar $ symbols with $ \varepsilon $ and connect it to $ \mathcal{W}_2$: For each incoming transition to $ q $ in $ \mathcal{W}_1 $, we create a copy of the transition and connect it to the copy of $ q $ in $ r_q^c $. The $ \dollar $-transitions inherited from $ \mathcal{W}_1 $ are removed from $\mathcal{W}_2  $ and any accept state of $ \mathcal{W}_1 $ is no longer an accept state in $ \mathcal{W}_2 $. $ \mathcal{W}_2 $ simulates the computation of $ \mathcal{W}_1 $ on any non-empty string until scanning the $ \dollar $ and then follows the transitions in the newly added states to reach an accept state. 

We can safely remove any remaining tape-head directions which move the tape-head to the right from $ \mathcal{W}_2 $ and we obtain a 1NFAMW without the tape-head directions and that does not process the end-marker recognizing the same language as $ \mathcal{W}_1 $. 	
\end{proof}
From now on, we may assume that a 1NFAMW is defined without the tape-head directions and does not process the end-marker.

\section{Background on Algebra}

In this section we provide definitions for some basic notions from algebra and group theory. See \cite{Fr03,LS77} for further references.

\subsection{Algebraic Structures}
Let $ A  $ be a set. A \textit{binary operation} $ * $ on a set $ A $ is a function from $ A \times A $ to $ A $. Let $ B $ be a subset of $ A $. The subset $ B $ is \textit{closed under} $ * $ if for all $ a,b \in B $, we also have $ a*b \in B $.

A binary operation $ * $ is called 
\begin{itemize}
	\item \textit{associative} if for all $ a,b,c \in A $, we have $ (a*b)*c =a * (b*c)$,
	\item \textit{commutative} if $ a*b=b*a $ for all $ a,b\in A $. 
\end{itemize}

A set $ A $, together with a binary operation $ * $ is called an \textit{algebraic structure} $(A,*)  $. Let $ (A,*) $ and $ (A',*') $ be binary algebraic structures. An \textit{isomorphism} of $ A $ with $ A' $ is a one-to-one function $ \phi  $ mapping $ A $ onto $ A' $ such that $ \phi(a * b)=\phi(a) *'\phi(b) $ for all $ a,b\in A $. If such a map exists, then $ A $ and $ A' $ are \textit{isomorphic} binary structures which is denoted by $ A \simeq A' $.

For an algebraic structure $ (A,*) $,
\begin{itemize}
	\item an element $ e \in A $ is called the \textit{identity element} if for all $ a \in A $ $ e * a = a * e = a $,
	\item element $ a \in A $ has an \textit{inverse} if there is an element $ a' $ in $ A $ such that $ a * a' = a' * a = e $. 
\end{itemize}

\subsection{Groups, Monoids, Semigroups}

The following are fundamental algebraic structures:

\begin{itemize}
	\item A \textit{semigroup} $ (S,*) $ is an algebraic structure with an associative binary operation.
	\item A \textit{monoid} $ (M,*) $ is a semigroup with an identity element $ e$.
	\item A \textit{group} $ (G,*) $ is a monoid where for every element $ a \in G $, there is a unique inverse of $ a \in G $ denoted by $ a^{-1} $.
\end{itemize}

A group may be referred simply as $ G $ instead of $ (G,*) $ and most of the time the operation sign is omitted and $ a*b $ is simply denoted by $ ab $. The \textit{order} $ |G| $ of $ G $ is the number of elements in $ G $. The \textit{order of an element} $ g $ of a group $ G $ is the smallest positive integer $ m $ such that $ g^m=e $. A group is \textit{Abelian} if its binary operation is commutative. A monoid or a semigroup is called \textit{commutative}, if its binary operation is commutative.

A monoid is called \textit{inverse} if for every $ x \in M $, there exists a unique $ y\in M $ such that $ x =xyx $ and $ y=yxy $. Note that it is not necessary that $ xy $ is equal to the identity of $ M $.

A subset $ H $ of a group $ G $ is called a \textit{subgroup} of $ G $ if
\begin{enumerate}
	\item $ H $ is closed under the binary operation of $ G $,
	\item The identity element $ e $ of $ G $ is in $ H $,
	\item For all $ a \in H $, $ a^{-1} \in H $.
\end{enumerate}
 A subset $ H $ of a monoid $ M $ is a \textit{submonoid} if (i) and (ii) hold and a subset $ H $ of a semigroup $ S $ is a \textit{subsemigroup} if (i) holds. 

  Let $ A \subseteq G $. The \textit{subgroup generated} by $ A $, denoted by $ \langle A \rangle $, is the subgroup of $ G  $ whose elements can be expressed as the finite product of elements from $ A $ and their inverses. If this subgroup is all of $ G $, then $A $ \textit{generates} $ G $ and the elements of $ A $ are called the \textit{generators} of $ G $. If there is a finite set that generates $ G $, then $ G $ is \textit{finitely generated}. The smallest cardinality of a generating set for $ G $ is the \textit{rank} of the group $ G $. The notion of generating sets also applies to monoids and semigroups.

Let $ H $ be a subgroup of a group $ G $. The subset $ aH=\{ah | h \in H\} $ of $ G $ is the \textit{left coset} of $ H $ containing $ a $, and the subset $ Ha = \{ha | h \in H\}$ is the \textit{right coset} of $ H $ containing $ a $. The number of left cosets of $ H $ in $ G $ is the \textit{index} of $ H $ in $ G $.

\newpage
Let $ G_1,G_2,\dots,G_n $ be groups. For $ (a_1,a_2,\dots,a_n) $ and $ (b_1,b_2,\dots,b_n) $ in $ \prod_{i=1}^{n} G_i = G_1 \times G_2 \times \dots \times G_n$, define  $ (a_1,a_2,\dots,a_n)(b_1,b_2,\dots,b_n) $ to be the element $ (a_1b_1,a_2b_2,\dots,$  $a_nb_n) $. Then $ \prod_{i=1}^{n} G_i$ is the \textit{direct product of the groups} $ G_i $ under this binary operation. Direct product of monoids and semigroups can be defined similarly.

\subsection{Groups of Integers, Vectors, Rational Numbers}

The set of integers together with the binary operation addition forms a group denoted by $ (\mathbb{Z},+) $. It can be generated by the set $ \{1\} $ and its identity element is 0. 

The set of $ k $-dimensional integer vectors for some $ k \geq 2 $ under addition also forms a group denoted by $ (\mathbb{Z}^k,+) $ and it is finitely generated by $ k $ vectors, $ i $'th vector having 1 in its $ i $'th entry and 0 in the remaining entries for $i=1\dots k  $.

The set of positive rationals with the binary operation multiplication forms an infinitely generated group denoted by $ (\mathbb{Q}^+,\cdot) $, with the identity element 1.

Note that all groups introduced above are Abelian.

\subsection{Matrix Groups and Monoids}
The set of all $ n \times n $ matrices with integer entries with the operation of matrix multiplication forms a monoid, which is denoted by $ M_n(\mathbb{Z}) $. 

We denote by $GL(n,\mathbb{Z})$ \textit{the general linear group of degree} $ n $
over the field of integers, that is, the
group of $ n\times n $ invertible matrices with integer entries. Note
that these matrices have determinant $\pm 1$. Restricting the matrices
in $GL(n,\mathbb{Z})$ to those that have determinant 1, we obtain the
\textit{special linear group of degree} $ n $ over the field of integers,
$SL(n,\mathbb{Z})$.

Analogously, $GL(n,\mathbb{Q})$ is the group of $ n\times n $ invertible matrices with rational entries and $SL(n,\mathbb{Q})$ is the group of $ n \times n $ invertible matrices with rational entries with determinant 1.

\subsection{Free Monoids and Free Groups}

Let $ A=\{a_1,a_2,\dots,a_n\} $ be a finite set. We think of $ A $ as an alphabet and its elements $ a_i $ as the letters of the alphabet. A \textit{word} is a concatenation of finite elements of $ A $.

The set of all words over $ A $ is denoted by $ A^* $. $ A^*$ together with the binary operation concatenation is called the \textit{free monoid over} $ A $. The empty word $ \varepsilon $, which is obtained by concatenation of zero elements, is the identity element of the free monoid. The set of all nonempty words over $ A $ is called the \textit{free semigroup over} $ A $ and often denoted by $ A^+ $. 

Now assume that for every $ a_i \in A $, there is a corresponding inverse symbol $ {a_i}^{-1} $ and let $ A^{-1}=\{{a_1}^{-1}, {a_2}^{-1}, \dots,{a_n}^{-1}\} $. Consider the set of all words over $ X=A \cup A^{-1} $. We can simplify a word by removing occurrences of $ a_i{a_i}^{-1} $, for each $ i $. A word is called \textit{reduced} if it cannot be further simplified. The set of all reduced words over $ X $
is called the \textit{free group over} $ A $.  The number of elements in $ A $ is called the \textit{rank of the free group}. An arbitrary group $ G $ is called \textit{free}, if it is isomorphic to the free group generated by a subset $ S $ of $ G $. Informally, a group is free if no relation holds among the generators of the group. Two free groups are isomorphic iff they have the same rank.

We will denote the \textit{free group} of rank $r$ by $ \mathbf{F}_r $. Note that $ \mathbf{F}_0 $ is the trivial group and $ \mathbf{F}_1 $ is Abelian and isomorphic to $ (\mathbb{Z},+) $. 2 is the smallest rank of a non-Abelian free group.

The well known Nielsen-Schreier Theorem for free groups states the following. 

\begin{fact}\label{fact: nielsen} \textup{\cite{Sc27,Ni21}}
Every subgroup of a free group is free.
\end{fact}

Furthermore, $ \mathbf{F}_2 $ contains free subgroups of every finite rank.

\subsection{Free Abelian Groups}

Let $ A $ be a subset of a nonzero Abelian group $ G $ and suppose that each nonzero element $ g  $ in $ G $ can be expressed uniquely in the form $$  g=n_1a_1 + n_2a_2+ \dots + n_ra_r $$ for $ n_i \neq 0 $ in $ \mathbb{Z} $ and distinct $ a_i \in A $. $ G $ is called a \textit{free Abelian group} of rank $ n $ and $ A $ is called a \textit{basis} for the group. Any two free Abelian groups with the same basis are isomorphic.

A free Abelian group of rank $ r $ is isomorphic to $ \mathbb{Z}\times \mathbb{Z}\times \dots \times \mathbb{Z} = \mathbb{Z}^r$. Hence, the group of integers $ \mathbb{Z} $ and integer vectors $ \mathbb{Z}^n $ are finitely generated free Abelian groups. The group of positive rational numbers $ \mathbb{Q}^+ $ is the free Abelian group of infinite rank.

\subsection{Word Problem}
 
For any finitely generated group $ G $ with the set of generators $ A $, we have a homomorphism $ \phi: X^* \rightarrow G $ where  $ X =\{A \cup A^{-1}\}  $. Given a group $ G $ generated by the set $ A $, the \textit{word problem for group $ G $} is the problem of deciding whether $ \phi(w)=1 $ for a given word $ w \in X^*$, where $ 1 $ is the identity element of $ G $. The \textit{word problem language} of $ G $ is the language $ W(G,A) $ over $ X $ which consists of all words that represent the identity element of $ G $, that is $ W(G,A)=\{w \in X^* | \phi(w)=1\} $. Most of the time, the statements about the word problem are independent of the generating set and the word problem language is denoted by $ W(G) $.

\chapter{EXTENDED FINITE AUTOMATA} \label{chap: efa}

In this chapter, we investigate extended finite automata over matrix groups. The theory of extended finite automata has been essentially developed in the
case of free groups and in the case of free
Abelian groups, where strong theorems allow the
characterization of the power of such models and the combinatorial properties
of the languages recognized by these automata. For groups that are not of the types mentioned above, even in the case
of groups of matrices of low dimension, the study of group automata quickly becomes nontrivial, and there are remarkable classes of linear
groups for which little is known about the automaton models that they
define. 

We start with a survey of extended finite automata, and present the basic definitions and observations in Section \ref{sec: efadefn}. 

In Section \ref{sec: efalang}, we present several new results about the classes of
languages recognized by finite automata over matrix groups. We focus
on matrix groups with integer and rational entries. For the case of $
2 \times 2 $ matrices, we prove that the corresponding group automata
for rational matrix groups are more powerful than the corresponding
group automata for integer matrix groups, which recognize exactly the class of context-free languages. We also explore finite
automata over some special matrix groups, such as the discrete
Heisenberg group and the Baumslag-Solitar group. The ``zoo" of
language classes associated with different groups is presented,
visualizing known relationships and open problems.

We also introduce the notion of time complexity for group automata,
and use this additional dimension to analyze the relationships among
the language families recognized by finite automata over various groups. We
develop a method for proving that automata over groups where the growth rate of the group and the time are bounded cannot recognize certain languages, even if one uses a very weak definition of time-bounded computation, and use this to demonstrate
some new relationships between time-bounded versions of our language
classes. The case of linear-time bounds is examined in detail
throughout our repertory of matrix groups. The results are presented in Section \ref{sec: efatime}.

In Section \ref{sec: efadecision}, we make a connection between the membership and identity problems for matrix groups and semigroups and the corresponding extended finite automata. We prove that the decidability of the emptiness and universe problems for extended finite automata are sufficient conditions for the decidability of the subsemigroup membership and identity problems. Using these results, we provide an alternative proof for the decidability of the subsemigroup membership problem for $ GL(2,\mathbb{Z}) $ and the decidability of the identity problem for $ M_2(\mathbb{Z}) $. We show that the emptiness and universe problems for $ SL(4,\mathbb{Z}) $ are undecidable. 

\section{Basic Notions, Definitions and Survey on Extended Finite Automata}\label{sec: efadefn}

This introductory section provides a brief overview of extended finite automata and reviews the literature. 

We gave the definition for finite automaton in Section \ref{sec: back-c-f}. Now let us look from the point of view of combinatorial group theory.

In \cite{Gi96}, a \textit{finite automaton} $ \cal F $ \textit{over a monoid $ M $} is defined as a finite directed graph whose edges are labeled by elements from $ M $. $\cal  F $ consists of a vertex labeled as the initial vertex and a set of vertices labeled as the terminal vertices such that an element of $ M $ is accepted by $ \mathcal{F} $ if it is the product of the labels on a path from the initial vertex to a terminal vertex. A subset of $ M $ is called \textit{rational} if its elements are accepted by some finite automaton over $ M $. The idea of rational subset of a monoid is introduced for the first time in \cite{ES69}.

When $ M $ is a free monoid such as $ \Sigma^* $, then the accepted elements are words over $ \Sigma $ and the set of accepted words is a language over $ \Sigma $. If $ M $ is finitely generated by a set $  A$, then equivalently a subset of $ M $ is called rational if its elements are accepted by some finite automaton over $ A $. By letting $ A $ to be a finite alphabet $ \Sigma $, we see that the two definitions of finite automaton coincide. Rational subsets of a free monoid are called rational (regular) languages.

An $ M $-automaton recognizing a language over the alphabet $ \Sigma $ can be seen as a finite automaton over the monoid $ \Sigma^* \times M $ such that the accepted elements are $ (w,1) $, where $ w \in \Sigma^* $. This is stated explicitly in the following proposition by Corson \cite{Co05}. The proof involves constructing an $ M $-automaton from a finite automaton over $ \Sigma^* \times M $ and vice versa.

\begin{fact} \label{fact: corson}\textup{\cite{Co05}} 
Let $ L $ be a language over an alphabet $ \Sigma $. Then $ L $ is recognized by an $ M $-automaton if and only if there exists a rational subset $ R \subseteq \Sigma^* \times M $ such that $ L = \{w  \in  \Sigma^{∗} | (w,1) \in R\} $.
\end{fact}

Adopting the same definition, one can define a pushdown automaton as a finite state automaton over $  \Sigma^* \times M_{cf}  $ where $  M_{cf} $ is a certain monoid characterizing the context-free languages \cite{Gi96}. A word $ w \in \Sigma^* $ is accepted by a pushdown automaton if there is a path from initial vertex to terminal vertex with label $ (1,w) $. Replacing $ M_{cf} $ with different monoids, it is possible to recognize other classes of languages. This idea coincides with the formal definition of extended finite automaton, which will be discussed next.

The definition of extended finite automata appeared explicitly for the first time in a series of papers by Dassow, Mitrana and Steibe \cite{MS97,DM00,MS01}. An extended finite automaton is formally defined as follows:

Let $ G $ be a group under the operation denoted by 
$ \circ $ with the neutral element denoted by $ e $. 
An \textit{extended finite automaton} over the group $G$ is a
6-tuple
\[ \mathcal{E} = (Q, \Sigma,G,\delta, q_1,Q_a), \]
where the transition function $\delta$ is defined as
\[\delta: Q \times \Sigma_{\varepsilon} \rightarrow \mathcal{P}(Q\times G).\]

$ \delta(q,\sigma) \ni (q',g) $ means that when $\mathcal{E}$ reads the
symbol (or empty string) $\sigma \in \Sigma_{\varepsilon}$
in state $q \in Q$, it moves to state $q' \in Q$, and writes $ x\circ g $ in the register, 
where $ x $ is the old content of the register and $ g \in G $.

The initial value of the register is the 
neutral element $ e $ of the group $ G $. An input string $ w $ of length $ n$ is accepted if $\mathcal{E}$ enters an accept state with the tape head on the $ n+1 $'st square and the content of the register is equal to the identity element of $ G$. 

The class of languages recognized by an extended finite automaton over $ G $ will be denoted by 
$ \mathfrak{L}(G) $.

An extended finite automaton over $ G $ is also called a finite automaton over $ G $ or a $ G $-automaton and we will use the three terms interchangeably 

%Note that the definition of extended finite automaton does not involve processing the right endmarker. From Lemma \ref{lem: rtNBend}, we know that the computational power wouldn't change if one allowed the extra postprocesing step, as these machines are nondeterministic and blind.

Extended finite automata have appeared implicitly throughout the literature as many classical models can be regarded as finite automaton over a particular group. Pushdown automata \cite{Ch62}, blind counter automata \cite{Gr78} and finite automata with multiplication without equality \cite{ISK76} are extended finite automata where the group in consideration is the free group, the additive group of integer vectors $ \mathbb{Z}^k $, and the multiplicative group of nonzero rational numbers $ \mathbb{Q}^+ $, respectively.

Mitrana and Stiebe investigate the language recognition power of finite automata over Abelian groups and conclude the following result.

\begin{fact}\textup{\cite{MS97}}
For an Abelian group $ G $, one of the following relations hold: 
\begin{align*} 
\mathfrak{L}(G) &= \mathsf{REG},\\
\mathfrak{L}(G) &= \mathfrak{L}(\mathbb{Z}^k), \mbox{ for some } k,\\
\mathfrak{L}(G) &= \mathfrak{L}(\mathbb{Q}^+).
\end{align*}
\end{fact}
They also discuss the computational power of deterministic extended finite automata, proving that they are less powerful than their nondeterministic variants. Throughout this thesis, we will focus on nondeterministic extended finite automata.

In the case of the free groups, Dassow and Mitrana observe the following characterization for the classes of context-free languages and recursively enumerable languages. Although it is true that $ \mathbf{F}_2 $-automata recognize exactly the class of context-free languages, the proof given in \cite{DM00} is not correct.

\begin{fact}\label{fact: cf} \textup{\cite{DM00}}
$ \mathfrak{L}(\mathbf{F}_2) = \mathsf{CF}$.
\end{fact}

\begin{fact} \label{fact: re} \textup{\cite{MS97,MS01}}
	$ \mathfrak{L}(\mathbf{F}_2 \times \mathbf{F}_2) = \mathsf{RE}$.
\end{fact}

In 2005, Corson modified the definition of extended finite automaton by allowing the register to be multiplied with monoid elements \cite{Co05}. Extended finite automaton over a monoid is also called a monoid automaton or $ M$-automaton. Corson focuses on the connection between the word problem of a group $ G $ and the set of languages recognized by $ G $-automata. He also provides a proof for the fact that $ \mathfrak{L}(\mathbf{F}_2) = \mathsf{CF} $ by extending the work of Gilman \cite{Gi96}.

Another line of research on monoid automata was led by Kambites \textit{et al.} The results concerning the class of languages recognized by monoid automata appear in \cite{Ka09, RK09, Re10}. The connection between a given group $ G $ and the groups whose word problems recognized by $ G$-automata has been studied in \cite{EO04,Ka06,EKO08}.

Recent work on the subject includes that of Corson \textit{et al}. \cite{CR15, RCR17} and Zetzsche \cite{Ze16}. Corson deals with monoid automata recognizing word problems of free products of groups in \cite{CR15}. Real-time $ G $-automata where $ \varepsilon $-transitions are not allowed are analyzed in \cite{RCR17}. Zetzsche investigates the area from a different point of view, not focusing on particular monoids but proving and generalizing various properties of finite automata over a broad class of monoids. 

Even though extensive research has been carried out  
on extended finite automata, no study exists which deal with finite automata over matrix groups. 
Having presented the main findings on the subject, we will move on to discuss finite automata over matrix groups in the next section. 

\section{Languages Recognized by Finite Automata over Matrix Groups}
\label{sec: efalang}

In this section, we are going to prove some new results about the
classes of languages recognized by finite automata over matrix groups. We will start with some observations from the previous studies. In the remaining parts, we will analyze the language recognition power of finite automata over various matrix groups. 

\subsection{Observations}

Let us start by noting the following facts, which are true by the definitions of the machines.

\begin{itemize}
	\item A $ \mathbb{Z}^k $-automaton is equivalent to a one-way nondeterministic blind $k  $-counter automaton (1N$ k $BCA).
	\item A $ \mathbb{Q}^+ $-automaton is equivalent to a one-way finite automaton with multiplication without equality (1NFAMW). 
\end{itemize}

As mentioned earlier, the characterization of context-free languages by $ \mathbf{F}_2 $-automata was
first stated by Dassow and Mitrana \cite{DM00}, and proven in \cite{Co05}. 
Let us recall that $ \mathbf{F}_2 $ contains any free group
of rank $ n \geq 2 $ \cite{LS77}.
%
%\begin{fact}\label{fact: cf}\textup{\cite{DM00, Co05, Ka09}}
%	$\mathfrak{L}(\mathbf{F}_2)$ is the family of context-free languages.
%\end{fact}

The relation between the classes of languages recognized by free group
automata is summarized as follows.

\begin{fact}\label{fact: reg} \textup{\cite{DM00}}		
	$ \mathsf{REG} = \mathfrak{L}(\mathbf{F}_0) \subsetneq
	\mathfrak{L}(\mathbf{F}_1) = \mathfrak{L}(\mathbb{Z}) \subsetneq
	\mathfrak{L}(\mathbf{F}_2) = \mathsf{CF} $.
\end{fact}

The following result states the hierarchy between the classes of languages 
recognized by $ \mathbb{Z}^k $-automata. 
%This result also follows from the
%hierarchy between the class of languages recognized by nondeterministic 
%blind $ k $-counter automata.

\begin{fact}\textup{\cite{CEO06}}
	$ \mathfrak{L}(\mathbb{Z}^k) \subsetneq
	\mathfrak{L}(\mathbb{Z}^{k+1})$ for $ k \geq 1 $.
\end{fact}

Let us mention that the class of context-free languages and the class
of languages recognized by nondeterministic blind counter automata are
incomparable.

\begin{fact}\label{fact: cfzn}
	$\mathsf{CF}$ and $\mathfrak{L}(\mathbb{Z}^k)$ are incomparable for
	all $ k \geq 2 $.		
\end{fact}

\begin{proof} Consider the language ${L}=\{a^nb^n| n\geq 0\}$ which is a
	context-free language. Since context-free languages are closed under star, 
	${L}^*$ is a context-free language whereas it cannot be recognized 
	by any $\mathbb{Z}^k$-automaton for all $ k \geq 1 $ by
	\cite{Gr78}. On the other hand, the non-context-free language 
	${L}'=\{a^nb^nc^n|n \geq 0\}$ can be recognized by a
	$\mathbb{Z}^2$-automaton.
\end{proof}

\subsection{Automata on Groups of $ 2\times2 $ and $ 3\times3 $ Matrices}\label{Section: 23matrices}
Let $ {G_2} $ be the group generated by the matrices
\[
M_{a}=
\mymatrix{rr}{1&2\\
	0&1\\}
~~~\mbox{and}~~~
M_{b}=
\mymatrix{rr}{1&0\\
	2&1\\}.
\]
There exists an isomorphism $ \varphi $ from $\mathbf{F}_2 $ onto
${G_2} $ by \cite{KM79}, meaning that the group $ G_2 $ is isomorphic to $ \mathbf{F}_2	 $. Note that $ M_a $ and $ M_b $ are
integer matrices with determinant 1, which proves that $ \mathbf{F}_2
$ is a subgroup of $ SL(2,\mathbb{Z}) $.

Now the question is whether $ \mathfrak{L}(GL(2,\mathbb{Z}))$ and $
\mathfrak{L}(SL(2,\mathbb{Z}))$ correspond to larger classes of
languages than the class of context-free languages. We are going to
use the following fact to prove that the answer is negative.

\begin{fact}\textup{\cite{Co05}}  \label{fact: finite}
	Suppose $G$ is a finitely generated group and $H$ is a subgroup of
	finite index. Then $\mathfrak{L}(G) = \mathfrak{L}(H)$.
\end{fact}

Now we are ready to state our theorem.

\begin{thm}\label{theorem:gl}
	$ \mathsf{CF} =\mathfrak{L}(\mathbf{F}_2) = \mathfrak{L}(SL(2,\mathbb{Z})) = \mathfrak{L}(GL(2,\mathbb{Z}))$. 
\end{thm}
\begin{proof} We are going to use Fact \ref{fact: finite} to prove the result. 
	Since $ SL(2,\mathbb{Z}) $ has index 2 in $ GL(2,\mathbb{Z})$
	and $GL(2,\mathbb{Z})$ is finitely generated, $\mathfrak{L}(GL(2,\mathbb{Z})) =
	\mathfrak{L}(SL(2,\mathbb{Z}))$. Since $\mathbf{F}_2$ has index 12 in 
	$SL(2,\mathbb{Z})$ \cite{BO08} and $SL(2,\mathbb{Z})$ is finitely generated, $
	\mathfrak{L}(SL(2,\mathbb{Z})) = \mathfrak{L}(\mathbf{F}_2)$ which is equal to 
	the family of context-free languages by Fact \ref{fact: cf}.
\end{proof}

In the next theorem, we prove that allowing the register to be multiplied by integer matrices whose determinants are not $ \pm 1 $ does not increase the computational power.

\begin{thm} \label{thm: M2Z}
	$ \mathfrak{L}(M_2(\mathbb{Z})) = \mathsf{CF} $.
\end{thm}
\begin{proof}
Suppose that an $ M_2(\mathbb{Z}) $-automaton $ \mathcal{E} $ is given. When $ \mathcal{E} $ processes an input string, its register is initialized by the identity matrix and multiplied by matrices from $ M_2(\mathbb{Z}) $. Suppose that in a successful computation leading to acceptance, the register is multiplied by some singular matrix whose determinant is 0. Then the product of the matrices multiplied with the register will have determinant 0 and the register cannot be equal to the identity matrix again. Similarly, if the register is multiplied with a nonsingular matrix whose determinant is not equal to $ \pm 1$, then the determinant of the product of the matrices multiplied with the register cannot be equal to 1 again, since $ M_2(\mathbb{Z}) $ does not contain any matrix with non-integer determinants. Any such edges labeled by a matrix whose determinant is not equal to $ \pm 1 $ can be removed from $ \mathcal{E} $ to obtain a $ GL(2,\mathbb{Z}) $-automaton, without changing the accepted language. Since $ \mathfrak{L} (GL(2,\mathbb{Z})) = \mathsf{CF}$, the result follows by Theorem \ref{theorem:gl}. 
\end{proof}

Let us now investigate the group $SL(3,\mathbb{Z})$, the group of 
$3 \times 3 $ integer matrices with determinant~$1$. We start by looking at an important subgroup of $SL(3,\mathbb{Z})$,
the discrete Heisenberg group. The discrete Heisenberg group
$\mathbf{H}$ is defined as $\langle a,b | ab=bac,ac=ca,bc=cb
\rangle$,
where
$c=a^{-1}b^{-1}ab$ is called the \textit{commutator} of $a$ and $b$.
$$ a=
\mymatrix{rrr}{1&1&0\\
	0&1&0\\
	0&0&1}
~~~
b=
\mymatrix{rrr}{1&0&0\\
	0&1&1\\
	0&0&1}
~~~
c=
\mymatrix{rrr}{1&0&1\\
	0&1&0\\
	0&0&1}
$$
Any element $g \in \mathbf{H}$ can be written uniquely as $b^ja^ic^k$.
$$ g=
\mymatrix{rrr}{1&i&k\\
	0&1&j\\
	0&0&1}
= b^ja^ic^k
$$

It is shown in \cite{Re10} that the languages
$\mathtt{MULT}=\{x^py^qz^{pq} | p,q\geq 0\}$,
$\mathtt{COMPOSITE}=\{x^{pq}| p,q >1\}$ and
$\mathtt{MULTIPLE}=\{x^py^{pn}|p \in \mathbb{N}\}$ can be recognized
by $\mathbf{H}$-automata, using the special multiplication property
of the group.

Correcting a small error in \cite{Re10}, we rewrite the multiplication
property of the elements
of $\mathbf{H}$.
\begin{equation*}\label{eq0}
(b^xa^yc^z)(b^{x'}a^{y'}c^{z'})= b^{x+x'}a^{y+y'}c^{z+z'+yx'} 
\end{equation*}
We can make the following observation using the fact that
$\mathfrak{L}(\mathbf{H})$ contains non-context-free languages.

\begin{thm}\label{thm: sl2z3z}
	$\mathfrak{L}(SL(2,\mathbb{Z})) \subsetneq \mathfrak{L}(SL(3,\mathbb{Z}))$.
\end{thm}
\begin{proof}
	It is obvious that an $SL(2,\mathbb{Z})$-automaton can be simulated by
	an $SL(3,\mathbb{Z})$-automaton.
	Note that $\mathfrak{L}(SL(2,\mathbb{Z}))$ is the family of
	context-free languages by Theorem \ref{theorem:gl}. Since
	$\mathfrak{L}(\mathbf{H}) \subseteq \mathfrak{L}(SL(3,\mathbb{Z}))$
	and the non-context-free language $\mathtt{MULT}=\{x^py^qz^{pq} | p,q\geq 0\}$ can be recognized by an $ \mathbf{H} $-automaton \cite{Re10}, the
	result follows.
\end{proof}

The following result
is a direct consequence of Fact \ref{fact: finite}.

\begin{thm}\label{thm: sl3zgl3z}
	$ \mathfrak{L}(SL(3,\mathbb{Z}))=\mathfrak{L}(GL(3,\mathbb{Z})) $.
\end{thm}

\begin{proof} Since $ GL(3,\mathbb{Z}) $ is a finitely generated group and $
	SL(3,\mathbb{Z}) $ has finite index in $ GL(3,\mathbb{Z}) $, the
	result follows by Fact \ref{fact: finite}.
\end{proof}

We have talked about the discrete Heisenberg group $ \textbf{H} $. Now let us look at a
subgroup of $\mathbf{H}$ generated by the matrices $B$ and $C$, which
we will call ${H_2}$.
$$ B=
\mymatrix{rrr}{1&0&0\\
	0&1&1\\
	0&0&1}
~~~
C=
\mymatrix{rrr}{1&0&1\\
	0&1&0\\
	0&0&1}
$$	

${H_2}=\langle B,C |BC=CB\rangle $ is a free Abelian group of
rank 2, and therefore it is isomorphic to $\mathbb{Z}^2$.

We conclude the following about the language recognition power of $
\mathbb{Z}^2$ and $\mathbf{H}$.

\begin{thm}\label{theorem: Z2H}
	$\mathfrak{L}(\mathbb{Z}^2) \subsetneq \mathfrak{L}(\mathbf{H})$.
\end{thm}
\begin{proof}
	Since $\mathbb{Z}^2$ is a subgroup of $\mathbf{H}$,
	$\mathfrak{L}(\mathbb{Z}^2) \subseteq 
	\mathfrak{L}(\mathbf{H}) $ follows. The inclusion is proper since
	$\mathbf{H}$-automaton can recognize the language 
	$\mathtt{MULT}=\{x^py^qz^{pq} | p,q\geq 0\}$
	\cite{Re10}, whereas any bounded language in $ \mathfrak{L}(\mathbb{Q}^+) $ is semilinear by Fact \ref{fact: bounded}.
\end{proof}

Now let us move on to the discussion about matrix groups with rational entries. We will start with a special subgroup of $GL(2,\mathbb{Q})$.

For two integers $m$ and $n$, the \textit{Baumslag-Solitar group}
$BS(m,n)$ is defined as $BS(m,n)=\langle
a,b | ba^mb^{-1}=a^n\rangle$. We are going to focus on
$BS(1,2)=\langle a,b|bab^{-1}=a^2\rangle$.

Consider the matrix group $G_{BS}$ generated by the matrices
$$A=
\mymatrix{rr}{1&0\\
	-1 &1\\}
~~~\mbox{and}~~~
B=\mymatrix{rr}{1/2&0\\
	0 &1}
.$$

Consider the isomorphism $a \mapsto A$, $b \mapsto B$. The matrices
$A$ and $B$ satisfy the property $BAB^{-1} = A^2$,
\[ 
\mymatrix{rr}{1/2&0\\
	0 &1\\}
\mymatrix{rr}{1&0\\
	-1 &1\\}
\mymatrix{rr}{2&0\\
	0&1\\}
=
\mymatrix{rr}{1&0\\
	-2&1\\},  \]
and we conclude that $G_{BS}  $ is isomorphic to $BS(1,2)$.

We will prove that there exists a $ BS(1,2) $-automaton which
recognizes a non-context-free language.

\begin{thm}\label{theorem: upow}
	$ \mathfrak{L}(BS(1,2)) \nsubseteq \mathsf{CF} $.
\end{thm}
\begin{proof}
	Let us construct a $BS(1,2)$-automaton $\mathcal{E}$ recognizing the
	language $\mathtt{UPOW}=\{a^{2^n}|n\geq 0\}$. The state diagram of 
	$ \mathcal{E} $ and the matrices are given in Figure \ref{fig:upow}. 
	Without scanning any input symbol, $\mathcal{E}$ multiplies its register 
	with the matrix $A_1$ successively. $\mathcal{E}$
	nondeterministically moves to the next state reading the first input
	symbol without modifying the register.
	After that point, $\mathcal{E}$ starts reading the string and
	multiplies its register with the matrix $A_2$
	for each scanned $a$. At some point, $\mathcal{E}$
	nondeterministically stops reading
	the rest of the string and multiplies its register with the element
	$A_3$. After successive multiplications
	with $A_3$, $\mathcal{E}$ nondeterministically decides to move to an
	accept state. 
	\begin{figure}[!htb]
		\centering
		\includegraphics[width=0.7\linewidth]{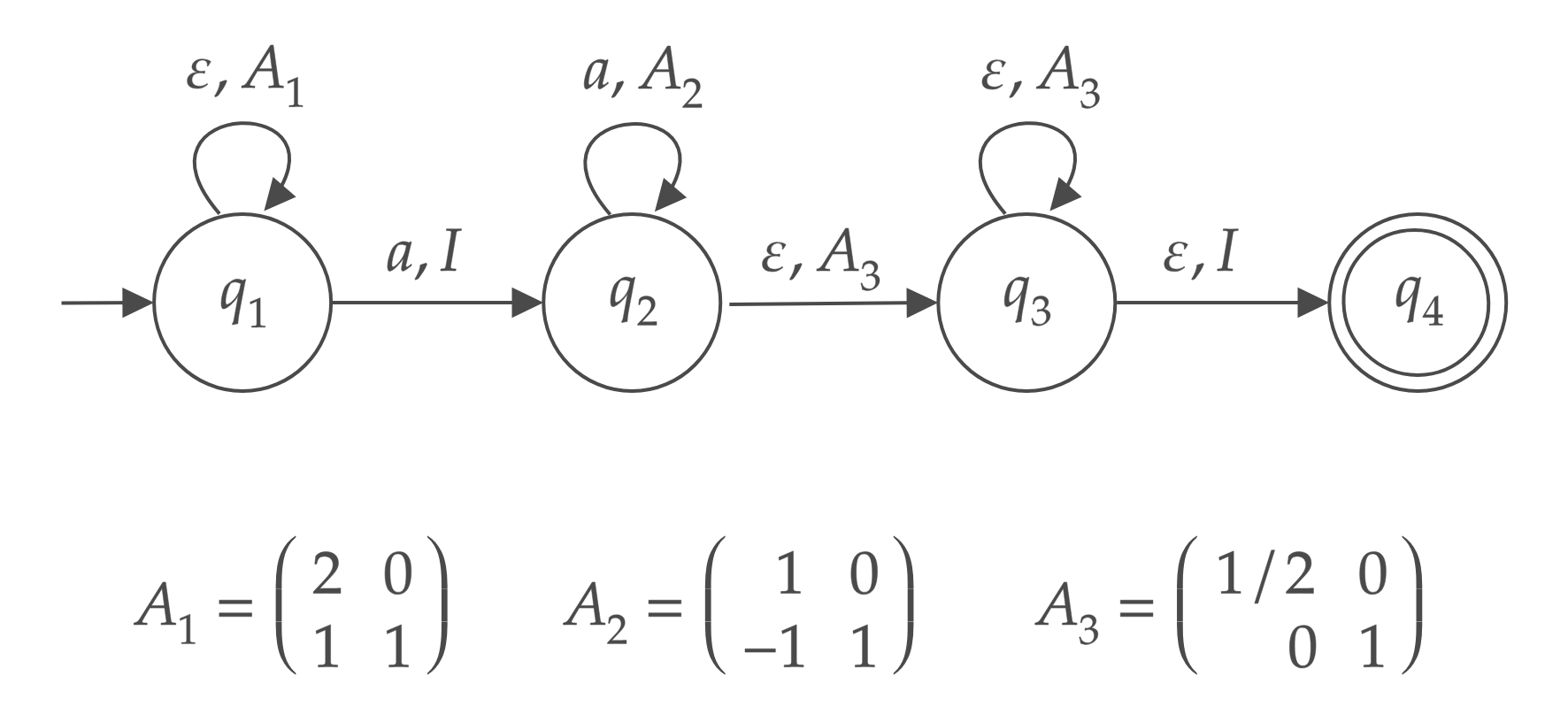}
		\caption{State transition diagram of $ \mathcal{E} $ recognizing $\mathtt{UPOW}$}
			\label{fig:upow}
			
	\end{figure}

	As a result of $ i $ multiplications with $ A_1 $, the register has the value 
	$$
	\mymatrix{rr}{	2^i&0\\
		2^i -1 &1\\}
$$
	before reading the first input symbol. Multiplication
	with each $A_2$ leaves $2^i$ unchanged while subtracting 1 from $2^i -
	1$ for each scanned $a$. The register will have the value 
	$$
	\mymatrix{rr}{	2^i&0\\
		2^i -1 - j &1\\}
$$
	as a result of $ j $ multiplications with the matrix $ A_2 $.
	
	For the rest of the computation, $\mathcal{E}$ will multiply its
	register, say $ k $  times, with $A_3$ resulting in the register value 
	
	$$
	\mymatrix{rr}{	\frac{2^i}{2^k}&0\\
		2^i -1 - j &1\\}
	$$
	since each multiplication with $ A_3 $ divides $ 2^i $ by 2.
	
	The register contains the identity matrix at the end of the computation if $ i=k $ and $ j=2^i-1 $ which is possible if the input string is of the form $ a^{1+2^{i}-1} = a^{2^i} $. In the successful
	branch, the register will be equal to the
	identity matrix and $\mathcal{E}$ will end up in the final state
	having successfully read the input string.
	
	For input strings which are not members of $\mathtt{UPOW}$, either the
	computation will end before reading the
	whole input string or the final state will be reached with the
	register value being different from the
	identity matrix.
	Note that $A_1=B^{-1}A^{-1}$, $A_2=A$ and $A_3=B$, where $A$ and $B$
	are the generators of the group $G_{BS}$ and recall that $G_{BS}$ is
	isomorphic to $BS(1,2)$. Since $ \mathtt{UPOW} $ is a unary nonregular
	language, it is not context-free and we conclude the result.
\end{proof}

Note that $\mathfrak{L}(\mathbb{Z}) \subsetneq \mathfrak{L}(BS(1,2)) $ 
since the subgroup generated by $ a $ in $ BS(1,2) $ is isomorphic to 
$ \mathbb{Z} $ and $\mathfrak{L} (BS(1,2)) $ contains a unary nonregular language.

We showed that allowing rational entries enlarges the
class of languages recognized by $ 2 \times 2 $ matrices. What about the group of $2 \times 2  $ rational matrices with determinant 1? We give a positive answer for the question, by constructing an $ SL(2,\mathbb{Q}) $-automaton recogizing a unary non-context-free langauge.

\begin{thm}\label{theorem: UPOWODD}
	$ \mathfrak{L}(SL(2,\mathbb{Z})) \subsetneq \mathfrak{L}(SL(2,\mathbb{Q}))  $.
\end{thm}
\begin{proof} It is obvious that $ \mathfrak{L}(SL(2,\mathbb{Z}))
	\subseteq \mathfrak{L}(SL(2,\mathbb{Q}))  $. We will prove that the
	inclusion is proper.
	
	Let us construct an $SL(2,\mathbb{Q})$-automaton $ \mathcal{E}$
	recognizing the language $\mathtt{UPOW_{odd}}=\{a^{2^{2n+1}}|$  $ n \geq 0 \}$. 
	The state diagram of $ \mathcal{E} $ and the matrices are given in 
	Figure \ref{fig:uoddpow}. Without scanning
	any input symbol, $\mathcal{E}$ first multiplies its register with
	the matrix $A_1$. $\mathcal{E}$ then multiplies its register with the
	matrix $ A_2 $ successively until nondeterministically moving to the
	next state. After that point, $\mathcal{E}$ starts reading the string
	and multiplies its register with the matrix $A_3$ for each scanned
	$a$. At some point, $\mathcal{E}$ nondeterministically stops reading
	the rest of the string and multiplies its register with the matrix
	$A_4$. After successive multiplications with $A_4$, $\mathcal{E}$
	nondeterministically decides moving to an accept state. 
	
	\begin{figure}[h]
		\centering
		\includegraphics[width=0.9\linewidth]{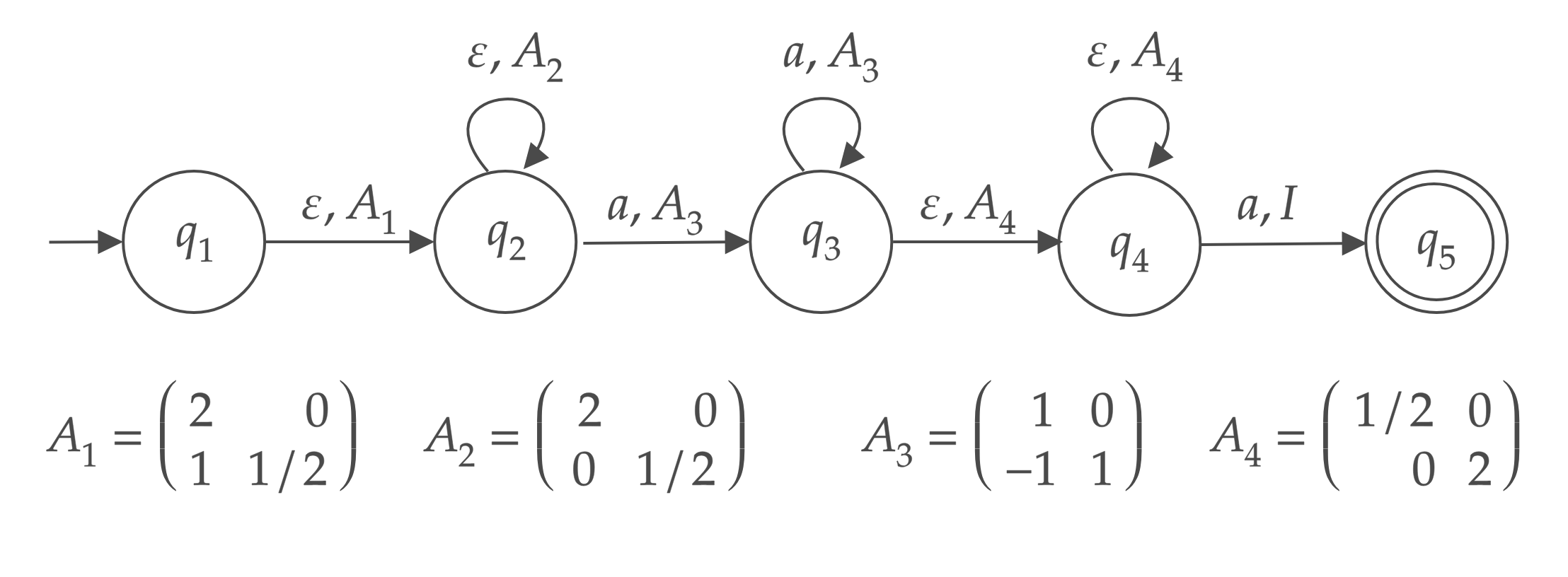}

		\caption{State transition diagram of $\mathcal{E} $ recognizing $ \mathtt{UPOW_{odd}} $ }
		\label{fig:uoddpow}	
	\end{figure}

	Let us trace the value of the register at different stages of the
	computation. Before reading the first input symbol, the register has
	the value
	$$
	\mymatrix{rr}{	2^{x+1}&0\\
		2^x  &\frac{1}{2^{x+1}}\\}
$$	
	\noindent as a result of the multiplications with the matrix $A_1$
	and $ x $ times the matrix $ A_2 $. Multiplication
	with each $A_3$ leaves $2^{x+1}$ and $ \frac{1}{2^{x+1}} $ unchanged
	while subtracting $ \frac{1}{2^{x+1}} $ from $2^x$ for each scanned
	$a$. As a result of $ y $  multiplications with $ A_3 $, the register
	will have the value
	$$
	\mymatrix{rr}{2^{x+1}&0\\
		2^x-\frac{y}{2^{x+1}} &\frac{1}{2^{x+1}}\\}
	.$$
	
	For the rest of the computation, $\mathcal{E}$ will multiply its
	register with $A_4$ until
	nondeterministically moving to the final state. As a result of $ z $
	multiplications with $ A_4 $, the register will have the value
	$$
\mymatrix{rr}{	\frac{2^{x+1}}{2^z}&0\\
	\bigl( 2^x-\frac{y}{2^{x+1}}\bigr)\frac{1}{2^z} &\frac{2^z}{2^{x+1}}\\}	
.$$
	
	The final value of the register is equal to the identity matrix when
	$ y=2^{2x+1} $ and $ z=x+1 $, which is possible only when the length
	of the input string is $ 2^{2x+1} $ for some $ x\geq 0 $. In the
	successful branch, the register will be equal to the identity matrix
	and $\mathcal{E}$ will end up in the final state having successfully
	read the input string. For input strings which are not members of
	${L}$, either the computation will end before reading the whole
	input string, or the final state will be reached with the register
	value not equaling the identity matrix.
	
	Since the matrices used during the computation are 2 by 2 rational
	matrices with determinant 1, $ {L} \in
	\mathfrak{L}(SL(2,\mathbb{Q})) $. $ \mathfrak{L}(SL(2,\mathbb{Q})) $
	contains a unary nonregular language, which is not true for $
	\mathfrak{L}(SL(2,\mathbb{Z})) $ by Theorem \ref{theorem:gl} and we
	conclude the result.
\end{proof}

Let us note that the set of languages recognized by $
\mathbb{Q}^+ $-automata is a proper subset of the set of languages
recognized by $ SL(2,\mathbb{Q}) $-automata.

\begin{thm}\label{thm: qsl2q}
	$ \mathfrak{L}(\mathbb{Q}^+) \subsetneq \mathfrak{L}(SL(2,\mathbb{Q}))  $.
\end{thm}

\begin{proof} Let $ {L} \in \mathfrak{L}(\mathbb{Q}^+) $ and let $
	\mathcal{E} $ be a $ \mathbb{Q}^+ $-automaton recognizing $ {L}
	$. We will construct an $ SL(2,\mathbb{Q}) $-automaton  $
	\mathcal{E}' $ recognizing $ {L} $. Let $ S =\{s_1,\dots,s_n\}
	$ be the set of elements multiplied with the register during the
	computation of $ \mathcal{E} $. We define the mapping $ \varphi $ as
	follows.
	\[ \varphi: s_i 	\mapsto
	\mymatrix{rr}{	s_i&0\\
		0&\frac{1}{s_i}\\}
	\]
	The elements $\varphi(s_i) $ are $ 2 \times 2 $ rational matrices with
	determinant 1. Let $ \delta $ and $ \delta' $ be the transition
	functions of $ \mathcal{E} $ and $ \mathcal{E}' $ respectively. We let
	$$ (q',s_i) \in \delta(q,\sigma) \iff (q',\varphi(s_i)) \in
	\delta'(q,\sigma)$$ for every $ q,q' \in Q$, $ \sigma \in \Sigma $ and
	$ s_i \in S $. The resulting $ \mathcal{E}' $ recognizes $ {L}$.
	
	The inclusion is proper since  $\mathtt{UPOW_{odd}} =\{a^{2^{2n+1}} | n \geq 0
	\} \in \mathfrak{L}(SL(2,\mathbb{Q}))$ by Theorem \ref{theorem:
		UPOWODD}, and $ \mathfrak{L}(\mathbb{Q}^+) $ does not contain any
	unary nonregular language by Fact \ref{fact: unary}, noting that 
	$ \mathbb{Q}^+ $-automata are equivalent to 1NFAMW's.
\end{proof}
\subsection{Automata on Matrices of Higher Dimensions}

As pointed out in Section \ref{sec: efadefn}, $\mathbf{F}_2 \times \mathbf{F}_2 $-automata are as powerful as Turing machines. Using this fact, we make the following observation.

\begin{thm}\label{thm: sl4z}
	$\mathsf{RE} =  \mathfrak{L} ( \mathbf{F}_2 \times \mathbf{F}_2 ) =
	\mathfrak{L} (SL(4,\mathbb{Z}) )$.
\end{thm}

\begin{proof} The first equality is Fact \ref{fact: re}.
	Recall from Section \ref{Section: 23matrices} that $ \varphi $ is an isomorphism from $ \mathbf{F}_2$ onto $
	{G_2} $, the matrix group generated by the matrices $ M_a $ and $
	M_b $.
	Let $ {G}' $ be the following group of matrices
	\[ \left \{
	\left (
	\begin{array}{clll}
	\multicolumn{2}{l}
	{\multirow{2}{*}{$M_1$}} & 0 & 0 \\
	& & 0 & 0 \\
	0&0 & \multicolumn{2}{c}{\multirow{2}{*}{ $M_2$}} \\
	0&0 & & \\
	\end{array}
	\right ), \ M_1, \ M_2 \in {G_2} \right \}
	.\]
	
	We will define the mapping $ \psi:\mathbf{F}_2 \times \mathbf{F}_2
	\rightarrow {G}' $ as $ \psi(g_1,g_2) =
	(\varphi(g_1),\varphi(g_2)) $ for all $ (g_1,g_2)\in \mathbf{F}_2
	\times \mathbf{F}_2 $ which is an isomorphism from $\mathbf{F}_2
	\times \mathbf{F}_2  $ onto $ {G}' $.
	
	This proves that $\mathbf{F}_2 \times \mathbf{F}_2  $ is isomorphic to
	a subgroup of $ SL(4,\mathbb{Z}) $. The fact that $  \mathfrak{L} (
	\mathbf{F}_2 \times \mathbf{F}_2 ) $ is the set of recursively
	enumerable languages lets us conclude that $
	\mathfrak{L}(SL(4,\mathbb{Z})) $  is the set of recursively enumerable
	languages.
\end{proof}

Let us also state that the classes of languages recognized by automata
over supergroups of $ SL(4,\mathbb{Z}) $ such as $ GL(4,\mathbb{Z})
$ or $ SL(4,\mathbb{Q}) $ are also identical to the class of
recursively enumerable languages.
\begin{thm}\label{thm: RE}
	$ \mathfrak{L}(G)  = \mathsf{RE} $, where $ G $ is any matrix group whose matrix entries are computable numbers and contains $ SL(4,\mathbb{Z})$ as a subgroup.
\end{thm}
\begin{proof}
	Note that any finite automaton over a matrix group can be simulated by a nondeterministic Turing machine which keeps track of the register simply by multiplying the matrices and checking whether the identity matrix is reached at the end of the computation, provided that the matrix entries are computable numbers. Since $ \mathsf{RE}=\mathfrak{L} (SL(4,\mathbb{Z}) ) $ and $ G $ contains $ SL(4,\mathbb{Z})$ as a subgroup, we conclude that $ \mathfrak{L}(G) $ is the set of recursively enumerable languages.
\end{proof}

We summarize the results in Figure \ref{fig: diagram}. Solid arrows
represent proper inclusion, dashed arrows represent inclusion and
dashed lines represent incomparability. 
\begin{figure}[h!]
	\centering
	\includegraphics[width=1\linewidth]{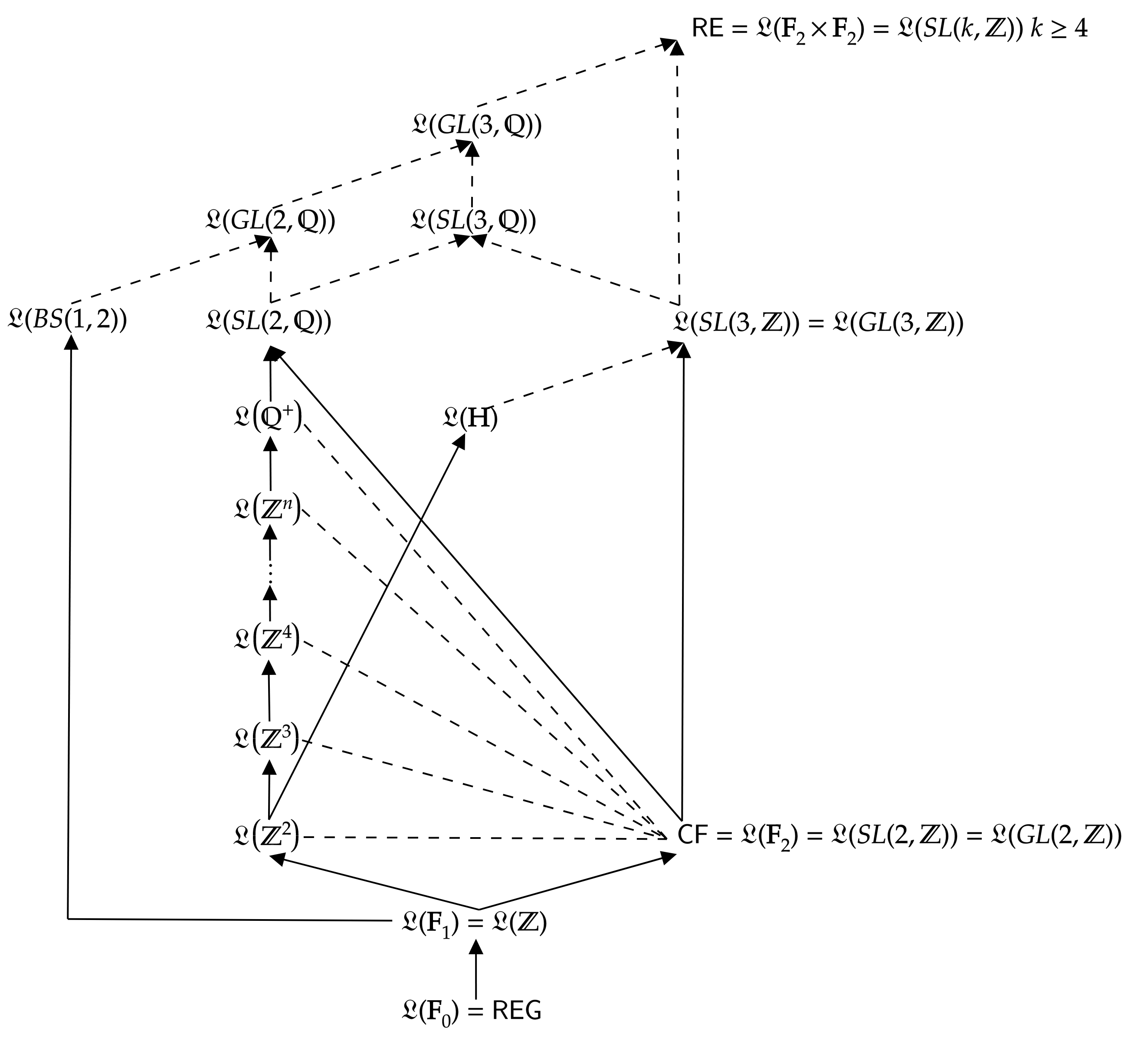}
	\caption{Language classes associated with groups}
	\label{fig: diagram}
\end{figure}

\section{Time Complexity}
\label{sec: efatime}
In the previous section, we compared various automaton models solely on the basis of the groups they employed as a computational resource. The theory of computational complexity deals with various different types of such resources, the allowed runtime of the machines being the most prominent among them. Some of the automata we saw in Section \ref{sec: efalang} (e.g. Figure \ref{fig:upow}) have arbitrarily long computations, and it is a legitimate question to ask whether our results, for instance, the relationships in Figure \ref{fig: diagram}, would still hold if one imposed common time bounds on the automata. We study such questions in this section.

\subsection{Definitions}

Before moving on with our discussion, we have to define some new concepts.

A $ G $-automaton $ \mathcal{E} $ recognizing language ${L} $ is said to be \textit{strongly $ t(n) $ time-bounded} if for any input string $ x $ with $ |x|=n $, every computation of $ \mathcal{E} $ on $ x $ takes at most $ t(n) $ steps. We will denote the set of languages recognized by strongly $ t(n) $-time bounded $ G $-automata by $ \mathfrak{L}(G)_{t(n)}^s $.

Although the strong mode of recognition defined above is standard in studies of time complexity, we will be able to prove the impossibility results of the next subsection even when the machines are subjected to the following, looser requirement:
A $ G $-automaton $ \mathcal{E} $ recognizing language ${L} $ is said to be \textit{weakly $ t(n) $ time-bounded} if for each accepted input string $ x \in {L} $ with $ |x|=n $, $ \mathcal{E} $ has a successful computation which takes at most $ t(n) $ steps. So any input string is allowed to cause longer computations, as long as none of those are accepting for inputs which are not members of ${L} $. We will denote the set of languages recognized by weakly $ t(n) $-time bounded $ G
$-automata by $ \mathfrak{L}(G)_{t(n)}^w $. Note that the statement $ \mathfrak{L}(G)_{t(n)}^s \subseteq \mathfrak{L}(G)_{t(n)}^w  $ is true by definition.

Let $ A $ be a generator set for the group $ G $. The \textit{length} of $ g \in
G $, denoted $|g|_A$, is the length of the shortest representative for $
g $ in $ (A \cup A^{-1})^* $. Let $$B^A_G(n)= \{g \in G,|g|_A\leq n \} $$ be the set of all elements in 
$ G $ which can be represented by a word of length at most $ n $. The \textit{growth function of a group} $ G $ with
respect to a generating set $ A $, denoted $ g^{A}_G(n)$, is the cardinality 
of the set $ B^A_G(n) $, that is $ g^{A}_G(n) = | B^A_G(n)|$. The growth function is asymptotically independent of the generating set, and we will denote the growth function of a group $ G $ by $ g_G(n)$.

For a positive integer $ n $, two strings $ w,w' \in \Sigma^* $ are $
n $-\textit{dissimilar for $ {L} $}  if $ |w|\leq n $, $ |w'|\leq n $, and there
exists a string $ v \in \Sigma^*  $ with $ |wv|\leq n
$, $ |w'v|\leq n $ such that $ wv \in {L} $ iff $w'v \notin
{L}  $. Let $ A_{L}(n) $ be the maximum $ k $ such that
there exist $ k $ distinct strings that are pairwise $n $-dissimilar.

A finite set of strings $ S $ is said to be a set of \textit{uniformly $ n $-dissimilar strings for $ {L} $} if for each string $ w\in S $, there exists a string $ v $ such that
$ |wv|\leq n$ and $wv \in {L} $ and for any string $ w'\in S $ such that $w\neq w' $,  $ |w'v|\leq n$ and $ w'v \notin {L} $. Let $ U_{L}(n) $ be the maximum $ k $ such that there exist $ k $ distinct strings that are uniformly $n $-dissimilar.

Note that the following is always true by definition, since the strings in a uniformly $ n $-dissimilar set are pairwise $ n $-dissimilar.

\begin{lem}\label{lemma: uniform}
	$ U_{L}(n) \leq A_{L}(n)$ for all $ n\geq 0 $. 
\end{lem}

\subsection{Limitations of Machines on Slow Groups Running in Short Time}\label{section: mainthm}

In this section, we are going to present a method for proving that certain languages cannot be recognized by finite automata over matrix groups when the growth rate of the group and the time are bounded.

\begin{thm}\label{thm: growth2}
	Let $ G $ be a group with growth function $ g_G(n)$. $ {L} \notin  \mathfrak{L}(G)_{t(n)}^w $ if $g_G(t(n)) \in o( U_{L}(n)) $. 
\end{thm}
\begin{proof}
	Suppose for a contradiction that there exists a weakly $ t(n) $ time-bounded $ G $-automaton $ \mathcal{E} $ recognizing $ {L}  $ in time $ t(n) $. For a sufficiently large $ n $, let $S $ be the set of uniformly $ n $-dissimilar strings such that $ |S|=U_{L}(n) $. For every string $ w_i \in S $, there exists a string $ v_i $ such that $ w_iv_i\in {L} $ and $ w_jv_i\notin {L} $ for all $ w_j \in S $ with $ i \neq j $ .
	
	Let $ S_{acc} $ be the set of accepted extended strings of the form $ w_iv_i \in {L} $ with $ |w_iv_i|\leq n $ where $ w_i \in S $ and $ w_jv_i \notin {L} $ for all $ w_j \in S $ with $ i\neq j $ and $ |w_jv_i|\leq n $. Let $ C $ be the set of $ t(n) $ time bounded accepting computation paths for the strings in $ S_{acc} $. The computation $ c_{w_iv_i } \in C $ on the string $ w_iv_i $  can be written as 
	\[ c_{w_iv_i} =c_{w_iv_i}^{w_i} c_{w_iv_i}^{v_i}\] where $ c_{w_iv_i}^{w_i} $ represents the computation up to the end of the prefix $ w_i $ and $ c_{w_iv_i}^{v_i} $ represents the rest of the computation on the string $ v_i $. 
	
	A configuration of a group automaton is a pair consisting of a state
	and a group element. Let us count the number of configurations that can be reached at the end of the computation $ c_{w_iv_i}^{w_i} $. Since the number of states is constant, the number of configurations that can be reached is dependent on the number of different group elements that can appear in
	the register. After reading a prefix $ w_i $ with $|w_i|= m\leq n $, the product of
	the labels on the edges can be given by $ l=g_{i_1}g_{i_2}\dots g_{i_{k}} $ for some $ k \leq t(m) $, since the computation in consideration is time bounded. $ l $ can be 
	expressed as a product of $ \kappa $ generators, where $ \kappa $ is at 
	most $C\cdot k $ for some constant $ C $, since each group
	element labeling a transition in $ \mathcal{E} $ is composed of at most some constant number of
	generators, which is independent of the length of the string. The number of 
	elements in $ G $ which can be represented as a product of at most $ \kappa$ 
	generators is given by $ g_G(\kappa) $ by the definition of the growth function 
	of $ G $. Hence, the number of different values that can appear in the register 
	after reading a string of length exactly $ m $ is less than or equal to 
	$ g_G(\kappa) $. Since  $ \kappa \leq  C\cdot k $ and $ k \leq t(m) $ and $g_G(t(n)) \in o( U_{L}(n)) $, we can 
	conclude that $$ g_G(\kappa)  \leq  g_G(C \cdot t(m)) \in o(U_{L}(n)). $$	
	
	Now it is easy to see that the number of different configurations that can be reached at the end of a computation $  c_{w_iv_i}^{w_i} $ is $  o(U_{L}(n)) $. Note that the cardinality of the set $ S $, and thus that of $ S_{acc} $, is equal to $ U_L(n) $. Due to the pigeonhole principle, the same configuration must be reached at the end of two computations $  c_{w_{i}v_i}^{w_{i}} $ and $  c_{w_jv_j}^{w_j} $ for some $ i\neq j $. This will result in the acceptance of the strings $ w_{i}v_j $ and $ w_{j}v_{i} $, which are not members of $ {L} $. We arrive at a contradiction and conclude that ${L} $ cannot be recognized by any weakly $ t(n) $ time-bounded $ G $-automaton.
\end{proof}  

In the next lemma, we set a lower bound on maximum cardinality of the set of uniformly $ n $-dissimilar strings in the word problem language of some group $ G $.

\begin{lem}\label{lemma: growth2}
	Let $ G $ be a finitely generated group with growth function $ g_G(n)
	$. Then $ U_{W(G)}(n)\geq g_G(\lfloor \frac{n}{2}\rfloor)$.
\end{lem}
\begin{proof}
	Let $ A $ be the generator set of $ G $. The number of distinct elements $ g $ in $ G $ which can be represented by a word of length less than or equal to
	$\lfloor \frac{n}{2}\rfloor $ is $ g_G(\lfloor \frac{n}{2}\rfloor)$, which is the cardinality of the set 
	$B_G^A(\lfloor \frac{n}{2}\rfloor)=\{g \in G,|g|_A\leq \lfloor \frac{n}{2}\rfloor\}$. Let $ T $ be the set containing the string representations of the elements in $ B_G^A(\lfloor \frac{n}{2}\rfloor) $. Every $ w_i \in T $ can be extended with $w_i^{-1}$ so that the extended string represents the identity element of $ G $ and has length less than or equal to $n $. 
	Since the strings in $ W(G) $ are those which belong to $ (A \cup A^{-1})^* $ and represent the identity element of $ G $, the extended string $ w_iw^{-1}_{i} \in W(G) $. For every string $ w_j\in T$ such that  $i \neq j$, $ w_jw_i^{-1} \notin W(G)$ since it is not possible for $ w_jw_i^{-1} $ to represent the identity element of $ G $. We conclude that the set $ S $ is uniformly $ n $-dissimilar. Since $ |T|=|B_G^A(\lfloor \frac{n}{2}\rfloor)| =  g_G(\lfloor \frac{n}{2}\rfloor)  $, it follows that $ U_{W(G)}(n)\geq g_G(\lfloor \frac{n}{2}\rfloor)$.
\end{proof}

The following theorem is about the language recognition power of finite automata over polynomial-growth groups which are weakly polynomial time-bounded.
\begin{thm}
	Let $ G $ and $ H $ be groups with polynomial and exponential growth
	functions $ g_G(n)$ and $  g_H(n) $, respectively. For any polynomial $ t(n) $, $\mathfrak{L}(H)
	\nsubseteq \mathfrak{L}(G)_{t(n)}^w  $. 
\end{thm}
\begin{proof}
	Since $ U_{W(H)}(n) \geq g_H(\lfloor \frac{n}{2}\rfloor) $ by Lemma $ \ref{lemma: growth2} $, and $ g_H(n) $ is an exponential function, $ U_{W(H)}(n)$ is also at least exponential.  $g_G(t(n))  $ is a polynomial function, since both $ g_G(n)$ and $ t(n) $  are polynomial. Hence, $ W(H) \notin \mathfrak{L}(G)_{t(n)}^w $ by Theorem \ref{thm: growth2}, and the result follows since $ W(H) $ is trivially in $ \mathfrak{L}(H) $.
\end{proof}

\begin{thm}\label{thm: polycf}
	Let $ G $ be a group with a polynomial growth function. For any polynomial $ t(n) $,
	$ \mathsf{CF} \nsubseteq \mathfrak{L}(G)_{t(n)}^w $.
\end{thm}
\begin{proof}
	It is known that the word problem of the free group of rank 2, $
	W(\mathbf{F}_2) $, has an exponential growth function \cite{Gi90}. Assuming that  $
	G $ is a group with polynomial growth function, $ W(\mathbf{F}_2) $
	cannot be recognized by any weakly $ t(n) $ time-bounded $ G $-automaton by
	Theorem \ref{thm: growth2}. Since $ W(\mathbf{F}_2) $ is a
	context-free language, the proof is complete.
\end{proof}

\subsection{Group Automata Under Linear Time Bounds}\label{sec: linear}

Having discussed methods for proving that certain languages can not be recognized by group automata under time restrictions, in this section will we focus on linear-time computation. 

Let $ X $ be a finite alphabet and let $X^*$ be the free monoid of words over $X$. For each symbol $ x\in X $, let $P_x$ and $Q_x$ be functions from $X^*$ into $X^*$
defined as follows:  for every $u\in X^*$,
\begin{align*}
P_x(u)= ux,\quad %P_x:&~ X^* \rightarrow X^* ~~~~~w \mapsto wx \\
Q_x(ux)=u.%Q_x:&~ X^*x \rightarrow X^* ~~~~ wx \mapsto w.
\end{align*}
Note that $Q_x$ is a partial function from $X^*$ into $X^*$ whose domain is the language $X^*x$. The submonoid of the monoid of all partial functions on $X^*$ generated by the
set of functions $\{P_x,\, Q_x \ | x \in X   \}$ turns out to be an inverse monoid, denoted by $P(X)$, called the \textit{polycyclic monoid on} $X$. Polycyclic monoids were explicitly studied by Nivat and Perrot in \cite{NP70} and have several applications in formal language theory and, in particular, 
define an interesting storage model of computation for the recognition of formal languages \cite{Co05,Gi96,Ka09,Ni70,NP70}. 

For any element $ x \in X, $ $P_xQ_x=1 $ where $ 1 $ is the identity element of $ P(X) $ and for any two distinct elements $ x, y \in X $, $ P_xQ_y $ is the empty partial function which represents the zero element of $ P(X) $. The partial functions $ \{P_x,Q_x\} $ model the operation of pushing and popping $ x$ in a PDA, respectively. In order to model popping and pushing the empty string, let us define $ P_{\varepsilon}$ and $Q_{\varepsilon}$ as $ P_{\varepsilon}=Q_{\varepsilon}=1 $. The equivalence between PDA with stack alphabet $ X $ and $ P(X) $-automata is due to the nature of the functions $ P_x $ and $ Q_x $, and investigated in various papers \cite{Co05,Gi96,Ka09}. The resemblance between the free group and $ P(X) $ is used to prove that $ \mathfrak{L}(\mathbf{F}_2)=\mathsf{CF} $ in  \cite{Co05} and \cite{Ka09}.

Our aim is to show that $ \mathbf{F}_2 $-automata working in linear time can recognize all context-free languages. It is stated in \cite{Ze13} that $P(X)$-automata which consume at least one input symbol at each step are as powerful as $P(X)$-automata without any time bound. However, it is not straightforward to see whether the same is true for $ \mathbf{F}_2 $-automata.  

\begin{thm}\label{thm: rtF2}
	$ \mathfrak{L}(\mathbf{F}_2)_{O(n)}^w =  \mathsf{CF}$.
\end{thm}
\begin{proof}
	We are going to use the construction of Kambites \cite{Ka09} to prove that any context-free language can be recognized by a weakly linear-time bounded $ \mathbf{F}_2 $-automaton.  
	
	Let $ {L} $ be a context-free language and let $ \mathcal{E}=\{Q,\Sigma, P(X), \delta,q_1,Q_a\} $ be a polycyclic monoid automaton recognizing $ {L} $. $ P(X) $ is the polycyclic monoid on $ X $ where the cardinality of the set $ X $ is $ n $ for some $ n \geq 2 $. The construction of Kambites provides an $ \mathbf{F}_{n+1}$-automaton $ \mathcal{E}'=\{Q',\Sigma,\mathbf{F}_{n+1}, \delta',q_{1}',Q_a'\} $ recognizing the language $ {L} $. The generator set for $ \mathbf{F}_{n+1} $ is $ X' $, where $ X'=X\cup \# $. 
	
	Let us analyze the construction in more detail.
	
	\begin{itemize}
		\item $ Q'=Q_{-} \cup Q_+ $ where $ Q_{-}=\{q_-|q \in Q\} $ and  $ Q_{+}=\{q_+|q \in Q\} $ 
		\item $ q_1' $=$ q_+ $ where $ q=q_1 $. 
		\item $ Q_a'=\{q_-|q\in Q_a \} $.
		\item$  \delta'(p_+,\sigma)=(q_+,x\#) $ if $ \delta(p,\sigma)=(q,x\#) $ where $ x $ is a positive generator for all $ \sigma \in \Sigma $.
		\item$ \delta'(p_-,\sigma)=(q_+,x'\#) $ if $ \delta(p,\sigma)=(q,x'\#) $ where $ x' $ is a negative generator for all $ \sigma \in \Sigma $.
		\item $ \delta'(p_+,\sigma)=(q_+,1) $ if $ \delta(p,\sigma)=(q,1) $ for all $ \sigma \in \Sigma $.
		\item $ \delta'(q_+,\varepsilon)=(q_-,1) $ for each $ q\in Q $.
		\item $ \delta'(q_-,\varepsilon)=(q_-,\#^{-1}) $ for each $ q\in Q $.
	\end{itemize}
	
	We  will prove that $ \mathcal{E}' $ actually runs in linear time. There are two transitions where the automaton is allowed to move without consuming any input symbols.
	
	For each state $q\in Q $, there are two states $ q_+ $ and $ q_- $ in $ \mathcal{E}'  $ which are connected with an edge labeled $ (\varepsilon, 1) $. These transitions do not change the register value, and cannot contribute more than half of the runtime of the machine, since at least one input symbol has to be consumed between any two executions of such transitions.
	
	$ \varepsilon $-loops exist in the machine $ \mathcal{E}' $ for each state $ q_- $ where the loop is labeled by $ (\varepsilon, \#^{-1}) $. Although this looks worrisome at first for the purpose of bounding the runtime, the number of times these loops are traversed is actually bounded, as the following argument suggests. Suppose that the register is multiplied with $ l_1 $, $ l_2 $, $ \cdots $, $ l_m $ while reading some input string $ w$ of length $ n $, resulting in the register value $l=l_1l_2\cdots l_m(\#^{-1})^k$, where $ k \in \mathbb{N} $, at the end of the computation. If  $w $ is accepted by the machine, $ l $ should satisfy the following, as well as being equal to the identity element:
	\[
	l_i = \Biggl\{\begin{array}{lr}
	(\#^{-1})^px_i\# \mbox{ for some } p \in \mathbb{N}, & \mbox{if $ x_i $ is a negative generator}\\
	x_i\# , & \mbox{if $ x_i $ is a positive generator}\\
	\end{array}
	\]

	This is called a \textit{permissible padding} in \cite{Ka09}. By looking at the transition function of $ \mathcal{E}' $, one can see that the register is multiplied by a $ \# $ only when an input symbol is consumed. Hence, the number of $ \# $'s that occur in $ l $ is less than or equal to the length of the string. The register is multiplied with $ \#^{-1} $ without consuming any input symbol. In order for the $ \# $'s and $ \#^{-1} $'s to cancel each other, they should be equal in number. Therefore, it can be concluded that the $ \varepsilon $-loops are traversed at most $ n $ times.

	We can conclude that any context-free language can be recognized by a weakly linear-time bounded free group automaton. Since $ \mathbf{F}_2 $ contains every free group of countable rank, the proof is complete.
\end{proof}	

We state the following theorem, which is the linear-time equivalent of Fact \ref{fact: finite} \cite{Co05}.

\begin{thm}\label{thm: rtindex}
	Suppose $ G $ is a finitely generated group and $ H $ is a subgroup of finite index. Then $ \mathfrak{L}(G)_{O(n)}^w = \mathfrak{L}(H)_{O(n)}^w	 $.
\end{thm}
\begin{proof}
	We know that the statement is true in general when there is no time bound by \cite{Co05}. The proof in \cite{Co05} still works when all automata in the constructions are required to work in linear time.
\end{proof}

Now we can show that Theorem \ref{theorem:gl} also holds for linear-time bounded group automaton.

\begin{thm}\label{thm: cfrt}
	$  \mathsf{CF}= \mathfrak{L}(\mathbf{F}_2)_{O(n)}^w= \mathfrak{L}(SL(2,\mathbb{Z}))_{O(n)}^w=\mathfrak{L}(GL(2,\mathbb{Z}))_{O(n)}^w  $.
\end{thm}
\begin{proof}
	The proof is identical with the proof of Theorem \ref{theorem:gl} by using Theorem \ref{thm: rtindex}.
\end{proof}

By using the results proven in Subsection \ref{section: mainthm}, we can demonstrate the language recognition power of weakly linear-time bounded $ \mathbf{H} $-automata.

\begin{thm}
	$ \mathfrak{L}(\textbf{H})_{O(n)}^w \subsetneq  \mathfrak{L}(SL(3,\mathbb{Z}))_{O(n)}^w $.
\end{thm}		

\begin{proof}
	$ \mathfrak{L}(\mathbf{H})^w_{O(n)} \subseteq \mathfrak{L}(SL(3,\mathbb{Z}))_{O(n)}^w $ since $ \mathbf{H} $ is a subgroup of $ SL(3,\mathbb{Z}) $.  
	Since the Heisenberg group has polynomial growth function \cite{Ha00}, there exists a context-free language which cannot be recognized by any $ \textbf{H} $-automaton in polynomial time by Theorem  \ref{thm: polycf}. 	Since $ \mathsf{CF}=\mathfrak{L}(SL(2,\mathbb{Z}))_{O(n)}^w$ by Theorem \ref{thm: cfrt}, the result follows. 
\end{proof}

\begin{thm}\begin{enumerate}
		\item For $ k\geq 5 $, $ \mathfrak{L}(\mathbf{H})^w_{O(n)} $ and $ \mathfrak{L}(\mathbb{Z}^k)^w_{O(n)}$  are incomparable.
		\item $ \mathfrak{L}(\mathbf{H})^w_{O(n)} $  and $ \mathsf{CF} $ are incomparable.
	\end{enumerate}	
\end{thm}
\begin{proof}
	\textit{ i.} 	In \cite{Re10}, a weakly linear-time bounded $ \mathbf{H} $-automaton which recognizes the language $\mathtt{MULT}=\{x^py^qz^{pq} | p,q\geq 0\}$ is constructed. The language $\mathtt{MULT}$ cannot be recognized by any $ \mathbb{Z}^k $-automaton, since any bounded language in $ \mathfrak{L}(\mathbb{Q}^+) $ is semilinear by Fact \ref{fact: bounded}. 
	
	In \cite{GS98}, it is implicitly proven there exists a uniformly $ n $-dissimilar set of size $ \Theta(n^k) $ for the language $ {L}_k=\{0^{a_1}10^{a_2}1\dots 0^{a_k} 10^{a_1}10^{a_2}1\dots 0^{a_k} 1\}  $ for all integers $ k $. For $k=5 $, there exists a uniformly $ n $-dissimilar set of size $ \Theta(n^5) $ for the language $ {L}_5$ and $ U_{{L}_5}(n) \geq n^5 $. Since $ g_{\mathbf{H}}(n)$ is a polynomial of order 4 \cite{Ha00} and $ t(n)=O(n) $, $ g_{\mathbf{H}}(t(n)) \in o(U_{{L}_5}(n) )$. By Theorem \ref{thm: growth2}, we conclude the result.

	\noindent $ ii. $ The language $\mathtt{MULT}=\{x^py^qz^{pq} | p,q\geq 0\}$ is not a context-free language. Since $ \mathbf{H} $ has a polynomial growth function \cite{Ha00}, there exists a context-free language which cannot be recognized by any $ \mathbf{H} $-automaton in polynomial-time by Theorem \ref{thm: polycf}.
\end{proof}

Let us note that $ {L}_5 $ can be recognized by a $ \mathbb{Z}^5 $-automaton in real time. The existence of the languages $ {L}_k $ can be used to prove the linear-time nondeterministic counter hierarchy, with the help of Theorem \ref{thm: growth2}.

\begin{thm}
	$ \mathfrak{L}(\mathbb{Z}^k )^w_{O(n)} \subsetneq \mathfrak{L}(\mathbb{Z}^{k+1} )^w_{O(n)} 	 $ for $ k\geq 1 $.
\end{thm}
\begin{proof}
	The language $ {L}_{k+1}=\{0^{a_1}10^{a_2}1\dots 0^{a_{k+1}} 10^{a_1}10^{a_2}1\dots 0^{a_{k+1}} 1\}  $	can be recognized by a $ \mathbb{Z}^{k+1}  $-automaton in real time. While scanning the first $ k+1 $ segments of $ 0 $'s, the $ i $'th counter is increased for each scanned $ 0 $ as $ 0^{a_i} $ is read. In the remainder of the computation, the $ i $'th counter is decreased for each scanned $ 0 $ when $ 0^{a_i} $ is read.
	
	There exists a uniformly $ n $-dissimilar set of size $ \Theta(n^{k+1}) $ for the language $ {L}_{k+1} $, so $ U_{{L}_{k+1}}(n) \geq n^{k+1} $. Since $ t(n)= O(n) $ and $ g_{\mathbb{Z}^{k}}(n) = n^k$ \cite{Gi90}, $ g_{\mathbb{Z}^{k}}(t(n)) \in o(U_{{L}_5}(n) )$. We conclude by Theorem \ref{thm: growth2}.
\end{proof}

A celebrated result of the field of computational complexity, the nondeterministic time hierarchy theorem, will enable us to demonstrate that the computational power $ \mathbf{F}_2 \times \mathbf{F}_2$-automata is dependent on the time allotted for their execution. 

\begin{fact}\label{fact: nth} \textup{\cite{Zak83}}		
	If $g(n)$ is a time-constructible function, and $f(n+1) = o(g(n))$, then there exists a language which cannot be recognized by any nondeterministic Turing machine in time $f(n)$, but can be recognized by a nondeterministic Turing machine in time $g(n)$.
\end{fact}

Assume that any recursively enumerable language can be recognized by some linear-time $ \mathbf{F}_2 \times \mathbf{F}_2$-automaton. One can easily build a nondeterministic Turing machine that simulates such a  $ \mathbf{F}_2 \times \mathbf{F}_2$-automaton with only a polynomial slowdown. But this would mean that any recursively enumerable language can be recognized by  some nondeterministic TM in polynomial time, contradicting Fact \ref{fact: nth}, which implies that there exist languages which can only be recognized by nondeterministic Turing machines which run in at least exponential time. We have proven the following theorem.

\begin{thm}
	$ 	\mathfrak{L}(\mathbf{F}_2 \times \mathbf{F}_2)_{O(n)}^w  \subsetneq \mathsf{RE} $.
\end{thm}

Using the ability of Turing machines to simulate any finite automaton over a computable matrix group, the statement of the above theorem can be extended as follows.

\begin{thm}
	$ \mathfrak{L}(G)_{O(n)}^w  \subsetneq \mathsf{RE} $ for any matrix group $ G $ whose matrix entries are computable numbers.
\end{thm}
\begin{proof}
	In Theorem \ref{thm: RE}, we have mentioned that Turing machines can simulate any finite automaton over a computable matrix group. By the nondeterministic time hierarchy theorem, it can be shown that there exist some languages which cannot be recognized by any finite automata over matrix groups in linear time. 
\end{proof}

\begin{thm}
	$ 	\mathfrak{L}(\mathbf{F}_2)_{O(n)}^w \subsetneq	\mathfrak{L}(\mathbf{F}_2 \times \mathbf{F}_2)_{O(n)}^w  $.
\end{thm}
\begin{proof}
	It is obvious that an $ \mathbf{F}_2 $-automaton can be simulated by an $ \mathbf{F}_2 \times \mathbf{F}_2 $-automaton. $ \mathfrak{L}(\mathbf{F}_2)_{O(n)}^w = \mathsf{CF} $ by Theorem \ref{thm: cfrt}. The inclusion is proper since the non-context-free language $ {L}=\{a^nb^nc^n|n\geq 0\} $ can be recognized by an $ \mathbf{F}_2 \times \mathbf{F}_2 $-automaton in real time by using the two registers as two counters. 
\end{proof}

In the rest of the section, the linear-time counterparts of the relationships in Figure \ref{fig: diagram} will be stated.

\begin{thm}
	\begin{enumerate}
		\item $ \mathfrak{L}(\mathbb{Q}^+)^w_{O(n)} \subsetneq \mathfrak{L}(SL(2,\mathbb{Q}))^w_{O(n)} $.
		\item $ \mathfrak{L}(\mathbb{Z})^w_{O(n)} \subsetneq \mathfrak{L}(BS(1,2))^w_{O(n)} \nsubseteq \mathsf{CF}$.
		\item $ \mathfrak{L}(SL(2,\mathbb{Z}))^w_{O(n)} \subsetneq  \mathfrak{L}(SL(3,\mathbb{Z}))^w_{O(n)}$.
		\item $\mathfrak{L}(\mathbb{Z}^2)^w_{O(n)} \subsetneq \mathfrak{L}(\mathbf{H})^w_{O(n)}$.
		\item 	$\mathsf{CF}$ and $\mathfrak{L}(\mathbb{Z}^k)^w_{O(n)}$ are incomparable for all $ k \geq 2 $.
		\item $ \mathfrak{L}(SL(3,\mathbb{Z}))^w_{O(n)} = \mathfrak{L}(GL(3,\mathbb{Z}))^w_{O(n)}$.
		\item 	$ \mathsf{REG} = \mathfrak{L}(\mathbf{F}_0)^w_{O(n)} \subsetneq
		\mathfrak{L}(\mathbf{F}_1)^w_{O(n)} = \mathfrak{L}(\mathbb{Z})^w_{O(n)} \subsetneq
		\mathfrak{L}(\mathbf{F}_2)^w_{O(n)} $.
		
	\end{enumerate}
\end{thm}
\begin{proof}
		(i,ii,iii,iv) Analogous results where no time bound was imposed on the machines were proven in Theorems \ref{thm: qsl2q}, \ref{theorem: upow}, \ref{thm: sl2z3z}, and \ref{theorem: Z2H}, respectively. The group automata recognizing 
	the witness languages ${L}=\{a^{2^{2n+1}} | n \geq 0 \}$,   $\mathtt{UPOW}=\{a^{2^n}|n\geq 0\}$ and $\mathtt{MULT}=\{x^py^qz^{pq} | p,q\geq 0\}$ operate in weakly linear time in all cases. 
	
	\begin{enumerate}\addtocounter{enumi}{4}
	 
		\item The equivalent result for the general case is given in Fact \ref{fact: cfzn}. The non-context-free language $ {L}'=\{a^nb^nc^n|n\geq 0\} $ can be recognized by a $ \mathbb{Z}^2 $-automaton in real time. 
		
		\item The equivalent result for the general case is given in Theorem \ref{thm: sl3zgl3z}. The result follows by Theorem \ref{thm: rtindex}.
		
		\item The equivalent result for the general case is given in Fact \ref{fact: reg}. $ \mathbf{F}_0 $ is the trivial group, and any regular language can be recognized by a deterministic finite automaton, which can be seen as finite automaton over $ \mathbf{F}_0 $, in real time. Since $ \mathbf{F}_1 $ is isomorophic to $ \mathbb{Z} $, the equality is obvious. Since the nonregular language $ {L}=\{a^nb^n|n\geq 0\}$ can be recognized by a $ \mathbb{Z} $-automaton in real time, the proper inclusion follows. Lastly, since $ \mathfrak{L}(\mathbf{F}_2)^w_{O(n)}  $ is equivalent to $ \mathsf{CF} $ by Theorem \ref{thm: cfrt}, the last proper inclusion is still valid.
	\end{enumerate}

\end{proof}

The results are summarized in Figure \ref{fig:diagram-time}. 

\begin{figure}[h!]
	\centering
	\includegraphics[width=1\linewidth]{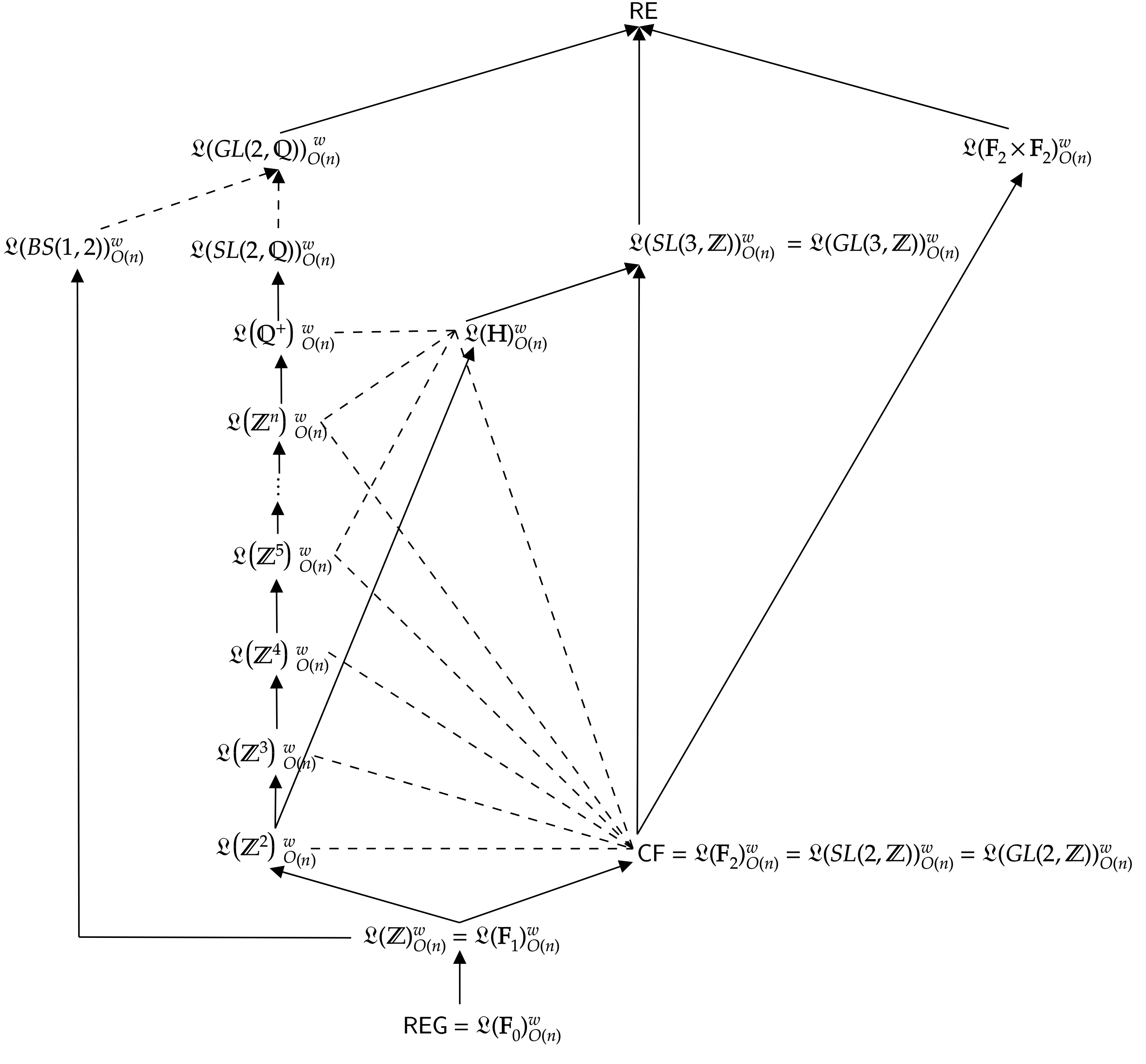}
	\caption{Language classes recognized by weakly linear-time bounded group automata }
	\label{fig:diagram-time}
\end{figure}

\section{Decision Problems for Matrix Semigroups} \label{sec: efadecision}

So far, we have focused on the language recognition power of extended finite automata over matrix groups. In this section, our aim is to make a connection between the theory of extended finite automata and the decision problems for matrix semigroups. Matrices play an important role in various areas of computation, which makes it interesting to study decision problems on matrices. Even for integer matrices of low dimension, many decision problems become non-trivial for finitely generated infinite semigroups.

For our purposes, we define $ S $-automata or extended finite automata over semigroups, generalizing the notion of $ M $-automata from monoids to semigroups. The emptiness problem is defined as the problem of deciding whether a given machine accepts any string. By using the decidability of the emptiness problem of the corresponding extended finite automata, we provide an alternative proof for the decidability of subsemigroup membership problem for $ GL(2,\mathbb{Z}) $ and the decidability of the identity problem for $ M_2(\mathbb{Z}) $. We show that the emptiness and universe problems for extended finite automata over $ SL(4,\mathbb{Z}) $ are undecidable, using the fact that the subgroup membership problem for $ SL(4,\mathbb{Z}) $ is undecidable. We also prove some results on the the decidability of the universe problem for extended finite automata, the problem of deciding whether a given machine accepts every string.

\subsection{Background}

Before proceeding to discuss our results, it is necessary to talk about some key definitions and previous studies on the decidability problems for matrix semigroups.

In the following sequel, let $ G  $ be a finitely generated group. 

Decidability is one of the popular topics of combinatorial group theory. In 1911, Dehn proposed several decision problems for groups including the famous word problem \cite{De11}. Let us recall the definition for the word problem of a group $ G $.

\textit{Word problem for $ G $:}  Given an element $ g \in G $, the problem is to decide whether $ g $ represents the identity element.

Explicitly introduced for the first time by Mikhailova \cite{Mi58}, a generalization of the word problem which is also known as the \textit{generalized word problem} is the following:

\textit{Subgroup membership problem for $ G $:} Given $ \{g_1,g_2,\dots,g_n\} \in G $ and an element $ g \in G $, the problem is to decide whether $ g $ belongs to the subgroup generated by the elements $ g_1, g_2, \dots,g_n $.

In \cite{Mi58}, it is proven that subgroup membership problem for $ \mathbf{F}_2 \times \mathbf{F}_2 $ is undecidable, which yields the undecidability of the subgroup membership problem for $ SL(4,\mathbb{Z}) $.

One can also consider \textit{submonoid membership problem for group $ G $} and \textit{subsemigroup membership problem for group $ G $}, in which case the problem is to decide whether $ g $ belongs to the submonoid and subsemigroup generated by $ \{g_1,g_2,\dots,g_n\} $, respectively.

A further generalization is the rational subset membership problem, as the notion of rational subset generalizes subgroups, submonoids and subsemigroups. 

\textit{Rational subset membership problem for $ G $}: Given a rational subset $ R $ of $ G $ specified using a finite automaton and $ g \in G $, the problem is to decide whether $ g $ belongs to $ R $.

Note that the hardness of the problems are in increasing order and the decidability of rational subset problem implies the decidability of the other problems. Similarly, the undecidability of the subgroup membership problem implies the undecidability of the remaining problems.

When we talk about the membership problems for matrices, it is more natural to consider matrix monoids or semigroups. Hence in the above definitions, we may replace the group $ G $ with a semigroup $ S $ or a monoid $ M $.

For matrices, the well studied decision problem is the subsemigroup membership problem. For $ 3 \times 3 $ matrices with integer entries, the subsemigroup membership problem for $ M_3(\mathbb{Z}) $ is known to be undecidable due to a result by Paterson \cite{Pa70}. The problem remains open for $ GL(3,\mathbb{Z}) $ and recently it has been proven that it is decidable for a subgroup of $ GL(3,\mathbb{Z}) $, the Discrete Heisenberg group \cite{COSW19}.  

An extensive study has been carried out for the matrices from $ M_2({\mathbb{Z}}) $. In \cite{CK05}, it is proven that subsemigroup membership problem for $ GL(2,\mathbb{Z}) $ is decidable. The result is extended by Potapov and Semukhin in \cite{PS17p} to subsemigroups of matrices from $ GL(2,\mathbb{Z}) $ extended by singular matrices and in \cite{PS17} to subsemigroups of nonsingular matrices from $ M_2({\mathbb{Z}}) $. The problem is still open for the general case of $M_2({\mathbb{Z}} ) $. 

When the subsemigroup membership problem is asked for the identity matrix, we obtain the identity problem.

\textit{Identity problem for $ S $:} Given matrices $ \{Y_1,Y_2, \dots ,Y_n\} \in S $, the problem is to decide whether the identity matrix $ I $ belongs to the semigroup generated by the elements $ Y_1, Y_2, \dots,Y_n $.

Note that the identity problem is a special case of the subsemigroup membership problem. The identity problem is also equivalent to the \textit{group problem}, i.e. the problem of deciding whether a subset of a given semigroup generates a nontrivial group.

The decidability of the identity problem for $ M_3(\mathbb{Z}) $ is still open. In \cite{KNP17}, it was proven that the identity problem for the discrete Heisenberg group $ \mathbf{H} $ is decidable (decidability of the membership problem for $ \mathbf{H} $ was unknown at the time). The decidability of the identity problem for $ SL(4,\mathbb{Z}) $, which was open for a long time, was established in \cite{BP10,KNP17}.

\subsection{$ S $-automaton}

An $ S $-automaton is an extended finite automaton where the group/monoid condition is loosened to a semigroup. In order to define the initialization and acceptance steps, we need an identity element. If $ S $ is a monoid or a group, then an identity element already exists and belongs to $ S $. Otherwise, we define 1 to be the identity element of $ S $.  Note that when $ S $ is not a monoid nor a group, then $ \mathcal{E} $ can accept only the empty string. Nevertheless, we define the concept of $ S $-automaton so that the machines in the proofs of Theorem \ref{theorem: ie} and \ref{theorem: iu} are constructed properly.

The \textit{Emptiness problem} for an automaton is the problem of deciding whether the language recognized by the machine is empty. \textit{Universe problem} is the problem of deciding whether the automaton accepts every string.

\subsection{Decidability of the Subsemigroup Membership Problem for $ GL(2,\mathbb{Z})$  }

It is proven that the subsemigroup membership problem for $ GL(2,\mathbb{Z})$ is decidable in \cite{CK05}. We are going to provide an alternative, automata theoretic proof. We will start by proving a series of lemmas.

For a finite index subgroup $ H $ of some finitely generated group $ G $, it is known that $ \mathfrak{L}(H)=\mathfrak{L}(G) $ by Fact \ref{fact: corson}. We will go over the proof details and use Fact \ref{fact: corson} to show that given a $ G $-automaton, one can construct an $ H $-automaton recognizing the same language.

\begin{lem}\label{lem: gh}
	Let $ G $ be a finitely generated group and let $ H $ be a subgroup of finite index. Any $ G $-automaton can be converted into an $ H $-automaton recognizing the same language.
\end{lem}
\begin{proof}
	Let $ A $ be the generator set for $ G $ and let $ X=A \cup A^{-1} $. Let $ \mathcal{E} $ be a $ G $-automaton recognizing language $ L $ over alphabet $ \Sigma $. Then there exists a rational subset $ R \subseteq \Sigma^* \times G $ such that $ L=\{w \in \Sigma^* |wR1 \} $. One can define the elements of $ G $ in terms of $ X $ to obtain a rational subset $ R_0 \subseteq \Sigma^* \times X^* $.
	
	Since $ H $ has finite index in $ G $, $ W(G) \in \mathfrak{L}(H)$ (\cite{Co05} Lemma 2.4). It follows that there exists a rational subset $ S \subseteq X^* \times H $ such that $ W(G)=\{w \in X^* | wS1\} $. 
	
	Then the composition $ R_0 \circ S $ is a rational subset of $ \Sigma^* \times H $ and it follows that $ L=\{w \in \Sigma^* | w(R_0 \circ S)1\} $ (\cite{Co05}, Theorem 3.1). The detailed construction of the finite automaton recognizing the composition is given in \cite{Gi96} (Theorem 5.3). Hence a finite automaton $\mathcal{ F} $ over $ \Sigma^* \times H $ recognizing $ L $ exists, from which an $ H $-automaton $ \mathcal{E}' $ recognizing $ L $ can be constructed.  
\end{proof}

The following construction of a pushdown automaton simulating an $ \mathbf{F}_2 $-automaton is left as an exercise in \cite{Ka06}. We present here some details of the construction.

\begin{lem}\label{lemma: f2pda}
	Any $ \mathbf{F}_2 $-automaton can be converted into a pushdown automaton recognizing the same language.  
\end{lem}
\begin{proof}
	Let $\mathcal{E} $ be an $ \mathbf{F}_2 $-automaton recognizing language $ L $ over $ \Sigma $ with the state set $ Q $ and let $ A=\{a,b\} $ be the generator set for $ \mathbf{F}_2 $. Let us construct a pushdown automaton $ \mathcal{A} $ recognizing the same language with the stack alphabet $ A $. Let $ (q',f) \in \delta (q,\sigma) $ be a transition of $ \mathcal{E} $ where $ q,q' \in Q $, $ \sigma \in \Sigma_{\varepsilon} $ and $ f \in \mathbf{F}_2 $ such that $$ f=f_1f_2\dots f_n \mbox{ where } f_i \in X=A \cup A^{-1} \mbox{ for } i=1\dots n. $$ In $ \mathcal{A}, $ we need additional $ n$ states $ q_1\dots q_{n} \notin Q$ to mimic each given transition of $ \mathcal{E} $. If $ f_i =a  $ or $ f_i =b $, then this corresponds to pushing $ a $ or $ b $ to the stack, respectively. Similarly, if $ f_i =a^{-1}  $ or $ f_i =b^{-1} $, then $ \mathcal{A} $ pops $ a $ or $ b $ from the stack. Each single transition of $ \mathcal{E} $ is accomplished by the pushdown automaton $ \mathcal{A} $ by going through the additional states and pushing and popping symbols. Initially, the register of $ \mathcal{E} $ is initialized with the identity element of $ \mathbf{F}_2 $, which corresponds to the stack of $ \mathcal{A} $ being empty. The acceptance condition of $ \mathcal{E} $, which is ending in an accept state with the register being equal to the identity element is realized in $\mathcal{A} $ by starting with an empty stack and accepting with an empty stack in an accept state. We conclude that $\mathcal{A}$ recognizes language $ L $.   
\end{proof}

\begin{lem}\label{lem: gl2z}
	The emptiness problem for $ GL(2,\mathbb{Z}) $-automaton is decidable. 
\end{lem}
\begin{proof}
	Let $ \mathcal{E} $ be a $ GL(2,\mathbb{Z}) $-automaton. Since $ \mathbf{F}_2 $ has finite index in $ GL(2,\mathbb{Z}) $, one can construct an $ \mathbf{F}_2 $-automaton recognizing $ L(\mathcal{E}) $ by Lemma $ \ref{lem: gh} $. The $ \mathbf{F}_2 $-automaton can be converted to a pushdown automaton $ \mathcal{A} $ using the procedure described in Lemma \ref{lemma: f2pda}. Since the emptiness problem for pushdown automata is known to be decidable, we conclude that the emptiness problem for $ GL(2,\mathbb{Z}) $-automata is also decidable since any $ GL(2,\mathbb{Z}) $-automaton can be converted to a pushdown automaton.
\end{proof}

Now we prove the main theorem of the section and establish the connection between the emptiness problem for $ G $-automata and the subsemigroup membership problem for $ G $.

\begin{thm}\label{theorem: me}
	Let $ G $ be a finitely generated group. If the emptiness problem for $ G $-automata is decidable, then the subsemigroup membership problem for $ G $ is decidable.
\end{thm}
\newpage
\begin{proof}
	Let $ H $ be a finitely generated subsemigroup of $ G $ and let $ g \in G $ be given. We are going to construct a $ G $-automaton $ \mathcal{E}_1 $ and show that $g \in H $ iff $ L(\mathcal{E}_1) $ is nonempty. The state transition diagram of $ \mathcal{E}_1 $ is given in Figure $ \ref{fig:me} $. The transition labeled by $ (a,h_i) $ stands for each one of the transitions that multiply the register with $ h_i $ for $ i=1\dots n $.
	
		$ \mathcal{E}_1 $ has two states: the initial state $ q_1 $ and the accept state $ q_2 $. The transition function of $ \mathcal{E}_1 $ is defined as $ \delta(q_1,a)=(q_2,g^{-1}) $ and $ \delta(q_2,a)=(q_2,h_i) $ for each $ i=1\dots n $, where the set $\{h_1,\dots,h_n\}  $ generates $ H $.

	\begin{figure}[h]
		\centering
				\includegraphics[width=1\linewidth]{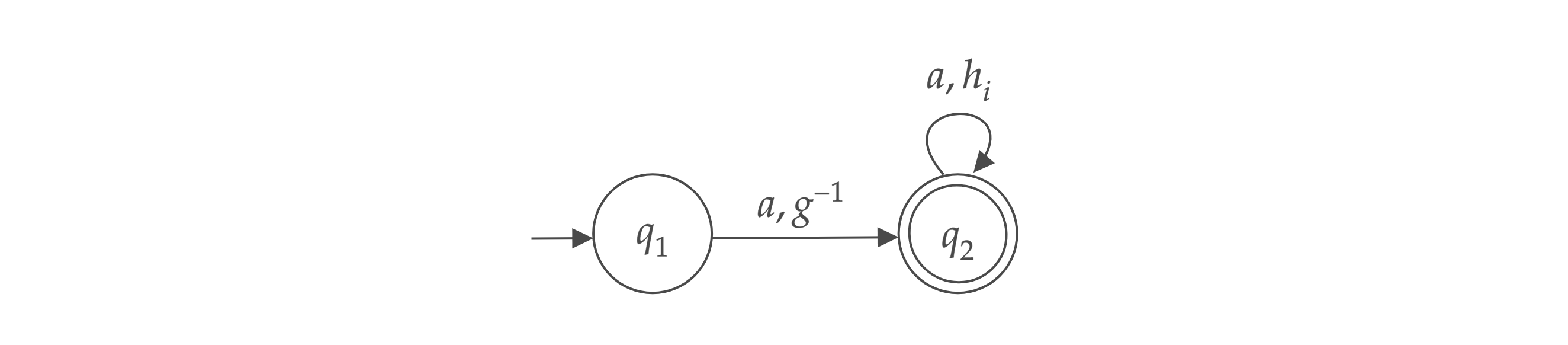}

		\caption{State transition diagram of $ \mathcal{E}_1 $}
		\label{fig:me}
	\end{figure}
	
	Note that $ g^{-1} $ exists and belongs to $ G $ since $ G $ is a group. If $ g \in H $, then there exists an integer $ k \geq 1 $ and $ i_1, i_2, \dots , i_k \in \{1, \dots , n\} $ such that $ h_{i_1} h_{i_2} \cdots h_{i_k} = g $. The string $ a^{k+1} $ is accepted by $\mathcal{E}_1$ as the register is initially multiplied by $ g^{-1}$ and there exists a product of elements yielding $ g $, from which we can conclude that the identity element can be obtained through a series of transitions of the machine $ \mathcal{E}_1 $. Hence, we can conclude that $ L(\mathcal{E}_1) $ is nonempty. 
	
	For the other direction, assume that $ L(\mathcal{E}_1) $ is nonempty, which means that some input string is accepted by $ \mathcal{E}_1 $. Since the acceptance condition requires that the product of the elements multiplied by the register of $ \mathcal{E}_1 $ is equal to the identity element and the register is initially multiplied by $ g^{-1} $, we can conclude that $ H $ contains $ g$. 
	
	Now suppose that the emptiness problem for $ G $-automaton is decidable. Then one can check if $ g $ is an element of $ H $ by constructing $ \mathcal{E}_1 $ and checking if $ L(\mathcal{E}_1) $ is nonempty. Hence, the subsemigroup membership problem $ G $ is also decidable.   	
\end{proof}

\newpage 
\textbf{Remark:} After obtaining the results reported above, we learned that a stronger version of Theorem \ref{theorem: me} was already proven in \cite{KSS07} and addressed later in \cite{Lo13 ,Ze15}.

\begin{fact}\textup{\cite{KSS07}}
	Let $ G $ be a finitely generated group, and $ M $ a finitely generated monoid. Then the following are equivalent:
	\begin{enumerate}
		\item The rational subset problem for $ G\times M $ is decidable;
		 \item The membership is decidable for $ G $-automaton subsets of M.
	\end{enumerate}
\end{fact}

Now we are ready to state our main result.

\begin{thm}
	The subsemigroup membership problem for $ GL(2,\mathbb{Z}) $ is decidable.	
\end{thm}
\begin{proof}
	From Lemma $ \ref{lem: gl2z} $, the emptiness problem for $ GL(2,\mathbb{Z}) $-automaton is decidable. Then by Theorem \ref{theorem: me}, the result follows.
\end{proof}

Note that we cannot extend this result to $ M_2(\mathbb{Z}) $. Even though the emptiness problem for $ M_2(\mathbb{Z}) $-automata is decidable, the construction in Theorem $ \ref{theorem: me} $ works only for group automata.

\subsection{Decidability of the Identity Problem for $ M_2(\mathbb{Z}) $}

In this section we provide an alternative proof for the decidability of the identity problem for $ M_2(\mathbb{Z}) $, which was originally proven in \cite{CK05}.

We should first show that the emptiness problem for $ M_2(\mathbb{Z}) $-automaton is decidable.

\begin{lem}\label{lem: m2z}
	The emptiness problem for $ M_2(\mathbb{Z}) $-automaton is decidable.	
\end{lem}
\begin{proof}
	Let $ \mathcal{E} $ be an $ M_2(\mathbb{Z}) $-automaton. Any transition labeled by a matrix whose determinant is not equal to $ \pm 1 $ can be safely removed from $ \mathcal{E} $, since after multiplication with such a matrix, it is not possible that the register is equal to identity matrix again and what we obtain is a $ GL(2,\mathbb{Z}) $-automaton. The emptiness problem for $ GL(2,\mathbb{Z}) $-automata is decidable by Lemma \ref{lem: gl2z}.
	
\end{proof}

In the next theorem, we make a connection between the identity problem for a semigroup $ S $ and the emptiness problem for the corresponding $ S $-automaton.

\begin{thm}\label{theorem: ie}
	Let $ S $ be a finitely generated semigroup. If the emptiness problem for $ S $-automata is decidable, then the identity problem for $ S $ is decidable.
\end{thm}
\begin{proof}
	We are going to construct an $ S $-automaton $ \mathcal{E}_2 $ and show that $ S $ contains the identity element iff $ L(\mathcal{E}_2) $ is nonempty. The state transition diagram of $ \mathcal{E}_2 $ is given in Figure $ \ref{fig:ie} $. The transitions labeled by $ (a,s_i) $ stand for each one of the transitions that multiply the register with $ s_i $ for $ i=1\dots n $.
	
	 $ \mathcal{E}_2$ has two states: the initial state $ q_1 $ and the accept state $ q_2 $. The transition function of $ \mathcal{E}_2$ is defined as $ \delta(q_1,a)=(q_2,s_i) $ and $ \delta(q_2,a)=(q_2,s_i) $  for each $ i=1\dots n $ where $ \{s_1,s_2,\dots s_n \} $ is the generator set for $ S $. 
	
	\begin{figure}[h]
		\centering
			\includegraphics[width=1\linewidth]{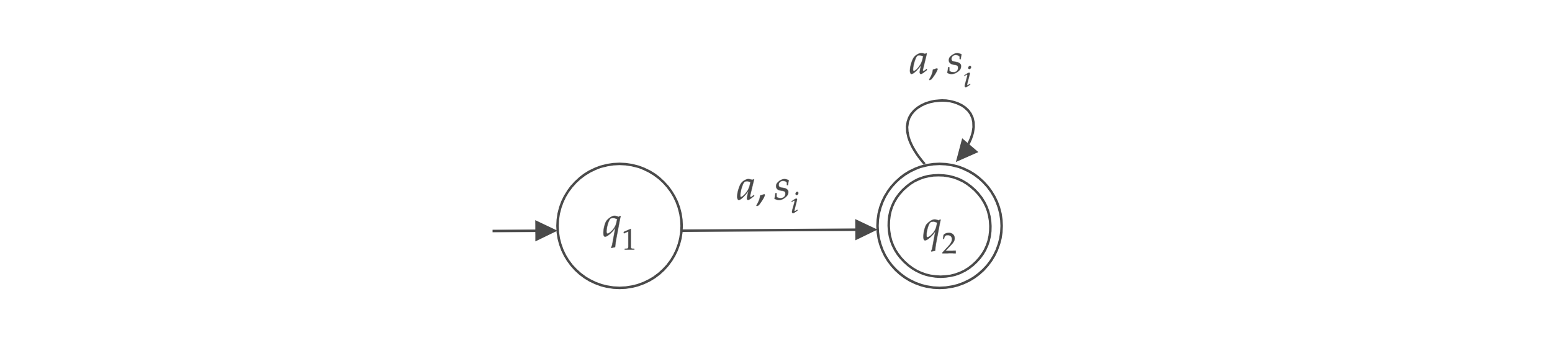}
	
		\caption{State transition diagram of $ \mathcal{E}_2 $}
		\label{fig:ie}
	\end{figure}
	
	If $ S $ contains the identity element, then there exists an integer $ k \geq 1 $ and $ i_1, i_2, \dots , i_k \in \{1, \dots , n\} $ such that $ s_{i_1} s_{i_2} \cdots s_{i_k} = 1 $. Then the string $ a^k $ is accepted by $ \mathcal{E}_2$ as there exists a product of elements yielding the identity element and this product can be obtained by a series of transitions. Hence, we can conclude that $ L(\mathcal{E}_2) $ is nonempty. For the converse, suppose that $ L(\mathcal{E}_2)$ is nonempty, which means that some input string is accepted by $ \mathcal{E}_2 $. Since the acceptance condition requires that the product of the elements multiplied by the register of $ \mathcal{E}_2 $ is equal to the identity element, we can conclude that $ S $ contains the identity element.    
	
	Now suppose that the emptiness problem for $ S $-automata is decidable. Then one can check if $ S $ contains the identity element by constructing $ \mathcal{E}_2 $ and checking if $ L(\mathcal{E}_2) $ is nonempty. Hence, the identity problem for $ S $ is also decidable.   	
\end{proof}

We connect the two results and state the following.

\begin{thm}
	The identity problem for $ M_2(\mathbb{Z}) $ is decidable.
\end{thm}
\begin{proof}
	By Lemma \ref{lem: m2z}, the emptiness problem for $ M_2(\mathbb{Z}) $-automaton is decidable. The result follows by Theorem \ref{theorem: ie}.
\end{proof}

\subsection{Undecidability of the Decision Problems for $SL(4,\mathbb{Z})$-automata}

In this section we are going to prove undecidability results for the emptiness and universe problems of $SL(4,\mathbb{Z})$-automata. 

To prove the undecidability of the universe problem for $ SL(4,\mathbb{Z}) $-automata, we first prove the following theorem.

\begin{thm}\label{theorem: mu}
	Let $ G $ be a finitely generated group. If the universe problem for $ G $-automata is decidable, then the subsemigroup membership problem for $ G $ is decidable.
\end{thm}
\begin{proof}
	We are going to construct a $ G $-automaton $ \mathcal{E}_3 $ such that $g \in H $ iff $ L(\mathcal{E}_3)=\Sigma^* $ where $ \Sigma=\{a\} $. $\{h_1,h_2,\dots h_n \} $ is the generator set for $ H $. The state transition diagram of $ \mathcal{E}_3  $ is given in Figure \ref{fig:mu}. The transition labeled by $ (a,h_i) $ and $ (\varepsilon,h_i) $ stands for each one of the transitions that multiply the register with $ h_i $ for $ i=1\dots n $.
	
	\begin{figure}[h]
		\centering

			\includegraphics[width=0.9\linewidth]{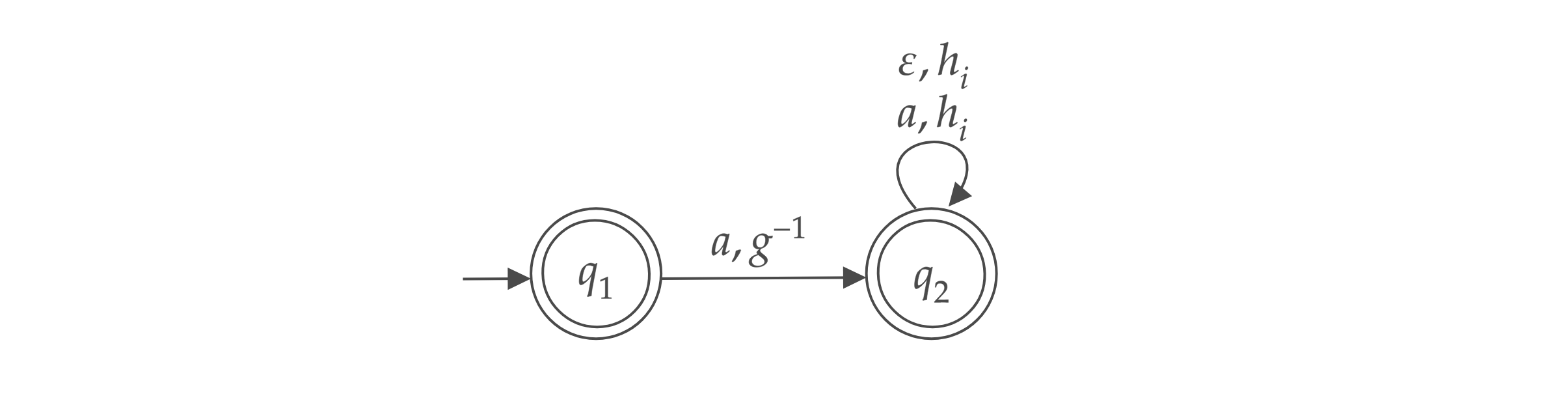}
		\caption{State transition diagram of $ \mathcal{E}_3 $}
		\label{fig:mu}
	\end{figure}
	
	The rest of the proof is similar to the proof of Theorem \ref{theorem: me} and omitted here.
\end{proof}

Now we state the undecidability of the emptiness and universe problems for $SL(4, \mathbb{Z})$-automata.

\begin{thm}\label{corollary: sl4z}
	Let $ S $ be a subsemigroup of $SL(4, \mathbb{Z})$. The emptiness and universe problems for $ S $-automata is undecidable.
\end{thm}
\begin{proof}
	Since the subgroup membership problem for $SL(4, \mathbb{Z})$ is undecidable \cite{Mi58}, by Theorem \ref{theorem: me} and by Theorem \ref{theorem: mu}, the emptiness and universe problems for $ SL(4,\mathbb{Z}) $-automata are undecidable.
\end{proof}

So far, we have established some connections between decision problems for groups and semigroups and the corresponding automata. Let us also state the following theorem, which links the identity problem for semigroups and the universe problem for the corresponding automata, for the sake of completeness.  

\begin{thm}\label{theorem: iu}
	Let $ S $ be a finitely generated semigroup. If the universe problem for $ S $-automata is decidable, then the identity problem for $ S $ is decidable. 
\end{thm}
\begin{proof}
	We are going to construct an $ S $-automaton $ \mathcal{E}_4 $ such that $ S $ contains the identity element iff $ L(\mathcal{E}_4)=\Sigma^* $ where $ \Sigma=\{a\} $. $ \{s_1,s_2,\dots s_n \} $ is the generator set for $ S $. The state transition diagram of $ \mathcal{E}_4$ is given in Figure $ \ref{fig:iu} $. The transition labeled by $ (a,s_i) $ and $ (\varepsilon,s_i) $ stands for each one of the transitions that multiply the register with $ s_i $ for $ i=1\dots n $.
	
	\begin{figure}[h]
		\centering
		
	\includegraphics[width=0.8\linewidth]{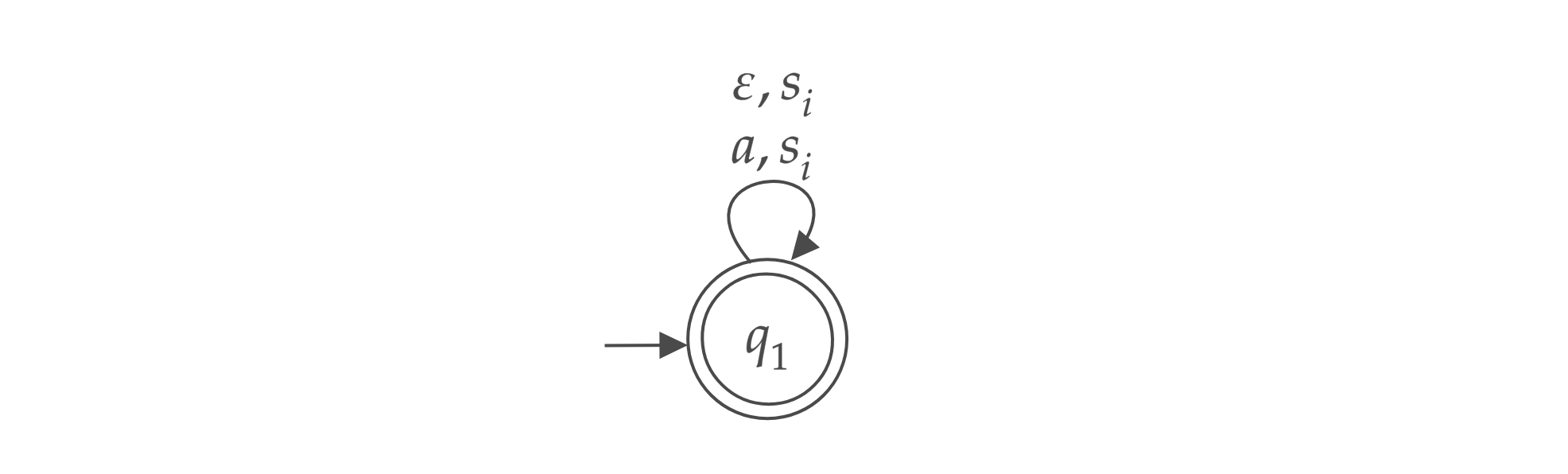}
		\caption{State transition diagram of $ \mathcal{E}_4 $}
		\label{fig:iu}
	\end{figure}
	The rest of the proof is similar to the proof of Theorem \ref{theorem: ie} and omitted here.
\end{proof}

Note that the converses of Theorem \ref{theorem: mu} and \ref{theorem: iu} are not true. For a given pushdown automaton, an $\mathbf{F}_2 $-automaton recognizing the same language can be constructed \cite{Ka06}. It is a well known fact that the universe problem for pushdown automata is undecidable, from which we can conclude that the universe problem for $\mathbf{F}_2 $-automaton is undecidable. On the other hand, $ \mathbf{F}_2 $ is a subgroup of $ SL(2,\mathbb{Z}) $ and the membership problem for $ SL(2,\mathbb{Z}) $ and thus the identity problem are known to be decidable \cite{KNP17}.

\section{Open Questions }\label{Section: open}

In this section, we are going to list some questions in need of further investigation.

\begin{itemize}
	\item Does there exist an $ SL(3,\mathbb{Z}) $-automaton recognizing $ W(\mathbb{Z}^3) $?
	Corollary 2 of \cite{CEO06} states that the word problem of a finitely generated 
	Abelian group $H$ is recognized by a $G$-automaton if and only if $H$ has a finite index 
	subgroup isomorphic to a subgroup of $G$. That corollary could be used to give an 
	affirmative answer to this open question. Unfortunately, the corollary is wrong: 
	Let $ H $ be an Abelian group and let $ G=\mathbf{F}_2 \times \mathbf{F}_2 $. 
	$ \mathfrak{L}(\mathbf{F}_2 \times \mathbf{F}_2 )$ contains the word problem of any 
	finitely generated Abelian group. Since $ \mathbf{F}_2 \times \mathbf{F}_2 $ is finitely 
	generated, any finite index subgroup of $ \mathbf{F}_2 \times \mathbf{F}_2 $ is also 
	finitely generated. Any finite index subgroup of $ \mathbf{F}_2 \times \mathbf{F}_2 $ is 
	either free or has a subgroup of finite index that is a direct product of free groups 
	\cite{BR84}. Any subgroup of an Abelian group is again Abelian. Hence, it is not 
	possible that $ G$ has a finite index subgroup isomorphic to a subgroup of $ H $.
	
	\item Can we describe the necessary properties of a group $ G $ so that 
	$ \mathfrak{L}(G) $ contains $ W(\textbf{F}_2) $?
	
	\item Little is known about $ BS(1,2) $-automata. Does $  \mathfrak{L}(BS(1,2)) $ 
	contain every context-free language?
	
\item Which, if any, of the subset relationships in Figure \ref{fig: diagram} are proper inclusions? 
	
	\item Can we add other classes above $ \mathsf{RE} $ in Figure \ref{fig: diagram} by examining groups on matrices with uncomputable entries?
	
	\item Theorem \ref{thm: growth2} uses the definition of uniform $ n $-dissimilarity requiring that $ g_G(t(n)) $  $\in o( U_{L}(n))  $. Would the theorem be still true if we replace $ U_{L}(n) $ by $ A_{L}(n) $ ? The gap between  $ U_{L}(n) $ and $ A_{L}(n) $ might be large as mentioned in \cite{GS98}. Consider the language $ {L}=\{a^ib^j|i\neq j\} $. It is stated in \cite{GS98} that a set of uniformly $ n $-dissimilar strings for $ {L} $ cannot contain more than two strings. However, $ A_{L}(n) \notin O(1) $, since $ {L} $ is not a regular language.
	
	\item Can real-time $ \mathbf{F}_2 $-automata recognize every context-free language?
	
		\item 	Can we prove a stronger version of Theorem \ref{thm: polycf}, which is independent 
	of the time component? For instance, for the case of $ \mathbf{F}_2 $, is it true that 
	$ W(\mathbf{F}_2) \notin \mathfrak{L}(\mathbf{H})$ in general?
	
	\item The decidability of the membership problem for $ GL(3,\mathbb{Z}) $ and the identity problem for $ M_3(\mathbb{Z}) $ are still open. We propose that investigating the decidability of the emptiness and universe problems for extended finite automata defined over $3 \times  3$ integer matrices is one possible way for obtaining results about the decision problems on these matrix semigroups.
\end{itemize}

\chapter{HOMING VECTOR AUTOMATA} \label{chap: hva}

The idea of augmenting the classical finite automaton model with an
external storage unit that can hold unlimited amounts of information,
yet can be accessed in a limited mode, is a celebrated topic of
automata theory. In this chapter we introduce homing vector automaton, a finite automaton equipped with a vector which can multiply its vector with an
appropriate matrix at each step and can check the entire vector for
equivalence to the initial value of the vector.

Matrices are fundamental objects in mathematics and computer science. They are also crucial in automata theory as many finite automaton models such as probabilistic and quantum can be simulated by vector matrix multiplications. Likewise, the vector matrix multiplication view of programming forms the basis of the computation process of homing vector automata. 

Homing vector automata are also closely linked to finite automata over matrix groups which we have discussed in Chapter \ref{chap: efa}. Although in both models the computation is carried out by a series of matrix multiplications, the nature of the registers and the acceptance conditions differentiate the two models. 

We examine homing vector automata under several different regimes, enabling us to determine the effect of definitional parameters such as whether the input is scanned in real-time or pausing the head on an input symbol for several steps is allowed, whether the machine can read its register during computation or is blind, with acceptance possible only if the register has returned to its initial value at the end, and whether nondeterminism confers any additional recognition power over deterministic programs.

Another way in which one can examine the nature of the computational power of homing vector automata is by examining models in which the matrices used at each step for transforming the vectors are restricted in some way.
Although the definition allows arbitrary rational matrices, one may constrain the matrix entries to belong to a particular set. In most automaton algorithms in this chapter, the entries of the matrices belong to the set $ \{-1,0,1\} $, as this basic set will be seen to already capture many capabilities of homing vector automata. Let us note that multiplications with matrices whose entries belong to this set can be used to perform additions, subtractions, resets, and swaps between the vector entries. It is possible to recognize some of the languages in the following discussion with homing vector automata of lower dimension when a larger set of matrix entries is allowed. 
%Some related open questions can be found in Section \ref{sec: hva-end}.

The rest of this chapter is structured in the following way:

In Section \ref{sec: hva-def}, we introduce homing vector automaton, giving the definitions for one-way, real-time, blind and non-blind versions. We begin with some observations about homing vector automata in Section \ref{sec: hva-obs}. A method we use for encoding strings on an alphabet of arbitrary size in a blind homing vector automaton, based on Stern-Brocot tree
\cite{St58,Br61}, may be of independent interest and is presented in Section \ref{sec: hva-SB}. In Section \ref{sec: hva-ca}, we investigate the relationship between counter automata and HVAs. In Section \ref{sec: hva-efa}, we establish a connection between the nondeterministic one-way blind version
of the HVA model and the extended finite automata, and use this
link to prove that these machines can recognize any Turing
recognizable language, even when the vector dimension is restricted to
four. We then focus on HVAs with real-time access to their input in Section \ref{sec: hva-rt}. We analyze the relationships between different versions of HVAs, present some closure properties, and analyze their stateless versions. Section \ref{sec: hva-end} lists some open questions.

\section{Definitions}\label{sec: hva-def}

Generalizing the idea of finite automaton equipped with a register, we
have previously introduced in \cite{SYS13} the vector automaton, a finite automaton which is endowed with a
vector, and which can multiply this vector with an appropriate matrix
at each step. We give the definition for the real-time deterministic version.
\newpage
A \textit{ real-time $k$-dimensional deterministic vector automaton} \textup{(DVA($k$))}
	\textup{\cite{SYS13}} is a 7-tuple	
	\[
\mathcal{	V} = (Q,\Sigma,\mathrm{M},\delta,q_1,Q_a,v),
	\]
where $ \mathrm{M} $ is a finite set of $k \times k$-dimensional rational-valued matrices, $ v $ is the initial ($k$-dimensional, rational-valued) row vector, and $ \delta $ is the transition function defined as
	\[
	\delta:Q \times \Sigma_{\dollar} \times \Omega \rightarrow Q \times \mathrm{M}.
	\]

Let $ w\in \Sigma^* $ be a given input. The automaton $ \mathcal{V} $ reads the sequence $ w\dollar $ from left to right symbol by symbol. It uses its states and its vector to store and process the information. In each step, it can check whether the first entry of the vector is equal ($=$) to 1 or not ($ \neq $). We call this feature the ``status" of the first entry and represent it by the set $ \Omega = \{=,\neq\}$. 
	
	The details of the transition function are as follows. When $ \mathcal{V} $ is in state $ q \in Q $, reads symbol  $ \sigma \in \Sigma_{\dollar} $, and the first entry status is $ \omega \in \Omega $, the transition $ \delta(q,\sigma,\omega) = (q',A) $ results in $\mathcal{V}$ entering state $ q' \in Q $, and its vector being multiplied by $ A \in \mathrm{M} $ from the right.
	
	At the beginning of the computation, $ \mathcal{V} $ is in state $ q_1 $ and the vector is $ v $. The initial vector is freely chosen by the designer of the automaton. Then, after reading each symbol, the state and vector are updated according to the transition function as described above. Thus the vector $  {v}^{(i)} $ at step $ i $ is obtained by multiplying the vector $  {v}^{(i-1)} $ at step $ i-1 $ by a specified matrix $ A $ so that $  { v}^{(i)}= { v}^{(i-1)}A $. The input $ w $ is accepted if the final state is an accept state and the first entry of the final vector is 1 after processing the right end-marker $\dollar$. Otherwise, the input is rejected. The set of all accepted strings is said to be the language recognized by $ \mathcal{V} $.
\newpage

 As the acceptance condition for many of the classical models requires that the register is equal to its initial value at the end of the computation, we adopt the same requirement for vector automata and propose homing vector automata. 
 
 A \textit {$ k $-dimensional homing vector automaton} (HVA($ k $)) is a 7-tuple
 \[{\mathcal{V}} =(Q,\Sigma,\mathrm{M},\delta,q_1,Q_a, { v} ),\]
 where $ \mathrm{M} $ is a set of $ k \times k $ rational valued matrices and $ v $ is an initial row vector with rational entries, as in the definition of vector automaton.
 
 A HVA is different from a vector automaton in two ways: (1) Homing vector automata do not read the right end-marker after reading the input, so there is no chance of postprocessing and, (2) instead of checking the status of the first entry, a homing vector automaton checks whether the complete current vector is identical to the initial vector or not. 

Formally, the transition function of a \textit{one-way $k$-dimensional deterministic homing vector automaton} (1DHVA($ k $)) is defined as
	\[\delta: Q \times
	\Sigma \times \Omega \ \rightarrow Q\times  \mathcal{D}  \times \mathrm{M},\]
	such that $\Omega$ is the set $\{=,\neq\}$, where $ = $  indicates
	equality to the initial vector $ { v} $, and $ \neq $ otherwise, $ \cal D $ is the set of head directions $ \{\downarrow,\rightarrow\} $ and $\mathrm{M}$ is a set of $k \times k$ rational-valued matrices. The initial vector is freely chosen by the designer of the automaton.

Specifically, $\delta(q,\sigma,\omega)= (q',d,A)  $ means that when ${\mathcal{V}}$ consumes $\sigma \in \Sigma$
in state $q\in Q$, with its current vector corresponding to $\omega \in \Omega$ ($\omega$ having the value = if and only if the current vector equals the initial vector),
it switches to state $q'\in Q$, multiplying its current vector with the matrix $A
\in \mathrm{M}$ on the right and moving the tape-head in direction $ d \in \mathcal{D}$.

A \textit{one-way $k$-dimensional deterministic blind homing vector automaton} (1DBHVA(\textit{k})) is a
restricted 1DHVA(\textit{k}) which is not allowed to check the vector until the end of the computation. The transition function $\delta$ is defined as
$$ \delta: Q \times \Sigma \rightarrow Q\times \mathcal{D} \times \mathrm{M}, $$
so that the next move of the machine does not depend on the current status of the vector.
%
%$\delta(q,\sigma)=(q',d,A)$ means that when ${V}$ reads symbol $\sigma \in \Sigma$
%in state $q\in Q$, it will move to state
%$q'\in Q$, multiplying the vector with the matrix $A \in \mathcal{M}$, and move its head in direction $ d \in \cal D $.

By omitting the tape-head directions and assuming that the tape-head moves right at each step, we obtain the \textit{real-time $k$-dimensional deterministic homing vector automaton} (DHVA($ k $)) and \textit{real-time $k$-dimensional deterministic blind homing vector automaton} (DBHVA($ k $)) models. The ranges of the corresponding transition functions are replaced with $Q \times \mathrm{M}  $.  

Now we are going to define the nondeterministic versions of homing vector automata. Formally, the transition function of a \textit{one-way $ $k-$dimensional $ nondeterministic homing vector automaton} (1NHVA(\textit{k})) is defined as
\[\delta: Q \times
\Sigma_{\varepsilon} \times \Omega \rightarrow \mathcal{P}(Q\times \mathrm{M}).\]

The blind version, \textit{one-way k-dimensional nondeterministic blind homing vector automaton} (1NBHVA(\textit{k})) is a 1NHVA(\textit{k}) which is not allowed to check the vector until the end of the computation. The transition function of a 1NBHVA($ k $) is defined as
$$ \delta: Q \times \Sigma_{\varepsilon} \rightarrow {\mathcal{P}}(Q\times \mathrm{M}).
$$

By not allowing $ \varepsilon $-moves, we obtain the real-time versions, \textit{real-time $k$-dimensional nondeterministic homing vector automaton} (NHVA(\textit{k})) and \textit{real-time $k$-dimensional nondeterministic blind homing vector automaton} (NBHVA(\textit{k})). The domains of the transition functions are replaced with $  Q \times \Sigma \times \Omega $ and  $ Q \times \Sigma $ respectively, by replacing $ \Sigma_{\varepsilon} $ with $ \Sigma $.

An input string $ w $ of length $ n $ is accepted by a homing vector automaton if 
there exists a computation in which the machine enters an accept state with the tape-head on the $ n+1 $'st tape square and the vector is equal to the initial value $ { v} $.

The abbreviations used for homing vector automata variants discussed so far are given in Table \ref{tbl: hva}.

		 \begin{table}[!h] \label{tbl: hva}
		 	\centering
		 	\vskip\baselineskip 
		 	 \caption{The abbreviations for HVA variants.}
		 	 \vspace{0.1in}
		 	\begin{tabular}{|p{0.3\textwidth}|p{0.15\textwidth}|p{0.15\textwidth}|}
		 		\hline 
		 		& Real-time & One-way \\
		 		\hline 
		 		Deterministic & DHVA$\displaystyle ( k)$ & 1DHVA$\displaystyle ( k)$ \\
		 		\hline 
		 		Deterministic blind & DBHVA$\displaystyle ( k)$ & 1DBHVA$\displaystyle ( k)$ \\
		 		\hline 
		 		Nondeterministic & NHVA$\displaystyle ( k)$ & 1NHVA$\displaystyle ( k)$ \\
		 		\hline 
		 		Nondeterministic blind & NBHVA$\displaystyle ( k)$ & 1NBHVA$\displaystyle ( k)$ \\
		 		\hline
		 	\end{tabular}
		 	
	 \end{table}

Given a homing vector automaton $\mathcal{V}=(Q,\Sigma,\mathrm{M},\delta,q_1,Q_a,v )$, we abbreviate it by $\textup{HVA}(k)_{\mathrm{M}} $ when we want to specify the set of matrices $ \mathrm{M } $ used by $ \mathcal{V} $. When we want to specify the number of states of a machine, we add an $ n $- (or $(n)$- to avoid any confusion) to the front of the model name, where $ n=|Q|$ is the number of the states. 

We will denote the set of $ k \times k $ matrices whose entries belong to the set $  \{-m,-m+1,\dots,\allowbreak 0,\dots,m-1,m\} $ for some positive integer $m$ by $ S_k(m) $.  

\section{Some Observations} \label{sec: hva-obs}

Homing vector automata are not allowed to perform postprocessing by definition, since the computation ends once they reach the right end-marker. We start by observing that allowing postprocessing does not bring any additional power to NBHVAs and 1NBHVAs. 

HVAs using end-marker will be denoted by the abbreviation $\textup{HVA}_\$ $.

\begin{lem}\label{lem: rtNBend}
Let $ L $ be a language recognized by a $\textup{XBHVA}_\$(k)$ $ \mathcal{V} $, where $X \in \{\textup{N,1N}\}  $. Then, $ L $ is also recognized by a $\textup{XBHVA}(k)$ $ \mathcal{V}' $.
\end{lem}
\begin{proof}
First of all, we assume that $\mathcal{V} $ does not contain any $ \dollar $-transitions that do not lead to an accept state, since any such transitions may be removed from $ \mathcal{V} $ without changing the language. We are going to analyze the cases of real-time and one-way computation and in each case, we will start constructing $\mathcal{V}'$ such that $\mathcal{V}'$ mimics the transitions of $\mathcal{V}$ on every possible symbol $\sigma \in \Sigma_{\varepsilon}$. 

If the computation is real-time, then it ends as soon as the right end-marker is processed. We create new transitions to handle the postprocessing, which emulate $ \mathcal{V} $'s last action before reading the end-marker (which would end up in an accept state) and the end-marker ($ \sigma \dollar $) at once: At any point during the scanning, if reading $ \sigma $ would cause $\mathcal{V}$ to switch to a state from which the end-marker $ \dollar $ would lead to an accept state, a new nondeterministic transition takes $ \mathcal{V}' $ to the additional state, which is an accept state. During this transition, the register is updated so that the update accounts for both reading $ \sigma $ and $ \dollar $. All other states of $\mathcal{V}'$, which are inherited from $\mathcal{V}$, are designated to be non-accept states. Thus, $ \mathcal{V}' $ simulates the computation of $ \mathcal{V} $ on any non-empty string, and accepts the input in some computational path if and only if $ \mathcal{V} $ accepts it. 

Now suppose that the computation is one-way. Let $ Q_{\dollar} $ be the set of states of $ \cal V $ that have an outgoing $ \dollar  $-transition. After finishing reading the string, $ \cal V $ should enter a state from $ Q_{\dollar} $, read the $\dollar  $ symbol and possibly make some $ \varepsilon $-transitions and eventually end in an accept state, to accept any string. Let $ G_{\dollar} $ be the graph obtained from the transition diagram of $ \cal V $, by removing all transitions except the $\dollar  $-transitions and $ \varepsilon $-transitions. Let $ r_q $ be the subgraph of $G_{\dollar}$, induced by the set of reachable vertices from $q$ in $G_{\dollar}$, for each $ q \in Q_{\dollar}  $. We create a copy of each subgraph $r_q  $ and denote it by $ r_q^c $, replace the $ \dollar $ symbols in $ r_q^c $ with $ \varepsilon $ and connect it to $ \cal{V} '$: For each incoming transition to $ q $ in $ \mathcal{V} $, we create a copy of the transition and connect it to the copy of $ q $ in $ r_q^c $. The $ \dollar $-transitions inherited from $ \cal V $ are removed from $\cal{V}'  $ and any accept state of $ \mathcal{V} $ is no longer an accept state in $ \mathcal{V}' $. $ \mathcal{V}' $ simulates the computation of $ \mathcal{V} $ on any non-empty string until scanning the $ \dollar $ and then follows the transitions in the newly added states to reach an accept state. 

If $ L $ contains the empty string, we add one more  state that has the following properties: (i) it becomes the new initial state of the resulting machine, (ii) it is an accept state, (iii) it causes the machine to behave (i.e. transition) like the original initial state of $ \mathcal{V} $ upon reading the first symbol, and (iv) there is no transition coming in to this state.
\end{proof}

The idea given in the proof of Lemma \ref{lem: rtNBend} does not apply for non-blind models since the status of the vector may be changed after reading the last symbol of the input (just before reading the right end-marker). In fact, one can show that DHVAs using end-marker are more powerful than ordinary DHVAs in terms of language recognition by the witness language $\mathtt{NEQ}=\{ a^ib^j | i \neq j\}$ .

\begin{thm} \label{thm: DHVAendmarker}
$ \bigcup_k\mathfrak{L}(\textup{DHVA($ k $)}) \subsetneq \bigcup_k \mathfrak{L}(\textup{DHVA($ k $)}_\$)$.
\end{thm}
\begin{proof}
The subset inclusion is immediate, since the postprocessing may very well involve multiplication with the identity matrix. For the inequality, consider the language $\mathtt{NEQ}=\{ a^ib^j | i \neq j\}$. Suppose for a contradiction that there exists a DHVA($ k $) $\mathcal{V}$ recognizing $\mathtt{NEQ}$ for some $k$. Let $v$ be the initial vector of $\mathcal{V}$. There exist sufficiently long strings $w_1=a^mb^n$ and $w_2=a^mb^o$, $m \neq n$, $m \neq o$, $n < o$ such that $\mathcal{V}$ is in the same accept state
after reading $w_1$ and $w_2$ and the vector is equal to $v$, since the strings belong to $\mathtt{NEQ}$. When both strings are extended with $b^{m-n}$, $a^mb^nb^{m-n} \notin \mathtt{NEQ}$ whereas  $a^mb^o b^{m-n} \in \mathtt{NEQ}$. Since the same vector is being multiplied with the same matrices associated with the same states during the processing of the string  $b^{m-n}$, it is not possible for $\mathcal{V}$ to give different responses.

Now let us prove that the language $\mathtt{NEQ}$ can be recognized by a $\textup{DHVA}_\$(2)$ $\mathcal{V}'$. The initial vector of $\mathcal{V}'$ is $v'=\mypar{1~~1}  $, and the state diagram of $ \mathcal{V}' $ is given in Figure \ref{fig: neq}.

\begin{figure}[h]
	\centering
	\includegraphics[width=0.7\linewidth]{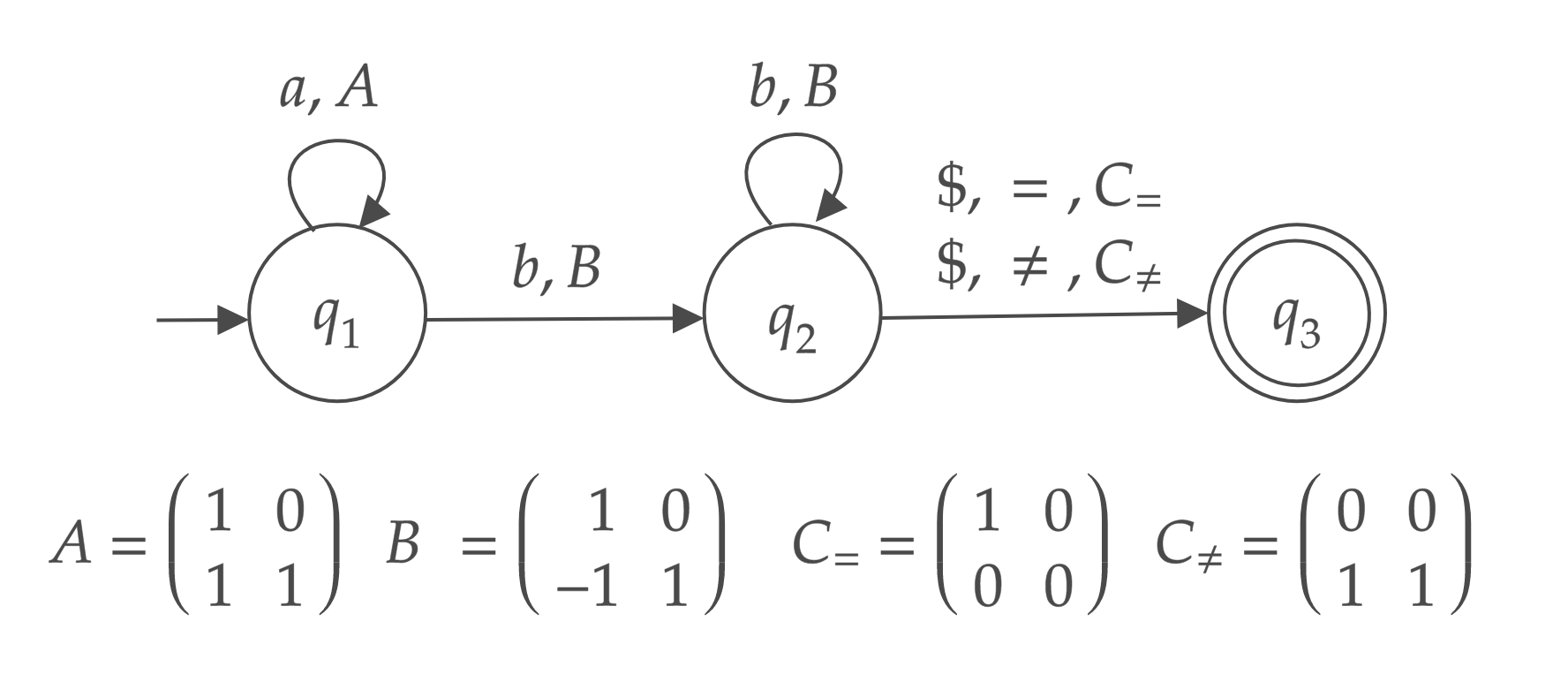}
	\caption{State diagram of $ \mathcal{V}' $ recognizing $ \mathtt{NEQ} $}
	\label{fig: neq}
\end{figure}

\begin{comment}

\[A=\mymatrix{ cc }{
1 & 0 \\
1 & 1
}~~
B=\mymatrix{ rc }{
1 & 0 \\
-1 & 1
}~~
C_= =\mymatrix{ cc }{
1 & 0 \\
0 & 0
}~~
C_{\neq}=\mymatrix{ cc }{
0 & 0 \\
1 & 1
}
\]

\end{comment}
For each $a$, the first entry is increased by 1 and for each $b$ the first entry is decreased by 1 using the second entry of the vector, which is equal to 1 throughout the computation. The increment and decrement operations are performed by the matrices $A$ and $B$.

When reading the end-marker, if the value of the vector is equal to its initial value,  meaning that the number of $a$'s and $b$'s were equal to each other,  the vector is multiplied with $C_=$, which sets the second entry to 0, so that the input string is not accepted. Otherwise, if the vector is not equal to its initial value, meaning that the number of $a$'s and $b$'s were not equal, the vector is multiplied with $C_{\neq}$, which sets the first entry to 1. This returns the vector to its initial value, and the input 
string is accepted.
\end{proof}

%In the following lemma, we show that any BHVA with end-marker whose matrices are rational valued can be simulated by a BHVA with end-marker and integer valued matrices in the cost of increasing the size of the vector by 2. 
Any BHVA with end-marker whose matrices are rational valued can be simulated by a BHVA with end-marker and integer valued matrices in the cost of increasing the size of the vector by 2. The proof is due to Abuzer Yakaryılmaz and can be found in \cite{SYS19}.
\begin{lem}
	\label{thm: ratint}
	For any given rational-valued $(n)$-$ \textup{XBHVA}_\$(k)$ $ \mathcal{V} $, where $ X \in \{\textup{D,1D,N,1N}\} $, there exists an integer-valued $(n)$-$ \textup{XBHVA}_\$(k+2)$ $ \mathcal{V}' $  that recognizes the same language.
\end{lem}

For nondeterministic HVAs, we can state the following corollary.

\begin{cor}
	Rational-valued \textup{XBHVA}s and integer-valued \textup{XBHVA}s where $X \in \{\textup{N,1N}\}  $ recognize  the same class of languages.
\end{cor}
\begin{proof}
	By using Lemma \ref{thm: ratint}, we can conclude that any rational-valued XBHVA can be simulated by an integer-valued XBHVA using the end-marker. Then, by using Lemma \ref{lem: rtNBend}, we can remove the end-marker.
\end{proof}

Note that although they recognize the same classes of languages, a rational valued HVA can be simulated by an integer valued HVA in the cost of increasing the vector size by 2 and using some additional states.

\section{Encoding Strings with Homing Vector Automata}\label{sec: hva-SB}
While recognizing certain languages, it may be necessary to hold information about the string that is read so far in the entries of the vector. We call this notion ``encoding the string''.  

In this section, we are going to discuss some encoding techniques that can be performed by homing vector automata. The methods are applicable by the most restricted, real-time deterministic and blind version and therefore can be carried out by any homing vector automata. In the first part, we present the generalized Stern-Brocot encoding which can be performed by $ k $-dimensional homing vector automata using only matrices belonging to the set $ S_k(1) $, for any string belonging to a $ k $-letter alphabet. In the second part, we present another encoding method which can be performed by 2-dimensional HVAs, regardless of the size of the alphabet. The method also can be used for base conversion, which may be necessary while recognizing some specific languages.

\subsection{Stern-Brocot Encoding}\label{sec:binary encoding}
The Stern-Brocot tree is an infinite complete binary tree whose nodes correspond one-to-one to positive rational numbers \cite{St58,Br61}. Crucially for our purposes, the Stern-Brocot tree provides a basis for representing strings as vectors of integers, as suggested for binary alphabets in \cite{GKP89}. 
The fractions in the Stern-Brocot tree can be stored as vectors of dimension 2, where the vector entries are the denominator and the numerator of the fraction. This representation allows us to perform the binary encoding easily in homing vector automata, as follows.

The empty string is represented by $\mypar{1 ~~ 1}$. Now suppose that we want to encode a binary string $ w $ of length $ n $. For $ i=1 $ to $ n $, if $ w[i]=0 $, we add the value of the first entry to the second one, and if $ w[i]=1 $, we add the value of the second entry to the first one, multiplying the vector with the appropriate one of the following matrices $ M_0 $ and $ M_1 $:
$$ 
M_{0}= \mymatrix{rr}{1&1\\0&1\\} ~~~~~
M_{1}= \mymatrix{rr}{1&0 \\ 1&1}
$$ 
A list of some binary strings and their encodings is as follows. A proof on the uniqueness of the encoding can be found in \cite{GKP89}. 
\begin{align*}
0 & \hspace{0.1in} \mypar{1~~2}
& 00 & \hspace{0.1in} \mypar{1~~3}
& 10 & \hspace{0.1in} \mypar{2~~3}
& 000 & \hspace{0.1in} \mypar{1~~4}
& 010 & \hspace{0.1in} \mypar{3~~5}   \\
1 & \hspace{0.1in} \mypar{2~~1}
 & 01 & \hspace{0.1in}  \mypar{3~~2}
 & 11 & \hspace{0.1in} \mypar{3~~1} 
 & 001 &\hspace{0.1in} \mypar{4~~3}  
 & 011 &\hspace{0.1in} \mypar{5~~2}   
\end{align*}   

Given the vector representation $  v_w $ of a string $ w $, it is also possible to decode the string with the following procedure: Let $ |w|=n $ and $  v_w= \mypar{
a~~ b} $. Set $ w[n]=0 $ if $ b>a $, and $ w[n]=1 $ otherwise. Subtract the smaller entry from the larger one to obtain $  v_w^{n-1} $ and repeat this routine until you obtain the vector $  \mypar{
1~~ 1} $. When the given vector is not a valid representation of a string, then it is not possible to obtain  $ \mypar{
1~~1} $. The  matrices required for this procedure are $ N_0 $, which has the effect of subtracting the value of the first entry of the vector from the second entry, and $ N_1 $, for the symmetric action. Note that $ N_0 = ({M_0})^{-1}$  and $ {N_1} = ({M_1})^{-1} $.
$$ 
N_{0}=\mymatrix{rr}{1&-1\\0&1\\}~~~~~
N_{1}=\mymatrix{rr}{1&0\\-1&1\\}
$$ 
%
%In \cite{Ni07}, it is noted that the procedure is similar to the Euclidean Algorithm and one ends up with the following corollary.
%
%\begin{cor}
%Let $  [\begin{array} {rr}
%a&b
%\end{array}] $  be a valid vector representing a binary string. Then $ a $ and $ b $ are relatively prime.
%\end{cor}

%\subsection{Generalized Stern-Brocot encoding}\label{sec:genstern}

We generalize the scheme mentioned above to strings on alphabets of arbitrary size and present a new method for encoding strings. Let $ \Sigma=\{a_1,a_2,\dots,a_k \} $, and $ w \in \Sigma^*$. With the \textit{generalized Stern-Brocot encoding} method described below, it is possible to uniquely encode $ w $ using a vector of size $ k $ and $ k \times k $ matrices whose entries belong to the set $ \{-1,0,1\} $. 

%Let us note that one can use other methods to encode strings on arbitrary alphabet size using a vector of a smaller dimension but matrices whose entries belong to a larger set.   

We start with the $ k $ dimensional vector $\mypar{1~~1~~\cdots ~~1}  $, which represents the empty string. Suppose that  $|w|=n  $. To encode $w$, for $ i=1 $ to $ n $, if $ w_i=a_j $, the vector is multiplied with the matrix $E^k_j  $, the $ k $ dimensional identity matrix whose $ j $'th column has been replaced with a column of $ 1 $'s. Multiplication with $ E^k_j $ causes the  $ j $'th entry of the vector to be replaced by the sum of all the entries in the vector. 

Among the different generalizations of the Stern-Brocot fractions, one that appears in \cite{Ga13} under the name of ``Stern's triatomic sequence'' is similar to the encoding we propose for the case $ k=3 $. The similarity lies in the construction of the sequence, but that sequence is not used for the purpose of encoding. As far as we know, no such generalization exists for the case $ k>3 $.

In the following lemma, we prove the uniqueness of this generalized encoding.

\begin{lem}\label{lem:unique}
	No two distinct strings on $\Sigma$ \textup{ ($ |\Sigma|=k $)} can be represented by the same  vector of size $ k $ using the generalized Stern-Brocot encoding.
\end{lem} 
\begin{proof}
	We will prove by induction on $n$ that if a $k$-dimensional vector $ v $ is the generalized Stern-Brocot encoding of a string of length $n$, then $ v$ is not the encoding of any other string of length at most $n$. 
	%We have described a way of representing strings belonging to an alphabet of size $ k $ using $ k $-dimensional vectors. Now our aim is to show that the vectors obtained by this procedure are unique, so that every string can be represented uniquely with a $ k $ dimensional vector. We prove our claim by induction on the length of the string.
	
	The empty string is represented by the $ k $-dimensional vector of 1's.
	%, which we will call $ e_k $. 
	The claim clearly holds for $n=0$, since no other strings of at most this length exist.
	%Encoding for the empty string is unique since encoding any string other than the empty string requires multiplication of $ e_k $ with some matrix $ A_j $, which would result in a vector different than $ e_k $. 
	Now assume that the claim holds for all natural numbers up to $ n-1 $. Let $ w $ be a string of length $ n $. The vector $  v_w $ representing $ w $ is obtained by multiplying the vector $  v_w^{n-1} $, representing the first $ n-1 $ symbols of $ w $, with $ E^k_j $ if $ w[n]=a_j $. We will examine various possibilities regarding this final multiplication. Note that at a single step, it is possible to modify only a single entry of each vector. Now consider any string $ u \neq w $ with $ |u|=l $ and $ l \leq n $. If $ w $ and $ u $ have the same first $ n-1 $ symbols, then $  v_w^{n-1}= v_u^{l-1} $, the last symbols of the two strings are unequal, and it is not possible to obtain $  v_w= v_u $ since the same vector is multiplied by different matrices. In the remaining case, we know by the induction hypothesis that $  v_w^{n-1}\neq  v_u^{l-1} $. If these vectors disagree in more than two entries, there is no way that one can obtain the same vector by multiplying 
	them once with some matrices of the form $E^k_j$. So we consider the case of the two vectors disagreeing in at most two entries.
	
	%Let us analyze the value of the vector in the next step by considering different possibilities. 
	Suppose that $  v_w^{n-1}$ and $ v_u^{l-1} $ differ only in the  $ i $'th entry. If the final multiplications both work on the $i$'th entries, they will be adding the same number to them, resulting again in vectors differing in their $ i $'th entries. If one or more of the final multiplications deals with another entry, then the final vectors will surely disagree in that entry. It is not possible in any case to end up with equal vectors,
	
	Now suppose that $  v_w^{n-1}$ and $ v_u^{l-1} $ differ in two entries.
	%$ i $ and $ j $. 
	If the final multiplications work on the same entry, then the final vectors will disagree in at least one entry.
	%$ i$'th ($ j $'th) entries of the both vectors are incremented, then in the next step the $ j $'th ($ i $'th) entries of the vectors will remain unchanged and the vectors will be different. 
	In the only remaining case, each one of the vectors is multiplied by a matrix updating a different one of the disagreeing entries. Let us represent the disagreeing entries of the vectors $  v_w^{n-1} $ and $  v_u^{n-1} $ by the pairs $(a~~ b)$ and $(c~~ d)$, respectively. Let $ x $ be the sum of the remaining $k-2$ entries in which the vectors agree. Without loss of generality, say that the entries become $(a~~~a+b+x)$ and $(c+d+x~~~d)$ after the final multiplication.
	But if the final vectors are equal, these pairs should also be equal, implying  $ c+b+2x=0 $, an impossibility. 
	
	We therefore conclude that it is not possible to have $  v_w= v_u $ for any string $ u $ of length at most $n$.
\end{proof}

Like in the binary case, given the vector representation  of a string, it is possible to reconstruct the string. The all-ones vector corresponds to the empty string. Any other vector $  v_w$ encoding a string $w$ of length $n$ in this encoding has a unique maximum entry, say at position $j$. Then $ w[n]$ is $a_j $, and we obtain $  v_w^{n-1} $ by subtracting the sum of the other entries from the greatest entry. One repeats this procedure, reconstructing the string from right to left, until  one ends up with the all-ones vector. In terms of matrices, multiplications with the inverses of $ E^k_j $s capture this process. 

We demonstrate the use of generalized Stern-Brocot encoding in the following example.

\begin{ex} \label{ex: stern}
$ \mathtt{MPAL_l}=\{w\#w^r|w\in\{a_1,a_2,\dots,a_l\}^*\} \in \mathfrak{L}($\textup{DHVA($ l $)}$_{S_l(1)} )$.
\end{ex}
Let us construct a DHVA($ l $)$ _{S_l(1)} $ $ {\mathcal{V}} $ recognizing $\mathtt{MPAL_l}  $. The input alphabet is $ \{a_1,a_2,\dots,a_l\} $,  and the corresponding matrices are $ \{E^l_1,E^l_2,\dots,E^l_l\}$. Starting with the $ l $ dimensional vector of 1's,  $ {\mathcal{V}} $ encodes the string by multiplying its vector with the matrix $ E^l_j $ whenever it reads an $ a_j $ until it encounters a $ \# $ . After reading the $ \# $, $ {\mathcal{V}} $ starts decoding by multiplying the vector with matrix $ ({E^l_j}) ^{-1}$ whenever it reads an $ a_j $.

If the string is of the form $ w\# w^r $, the vector will be multiplied with the inverse matrices in the correct order and the resulting value of the vector will be $\mypar{1~~1~~\cdots ~~1}$.

We also need to show that the input string is not accepted when it is not of the form $ w\#w^r $. Consider an input string $ x\#y^r $ and suppose that it is accepted by $ {\mathcal{V}} $. Let $ v' $ denote the vector after reading $ x\# $ and let $ Y $ denote the product of the matrices the vector is multiplied while reading $ y^r $. Since the string is accepted, $ { v'} Y=\mypar{1~~1~~\dots~~1} $ must be true. Since the matrices $ {(E^l_j)}^{-1} $ are invertible, $ Y $ is also invertible, which implies that $  v' $ must be unique. Since $ y\#y^r \in \mathtt{MPAL}$, then $ v' $ must be the vector obtained after reading $ y $ . From Lemma \ref{lem:unique}, we know that every string has a unique representation and we conclude that $ x $ and $ y $ are identical.
\subsection{Base-$ m $ Encoding}

In this subsection, we discuss a well-known encoding technique, which can be easily adopted to homing vector automata.
%The well-known encoding performed by matrices can be easily adopted to homing vector automata.

%For any $ x \in \{1,2\}^+ $, let $ e(x) $ be the number encoded by $ x $ in base 3:
%\[
%e(x) = 3^{|x|-1} x[1] + 3^{|x|-2} x[2] + \cdots + 3^1 x[|x|-1] + 3^0 x[|x|]  .
%\]
%
%The encoding $e(x)$ can be easily obtained by using vector-matrix multiplications. Starting with the initial vector $v= (1~~0 )$, and multiplying the vector with 
%\[
%A_1 = \mymatrix{rr}{ 1 &  1\\ 0 &  3 }
%\mbox{ and }
%A_2 = \mymatrix{rr}{ 1 &  2 \\ 0 &  3},
%\]
%respectively
%for each scanned $1$ and $2$ from left to right, one obtains $(1~~e(x) )$ at the end.

%Decoding is performed by using the inverse matrices $ A_1^{-1} $ and $ A_2^{-1} $.
%
%\[
%A_1^{-1} = \mymatrix{rr}{ 1 &  -\frac{1}{3}\\ 0 &  \frac{1}{3} }
%\mbox{ and }
%A_2^{-1} = \mymatrix{rr}{ 1 &  -\frac{2}{3} \\ 0 &  \frac{1}{3}}.
%\]

%
%We can easily extend this encoding for any generic alphabet. 

For any $ w \in \{1,2,\ldots,m-1\}^+ $ and $ m \leq 10 $, let $e_m(w)$ be the base-10 number encoded by $ w $ in base-$m$:
\[
e_m(w) = m^{|w|-1} w[1] + m^{|w|-2} w[2] + \cdots + m^1 w[|w|-1] + m^0 w[|w|]  .
\]
The encoding $e_m(w)$ can be easily obtained by using vector-matrix multiplications. Starting with the initial vector $v= (1~~0 )$, and multiplying the vector with 
\[
A^m_i = \mymatrix{rr}{ 1 &  i \\ 0 &  m}
\]
for each symbol $i$, $ e_m(w) $ is obtained in the first entry of the vector.

When the vector is multiplied by $ A^m_i $, the second entry is multiplied by $ m $ and then incremented by $ i $. As a result, this process ends up with $ e_m(w) $ appearing in the second entry of the vector.

Given $ w \in  \{1,2,\ldots,m-1\}^+  $, it is also possible to obtain the encoding for $ w^r $, that is $ e_m(w^r) $ in the second entry of the vector. This can be accomplished starting with the initial vector $ \mypar{1~~0} $ and multiplying the vector with the matrices $ B^m_i $ for each symbol in $ \Sigma $.
\[ 
B^m_i=\mymatrix{rr}{m&i\\0&1}
\]
\newpage
Multiplication by $ B^m_i $ increments the second entry by $ i $ times the first entry and multiplies the second entry by $ m $. After reading $ k $ symbols, the first entry holds $ m^k $. 

By the base-$ m $ encoding, any number in base-$ m $ that does not contain a 0 digit can be converted to its base-$10 $ equivalent. One should be careful if 0 is included in the alphabet, since appending 0s at the beginning of the string wouldn't change its encoding and the encoding wouldn't be unique in such a case. Nevertheless, by letting the alphabet to be $ \{0,1,\dots,m-1\} $ and using the same matrices as above, any number in base-$ m $ can be converted to its base-$ 10 $ equivalent.  

For instance, for the case where the alphabet has size 2, the matrices used for the conversion have the following form:
\[ 
A^2_0=\mymatrix{rr}{1&0\\0 &2}~~~~
A^2_1=\mymatrix{rr}{1&1\\0 &2}.
\]

Hence given $ w \in 1\{0,1\}^* $, one obtains the base-10 integer corresponding to the binary number represented by $ w $ in the second entry of the vector upon multiplication by the matrices $ A^2_0 $ and $ A^2_1 $.

When $ m=10 $ and $ w $ is a string containing at most $ 9 $ different symbols, then note that one can obtain $ w $ in the first entry of the vector, by multiplying the vector with $ A^{10}_i $ for each symbol $ i \in \Sigma=\{0,1,\dots,9\}$.
\[ 
A^{10}_i=\mymatrix{rr}{1&i\\0 & 10}
\]

\section{Relationship with Counter Automata}\label{sec: hva-ca}

In this section, we are going to analyze the relationship between counter automata and homing vector automata. In the first part, we will examine the blind case and show that homing vector automata outperform counter automata. In the second part, we will present some incomparability results between real-time non-blind counter automata and homing vector automata.

\subsection{Blind Counter Automata}
We are going to start this section by showing that BHVA(1)'s and $ k $BCA's are equivalent in power.

\begin{thm}\label{thm: counters}
	$\bigcup_k\mathfrak{L}( \textup{X$ k $BCA}) = \mathfrak{L}(\textup{XBHVA(1)}) $ where $ X \in \{\textup{D,N,1D,1N}\} $.
\end{thm}
\begin{proof}
	Let $\mathcal{C}$ be an X$ k $BCA. We are going to construct a XBHVA(1) $ \mathcal{V}$ simulating $ \mathcal{C} $. The register of $ \mathcal{V} $ is initialized to 1. We are going to choose $ k $ distinct primes $ \{p_1,\dots,p_k\} $ to be multiplied with the register of $ \mathcal{V}$, each representing a counter. An increment and decrement of the $ i $'th counter of $ \mathcal{C} $ is simulated  by multiplying $ \mathcal{V} $'s register by $ p_i $ and $ \frac{1}{p_i} $ respectively. A string is accepted by $\mathcal{V} $ if the register is equal to 1 at the end of the computation, that is when all the counters are equal to 0.
	
	Now suppose that we are given a XBHVA(1) $ \mathcal{V} $. We are going to construct a X$ k $BCA $ \mathcal{C} $ simulating $ \mathcal{V} $. We may assume that the register of $ \cal V $ is not multiplied by 0, since such a computation will never be accepting. Let $A=\{a_1,a_2,\dots,a_n\}$ be the set of all rational numbers the register of $ \cal V $ can be multiplied with. Let $P=\{p_1,p_2,\dots ,p_k\}$ be the set of prime factors of the denominators and the numerators of the rational numbers in $A$. Then each $ a_i \in A$ can be expressed as 
	$$
	a_i=(-1)^{t_i}\frac{p_1^{x_{1_i}}p_2^{x_{2_i}}\cdots p_k^{x_{k_i}}}{p_1^{y_{1_i}}p_2^{y_{2_i}}\cdots p_k^{y_{k_i}}},
	$$ 
	\newpage
	\noindent where $ t_i=0 $ if $ a_i $ is positive and $ t_i=1 $ if $ a_i  $ is negative. $ \mathcal{C} $ will have $ k $  counters and the state set of $ \mathcal{C} $ will consist of two copies of $ \cal V $'s states, $ Q \times 0 $ and $ Q \times 1 $. The two copies are needed to encode the sign information of $ \mathcal{V}'s $ register in the states of $ \cal C $. Any computation starts in the first copy, in the state $( q_1 , 0 )$. Suppose that $ \cal V $ moves to state $ q' $ from state $ q $ by multiplying its register with $ a_i $. This transition is implemented in $ \cal C $ by adding the following transitions: If $ t_i=0 $, then the transitions from $ (q,0) $ to $ (q',0) $ and from $ (q,1) $ to $(q',1)  $ exist in $ \mathcal{C} $. If $ t_i=1 $, then the transitions from $ (q,0) $ to $ (q',1) $ and from $ (q,1) $ to $(q',0)  $ exist in $ \mathcal{C} $. In these transitions, the $ j $'th counter is incremented by $ x_{j_i}-y_{j_1} $. The only accept states of $ \cal C $ are those of the form $(q ,t_0 )$ where $ q  $ is an accept state of $ \cal V $. A string is accepted by $ \cal C $ when all the counters are equal to 0 , that is when the register of $ \cal V $ is equal to 1.
\end{proof}

 We list the following equivalences among the models.

\begin{itemize}
	\item $ \mathfrak{L}( \textup{1NBHVA(1)})=\bigcup_k\mathfrak{L}( \textup{1N\textit{k}BCA}) =\bigcup_k\mathfrak{L}(\mathbb{Z}^k) = \mathfrak{L}(\mathbb{Q}^+) = \mathfrak{L}(\textup{1NFAMW}) $

\end{itemize}

In the next theorem, we show that HVA(2)s are more powerful than counter automata. Before proving our result, we first prove a lemma to compare the computational power of real-time and one-way deterministic blind counter automata.

\begin{lem}\label{lem: detrt}
	$ \bigcup_k\mathfrak{L}(\textup{D$ k $BCA}) = \bigcup_k\mathfrak{L}(\textup{1D$ k $BCA}).$
\end{lem}
\begin{proof} We are going to show that any computation by a 1D$ k $BCA can be carried out in real-time by a D$ k $BCA. For each input symbol, there can be only one outgoing transition from each state, since the computation is deterministic. There cannot be any loop in the state diagram which does not move the tape head, since in such a case, the next symbol will never be scanned. The only transitions which don't move the tape head should be in the form schematized in Figure \ref{fig: loop}. We omit the counter updates in the figure.
	
	\begin{figure}[h]
		\centering
		\includegraphics[width=1\linewidth]{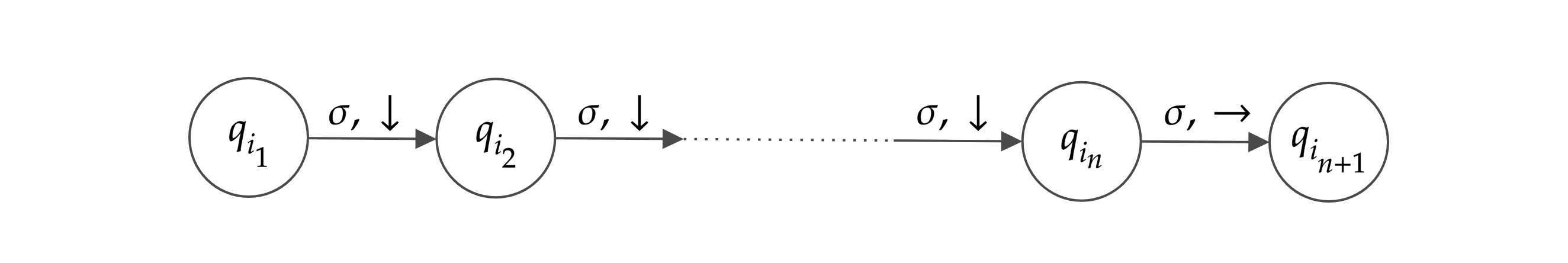}
		\caption{A part of the transition diagram of a one-way deterministic blind counter automaton}
		\label{fig: loop}
	\end{figure}
	
	There should be a sequence of transitions which do not move the tape head when scanning the symbol $ \sigma $, followed by a transition which moves tape head right, upon scanning the same symbol. Any other outcoming arrows from these states cannot be traversed as it is not possible to read any symbol different from $ \sigma $ without moving the tape head. This sequence of transitions can be handled by a single transition which moves the machine from $ q_{i_1} $ to $ q_{i_n} $, moves the tape head right and makes the necessary modifications on the counters. Hence a D$ k $BCA recognizing the same language as the original one which operates in real-time can be obtained. 
\end{proof}

\begin{thm}\label{thm: bhva2}
	$ \bigcup_k\mathfrak{L}\textup{(X$ k $BCA)} \subsetneq  \mathfrak{L} \textup{(XBHVA(2))}$ where $ X \in \{\textup{D,1D,N,1N}\} $.		
\end{thm}
\begin{proof}
	The inclusions follow by Theorem \ref{thm: counters}, as HVA(1)s are capable of simulating counter machines.
	For the deterministic case, it is known that no D$ k$CA can recognize the language $ \mathtt{MPAL}_2 $ in real-time for any $ k $ \cite{Pe11}. Hence we get that no D$ k $BCA and therefore no 1D$ k $BCA (as the two models are equivalent by Lemma \ref{lem: detrt}) can recognize $ \mathtt{MPAL}_2 $. On the other hand, $ \mathtt{MPAL}_2 $ can be recognized by a DBHVA(2) as shown in Example \ref{ex: stern}.
	
	For the nondeterministic case, any unary language recognized by a 1N$ k $BCA is regular by Fact \ref{fact: unary}, since $ \bigcup_k\mathfrak{L}$(1N$ k $BCA)= $ \mathfrak{L}$(1NFAMW). It is possible to recognize the unary nonregular language $ \mathtt{UPOW}'=\{a^{n+2^n} | n \geq 1 \} $ by a NBHVA(2) $ \mathcal{V} $ with the initial vector $ v = \mypar{1~~1} $, whose state transition diagram is given in Figure \ref{fig: upow}.
	
	\begin{figure}[h]
		\centering
		\includegraphics[width=0.8\linewidth]{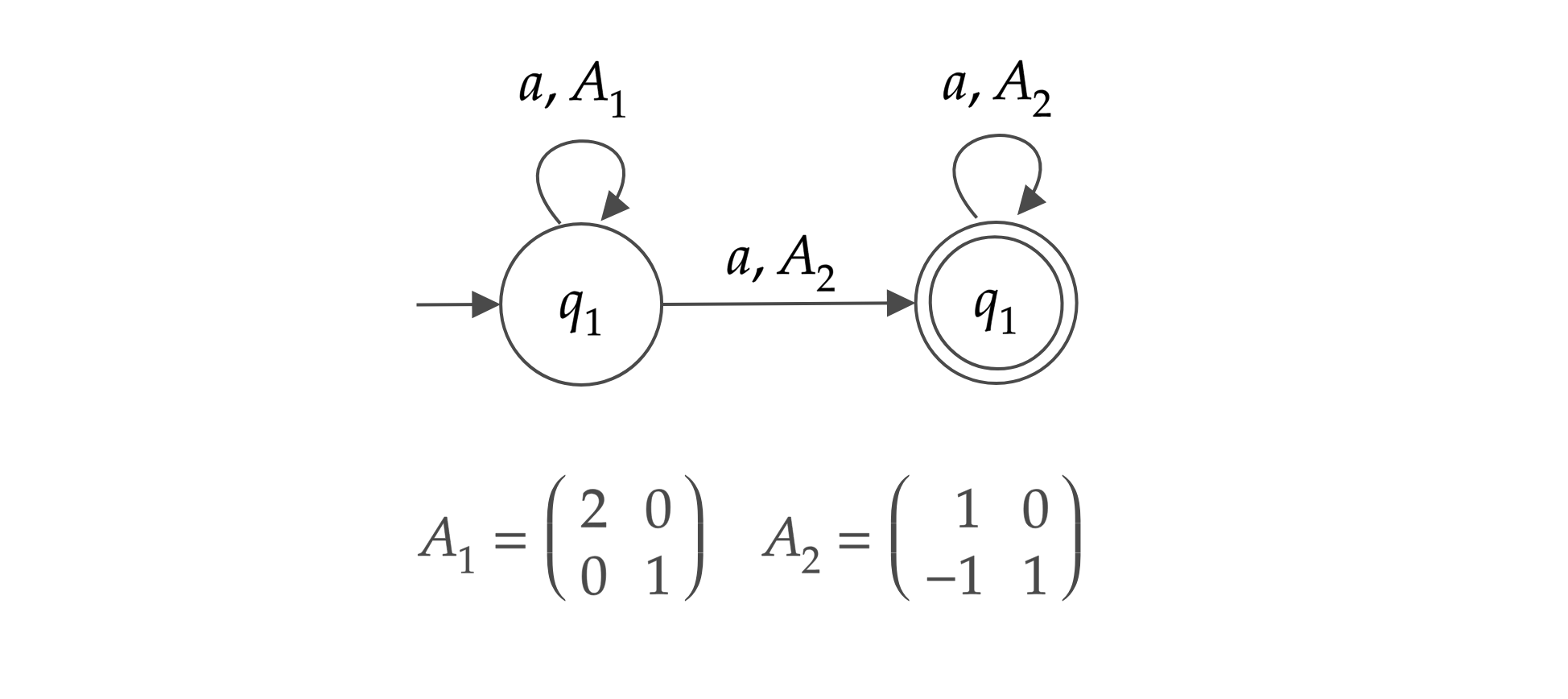}
		\caption{State transition diagram of $ \mathcal{V} $ recognizing $ \mathtt{UPOW' } $ }
		\label{fig: upow}
	\end{figure}
\end{proof}

Now we show that the same result can be achieved by 3-dimensional nondeterministic HVAs whose matrices are restricted to have integer entries.

\begin{thm}\label{thm: z3}
	$ \bigcup_k\mathfrak{L}(\textup{X$ k $BCA}) \subsetneq \mathfrak{L}\textup{(XBHVA(3)}_{M_3(\mathbb{Z})}) $ where $ \textup{X} \in \{\textup{N,1N}\} $.
\end{thm}
\begin{proof}
	Any X$ k $BCA can be simulated by a XBHVA(1) by Theorem \ref{thm: counters} and any XBHVA(1) can be simulated by a XBHVA$_\dollar $(3)$_{M_3(\mathbb{Z})}$ by Lemma $ \ref{thm: ratint} $. By using additional states, we can obtain an equivalent $\textup{XBHVA(3)}_{M_3(\mathbb{Z})} $ without end-marker by Lemma \ref{lem: rtNBend}. 
	
	We are going to show that the inclusion is proper by constructing a NBHVA(3)$_{S_3(1)} $ $ {\mathcal{V}} $ recognizing the unary nonregular language $ \mathtt{UPOW}'=\{a^{n+2^n}|n\geq 1\} $. The state transition diagram of $ \mathcal{V} $ is given in Figure \ref{fig: upow2}. Starting with the initial vector $  \mypar{1~~1~~1}$, $ {\mathcal{V}} $ multiplies the vector with matrix $ A_1 $ when reading each $ a $. The idea is to add the first and second entries together repeatedly to obtain powers of 2, so that after reading $ k $ symbols the value of the vector is equal to $ \mypar{2^k ~~ 2^k~~1} $. $ {\mathcal{V}} $ nondeterministically guesses $ n $ and starts decrementing the first entry from that point on by multiplying the vector with the matrix $ A_2 $ which fixes the second entry to 1 immediately. At the end of the computation, the value of the vector is equal to 
	$  \mypar{1~~1~~1} $ if and only if the input string is of the form $ a^{n+2^n} $ for some $ n $. 
%	$$U_{1}=
%	\left [
%	\begin{array} {rrr}
%	1&1&0\\
%	1&1&0\\
%	0&0&1\\
%	\end{array}
%	\right ]~~~~~
%	U_{2}=
%	\left [
%	\begin{array}{rrr}
%	1&0&0\\
%	0&0&0\\
%	-1&1&1\\
%	\end{array}
%	\right ]
%	$$ 
	
	We conclude the result since no X$ k $BCA can recognize the language $ \mathtt{UPOW'} $ by Fact \ref{fact: unary}.
	
	\begin{figure}[h]
		\centering
		\includegraphics[width=1.1\linewidth]{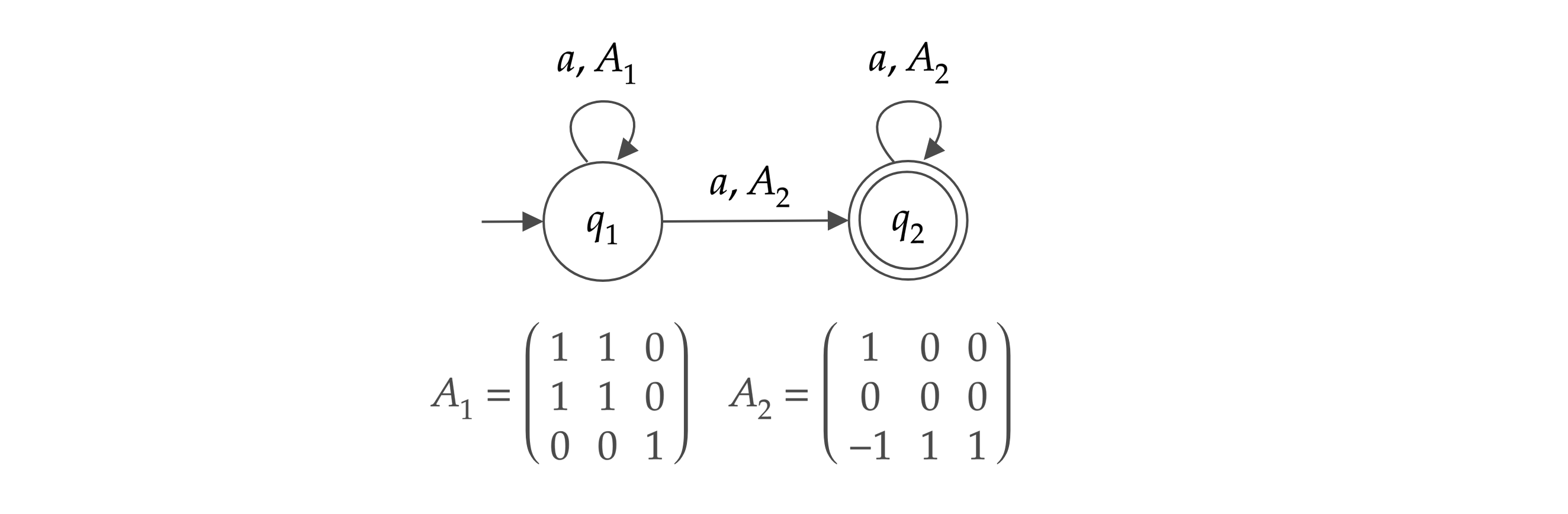}
		\caption{State transition diagram of $ \mathcal{V} $ recognizing $ \mathtt{UPOW}' $}
		\label{fig: upow2}
	\end{figure}	
\end{proof}	   

If we further restrict the matrices to have entries from the set $ \{-1,0,1\} $, then we obtain the following result.
\begin{thm}
	$ \bigcup_k\mathfrak{L}\textup{(X$ k $BCA)} \subsetneq \bigcup_k \mathfrak{L} \textup{(XBHVA($ k+1 $)}_{S_{k+1}(1)})$ where $ \textup{X} \in \{\textup{D,N,1D,1N}\} $.
\end{thm}
\begin{proof}
	Let us simulate a given X$ k $BCA by a XBHVA($ k+1 $)$_{S_{k+1}(1)} $ $ \mathcal{V} $. Let  $ \mypar{1~~1~~\cdots~~1} $ be the initial vector of $ {\mathcal{V}} $. The $ k+1 $'st entry of the vector will remain unchanged throughout the computation, which will allow the counter updates. At each step of the computation, $ {\mathcal{V}} $ will multiply the vector with the appropriate matrix $ A\in T $ where $ T \subseteq S_{k+1}(1) $ is the set of all $ (k+1)\times (k+1) $ matrices corresponding to possible counter updates. Since each counter can be decremented, incremented or left unchanged, $ |T|=3^k $. All matrices will have the property that $ A[i,i]=1  $ and $ A[k+1,k+1]=1 $. When the $ i $'th counter is incremented and decremented, then  $ A[k+1,i]=1 $ and $ A[k+1,i]=-1 $, respectively. At the end of the computation, the input will be accepted if the vector is equal to $ \mypar{1~~1~~\cdots~~1} $, which happens iff all counters have value 0.
	
	The inclusions are proper by the witness languages $ \mathtt{MPAL_2} $, which can be recognized by a DBHVA(2)$_{S_2(1)} $ by Example \ref{ex: stern} for the deterministic case and $ \mathtt{UPOW'} $, which can be recognized by a NBHVA(3)$ _{S_3(1)} $ by Theorem \ref{thm: z3}.
\end{proof}

\subsection{Non-blind Counter Automata}
The fact that the individual entries of the vector cannot be checked prevents us from simulating a non-blind counter automaton by a HVA. Since 1D2CA can recognize any recursively enumerable language, we focus on the real-time case. 

We are going to show that in the real-time deterministic case, counter automata and blind homing vector automata are incomparable. First, we give a characterization for DHVA($ k $)s when the alphabet is unary.

\begin{thm}\label{thm:unary}
	For any $k$, all languages over $\Sigma = \{a\} $ accepted by a \textup{DHVA($ k $)} are regular.
\end{thm}

\begin{proof}
	Let $ {L} $ be a unary language accepted by a DHVA($k$) $ {\mathcal{V}} $ and let $ v $ be the initial vector of $ {\mathcal{V}}$. We are going to construct a DFA recognizing $ {L} $ to prove that $ {L} $ is regular. We assume that $ {L} $ is infinite and make the following observation. Since $ {\mathcal{V}} $ has finitely many states, at least one of the accept states of $ {\mathcal{V}} $ will be accepting more than one string. Let $ w_1 $ and $ w_2 $ be the shortest strings accepted by an accept state $ q_a $ with $ |w_1|<|w_2| $. When accepting $ w_1 $ and $ w_2$, ${\mathcal{V}} $ is in state $ q_a $ and the value of the vector is equal to $ v $. After reading $ w_2 $, $ {\mathcal{V}} $ is in the same configuration as it was after reading $ w_1 $ and this configuration will be repeated inside a loop of $|w_2|-|w_1|= p $ steps. Therefore, we can conclude that all strings of the form $ a^{|w_1|+lp} $ for some positive integer $ l $ will be accepted by $ q_a $. 
	
	Between consecutive times $ q_a $ accepts a string, some other strings may be accepted by some other accept states. Let $ u $ be a string accepted by $ q_b $ with $ |w_1| < |u| < |w_2| $. Then all strings of the form $ a^{|u|+lp} $ for some positive integer $ l $ will be accepted by $ q_b$ since every time $ {\mathcal{V}} $  enters the accepting configuration at state $ q_a $, $ {\mathcal{V}} $ will enter the accepting configuration at state $ q_b $ after $ |u|-|w_1| $ steps. The same reasoning applies to any other accepting configuration inside the loop. 
	
	\newpage
	Now, let us construct a DFA $\cal  {F} $ accepting $ {L} $. $ \cal {F} $ has $ |w_1|+1+(p-1) $ states. The first $ |w_1|+1 $ states correspond to the strings of length at most $ |w_1| $ and the state $ q_{|w|} $ is an accept state for all $ w \in {L}$ that is of length at most $ |w_1| $. $ q_{|w_1|} $ and the next $ p-1 $ states $ q_{l_2},\dots,q_{l_p} $ stand for the configuration loop. States corresponding to accepting configurations inside the loop are labeled as accept states. 
	
	The transitions of the $ \cal F $ are as follows:
	\begin{align*}
	\delta(q_i,a)&=q_{i+1} \mbox{ for }  i=0,\dots,|w_1|-1 \\
	\delta(q_{|w_1|},a)&=q_{l_2} \\
	\delta(q_{l_i},a)&=q_{l_{i+1}}  \mbox{ for }  i=2,\dots,p-1 \\
	\delta(q_{l_p},a)&=q_{|w_1|} \\
	\end{align*}	
	Since $ {L} $ can be recognized by a DFA, $  {L} $ is regular. We conclude that any unary language accepted by a \textup{DHVA($ k $)}  is regular.
\end{proof}

\begin{thm}\label{th:last}
	$ \bigcup_k \mathfrak{L}\textup{(DBHVA($ k $))} $ and $ \bigcup_k \mathfrak{L} \textup{(D$ k $CA)} $ are incomparable.
	
\end{thm}
\begin{proof}
	\label{ap:counter}
	We know that $ \mathtt{MPAL_2}=\{w\#w^r|w\in\{0,1\}^*\} $ can be recognized by a DBHVA(2) by Example \ref{ex: stern}. In \cite{Pe11}, it is proven that no deterministic counter machine with $ k $ counters operating in time $ O(2^{n/k}) $ can recognize $ \mathtt{MPAL_2} $. Since we are working with real-time machines, this result applies to our case.
	
	On the other hand, it is known that the nonregular unary language $ \mathtt{UGAUSS}=\{a^{n^2+n} | n \in \mathbb{N} \}$ can be recognized by a D2CA \cite{SYS13}. By Theorem \ref{thm:unary}, we know that BHVA($ k $)'s and inherently DBHVA($ k $)'s can recognize only regular languages in the unary case. Hence, we conclude that the two models are incomparable.

\end{proof}

For the nondeterministic case we can state the following result.
\begin{thm}
	$ \bigcup_k \mathfrak{L}\textup{(NBHVA($ k $))} \nsubseteq \bigcup_k \mathfrak{L} \textup{(N$ k $BCA)} $.
\end{thm}
\begin{proof}
	It is known that no N$ k $BCA can recognize the language $ \mathtt{BIN}=\{wc0^{e_2(w)}cw^rc | w \in 1\{1,0\}^* \} $ \cite{Gr76} where $ e_2(w) $ is the integer representation of the binary number represented by $ w $, for any $ k $. A NBHVA(3) $ \mathcal{V} $ recognizing the language $ \mathtt{BIN} $ is given in Figure \ref{fig: binary}.
	
	\begin{figure}[h]
		\centering
		\includegraphics[width=1\linewidth]{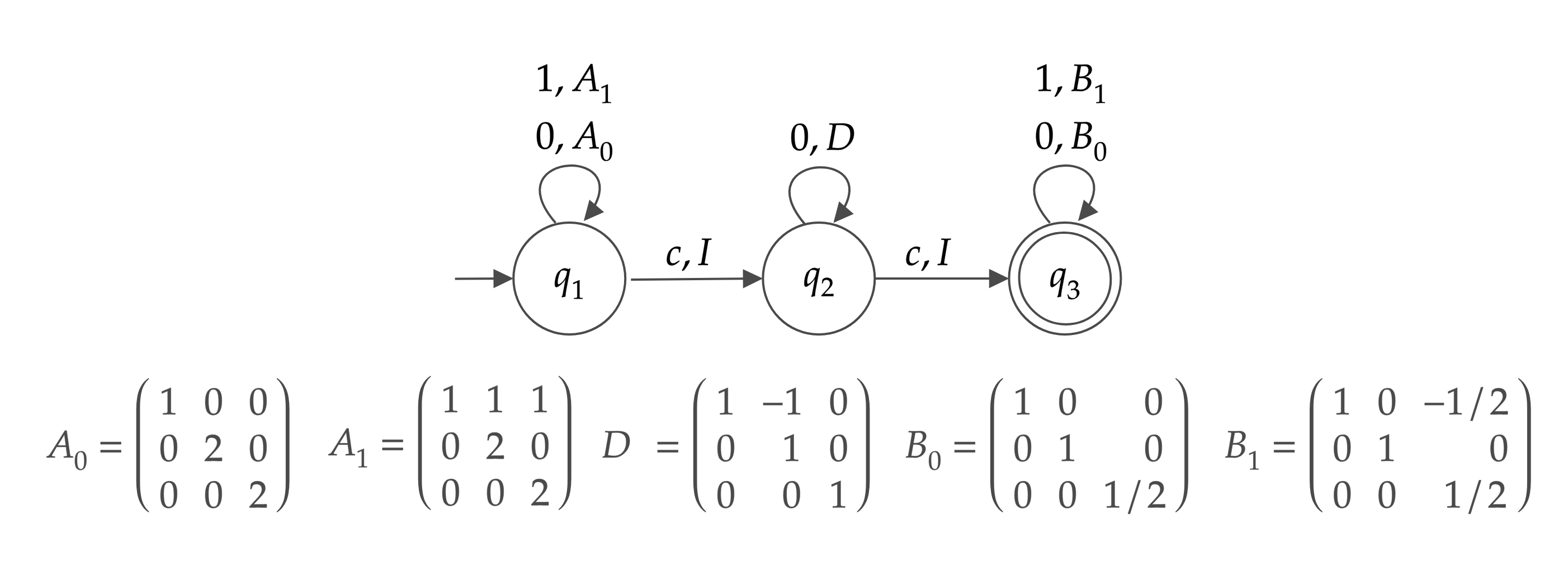}
		\caption{State transition diagram of $ \mathcal{V} $ recognizing $ \mathtt{BIN} $}
		\label{fig: binary}
	\end{figure}
	The initial vector of $ \mathcal{V} $ is $ \mypar{1~~ 0~~0}$. The idea is to use base-2 encoding to store $ e_2(w) $ in the second and the third entries of the vector with the help of the matrices $ A_0 $ and $ A_1 $. While reading the 0 sequence, the second entry is decremented using the matrix $ D $. While reading the last sequence between two $ c $'s, the inverses of the matrices used for base-2 encoding is used to check whether the sequence corresponds to the string $ w^r $.
\end{proof}

\section{Relationship with Extended Finite Automata}\label{sec: hva-efa}

In this section, we will exploit a relationship between 1NBHVA($ k $)'s and the extended finite automata over free groups to demonstrate the power of homing vector automata. 

\newpage
The two models seem to be linked in the case of finite automata over matrix groups, as the register is multiplied with a matrix at each step of the computation. Let us emphasize that the two models are different in the following sense. In a homing vector automaton, there is an initial vector $ {v} $, and the accepted strings are those which label a computation path along which the product of the sequence of matrices on the transitions is a matrix $ A $, such that  ${ v}= vA$. In the most general setting, the set of transition matrices belongs to the semigroup of rational matrices. In an accepting computation, the product of the matrices $ A $ belongs to the stabilizer subsemigroup of the set of rational matrices with respect to $ {v} $. In contrast, in an extended finite automaton over a matrix group, accepting computations are those in which $ A=I $. In that sense, one-way nondeterministic blind homing vector automata can be seen as akin to what someone who wanted to define a version of extended finite automata associated with general matrix semigroups, rather than groups, would come up with. In fact, the homing vector automaton can be seen as a special case of the rational semigroup automaton which is defined in \cite{RK10} as follows:

	Let $M$ be a monoid. An $M$-automaton is said to be \textit{with targets} if it is equipped with two subsets $ I_0, I_1 \subseteq M $ called the initial set and the terminal set respectively. An input string $ w \in \Sigma^* $ is accepted by the automaton if 
	there exists a computation from the initial state to some accepting state such that $ x_0x \in I_1 $, where
	$ x_0 \in I_0 $ and $ x  \in M$ is the content of the register of the machine after the reading of $w$. In the case that  $ I_0$ and $ I_1 $ are rational subsets of $ M $, the model is called  \textit{rational monoid automaton} defined by $M$ \cite{Re10,RK10}.  
	
	\textit{Rational semigroup automata} are defined analogously by taking $M$ as a semigroup instead of a monoid.

Note that the family of languages accepted by rational monoid automata 
where $I_0=I_1=\{1\}$ coincides with the set of languages recognized by ordinary $ M $-automata. Letting $I_0=I_1=v$ where $ v $ is the initial vector of a homing vector automaton and $ M $ to be a semigroup of matrices, one obtains homing vector automata.

%We assume a familiarity of the reader with some basic notions from free group theory (see \cite{KM79,LS77} for classical references of this topic). As was pointed out in the introduction, a  well-known theorem by Nielsen and Schreier states that every subgroup of a free group is free (Fact \ref{fact: nielsen}). In particular, for every $r$ there is a set $X$ of $r$ elements so that the subgroup generated by $X$ is isomorphic to $ \mathbf{F}_r $.
We start by looking at the case of $2 \times 2  $ matrices. Recall from Subsection \ref{Section: 23matrices} that $ \mathbf{F}_2 $ admits a representation by $ 2 \times 2 $ matrices. In particular, the group generated by the matrices
\[
M_{a}=
\mymatrix{rr}{1&2\\
	0&1\\}
~~~\mbox{and}~~~
M_{b}=
\mymatrix{rr}{1&0\\
	2&1\\}
\]
is isomorphic to $\mathbf{F}_2 $.

%We focus our attention on $ \mathbf{F}_2$. It is well known that $ \mathbf{F}_2$ admits a  representation by using matrices of
%the group of all invertible matrices  of dimension $2$ over the ring of integers.
%
%Let $n$ be a positive integer and consider the group $K_n$ of matrices generated by  
%\[ 
%M_{a}=
%\mymatrix{rr}{1&n\\
%	0&1\\},
%~~~
%M_{b}=
%\mymatrix{rr}{1&0\\
%	n&1\\}
%\]
%
%The following result holds.
%\begin{fact}\textup{\cite{KM79}}\label{thm: KM}
%	The group $K_n$ is isomorphic to  $ \mathbf{F}_2$. Moreover, if $ v = \mypar{1~~0}$, for every 
%	matrix $A$ of $K_n$ which is not a power of $M_b$, ${v} A \neq  v.$
%\end{fact}
%As a straightforward consequence, there exists a subgroup $H$ of $K_n$ which is isomorphic to $ \mathbf{F}_2$ such that:
%\begin{equation}\label{eq:fla}
%\forall \ A \in H \setminus \{I\}, \ 
%{v} A \neq  v. 
%\end{equation}
%Indeed, let $H$ be the subgroup of $K_n$ generated by $M_{a} M_{b} M_{a}^2$ and $M_{a}^2 M_{b} M_{a}$. 
%By the theorem of Nielsen and Schreier mentioned above, $H$ is freely generated by the latter two elements. In particular,  no element of $H$ equals a power of $M_b$. 
%This implies that (\ref{eq:fla}) holds for $H$. 
%Denote
%\begin{equation}\label{eq:flavar}
%\varphi : \mathbf{F}_2 \rightarrow H,
%\end{equation}
%the  isomorphism  from $\mathbf{F}_2$ onto $H$.

Using the matrix representation of $\mathbf{F}_2$, any $\mathbf{F}_2$-automaton can be simulated by a suitably defined homing vector automaton that is of dimension 2, blind, nondeterministic, and one-way. The proof is due to Flavio D'Alessandro and can be found in \cite{SSD16}.
\begin{thm}\label{thm: EFA-1NBHVA}
	$ \mathfrak{L}( \mathbf{F}_2)  \subseteq \mathfrak{L}\textup{(1NBHVA(2))}.$
\end{thm}
This allows us to draw the following conclusion about the class of languages recognized by 1NBHVA(2)'s.

\begin{cor}\label{cor: f2}
	The family of context-free languages is included in $ \mathfrak{L} \textup{(1NBHVA(2)}). $ 
\end{cor}
\begin{proof}
Since $ \mathfrak{L}(\mathbf{F}_2) = \mathsf{CF} $ by Fact \ref{fact: cf}, the result then follows by Theorem \ref{thm: EFA-1NBHVA}.
\end{proof}

In the proof of Theorem \ref{thm: EFA-1NBHVA}, the matrices used for the simulation of an $ \mathbf{F}_2 $-automaton belong to $ SL(2,\mathbb{Z}) $. In the following theorem, we show that 1NBHVAs are more powerful than the corresponding monoid automata, when the matrices are restricted to the monoid $ {M}_{2}(\mathbb{Z}) $.

\begin{thm}\label{thm: z22}
	$ \mathfrak{L}({M}_{2}(\mathbb{Z})) \subsetneq \mathfrak{L}(\textup{1NBHVA(2)} _{{M}_{2}(\mathbb{Z})}) $.
\end{thm}
\begin{proof} In Theorem \ref{thm: M2Z}, it is proven that $ {M}_{2}(\mathbb{Z}) = \mathsf{CF} $ and by Corollary \ref{cor: f2}, the inclusion follows.

	Now we are going to prove that the inclusion is proper. Let us construct a 1NBHVA(2)$_{{M}_{2}(\mathbb{Z})} $ $ \mathcal{V} $ recognizing the non-context-free language $ \mathtt{POW_r}=\{a^{2^n}b^n | n \geq 0\} $. The state diagram of $ \mathcal{V} $ is given in Figure \ref{fig: machine}.

	\begin{figure}[h]
		\caption{State transition diagram of $ \mathcal{V} $ recognizing $\mathtt{POW_r} $}
		\label{fig: machine}
		\centering
		
		\includegraphics[width=1\linewidth]{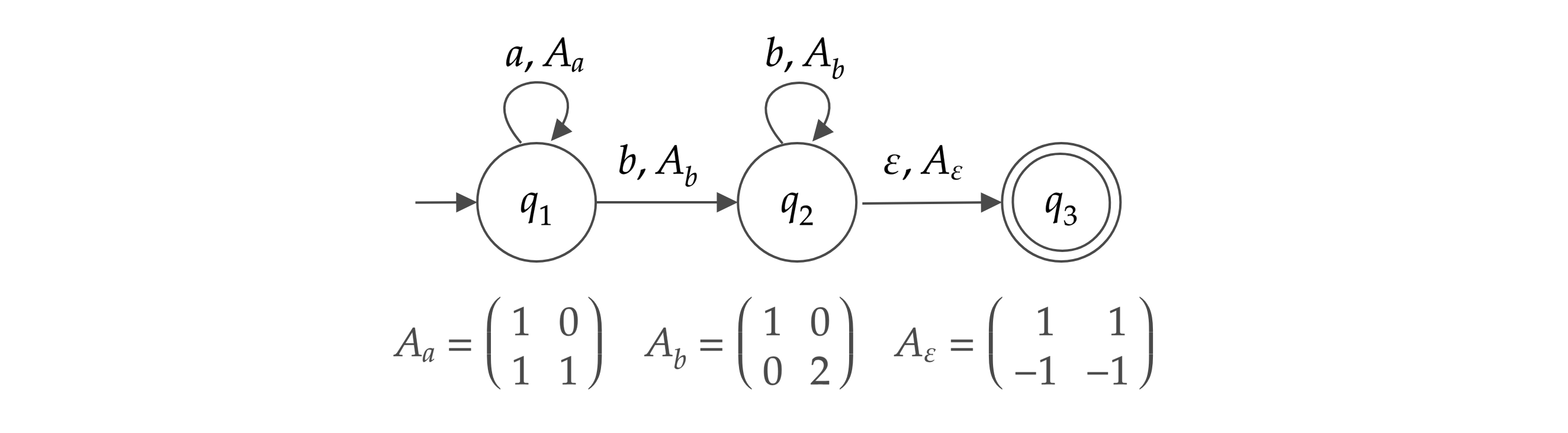}
		
	\end{figure}
%	\[ 
%	A_a=\mymatrix{ cc }{
%		1 & 0 \\
%		1 & 1
%	}  ~~
%	A_b=\mymatrix{ cc }{
%		1 & 0 \\
%		0 & 2
%	}~~
%	A_{\varepsilon}=\mymatrix{ rr }{
%		1 & 1 \\
%		-1 & -1
%	} 
%	\]
	The initial vector of $ \mathcal{V} $ is $ v=(1~~1) $. While reading $ a $ in $ q_1 $, $ \mathcal{V}  $ multiplies its vector with the matrix $ A_a $. It moves to $ q_1 $ when it scans the first $ b $ and multiplies its vector with the matrix $ A_b $ as it reads each $ b $. $ \mathcal{V} $ multiplies its vector with the matrix $ A_{\varepsilon} $ and moves to $ q_2 $.
	
	When the vector is multiplied by $ A_a $, the first entry of the vector is increased by 1 and when the vector is multiplied by $ A_b $, the second entry of the vector is multiplied by 2. Hence, after reading an input string of the form $ a^ib^j $, the vector is equal to $ \mypar{i+1~~2^j} $, as a result of the multiplication by the matrix product $ {A_a}^i{A_b}^j  $. After multiplication by $ A_{\varepsilon} $, the second entry of the vector is subtracted from the first entry and this value is stored in both entries of the vector. The value of the final vector is equal to $\mypar{1~~1} $ iff $  i+1-2^j=1 $. Hence, the accepted strings are those which are of the form $ a^{2^j}b^{j} $.
\end{proof}

For $ 3 \times 3 $ matrices, we don't have a comparability result among the classes of languages recognized by the two models. The major limitation lies in the fact that the limits of finite automata over $ 3 \times 3 $ matrix groups are not known. Furthermore, we don't know how to directly simulate a  finite automaton overa group of $ 3 \times 3 $ matrices with a 1NBHVA(3). Nevertheless, let us note that 1NBHVA(3)s defined with matrices from $ M_3(\mathbb{Z}) $ can recognize any language recognized by 1N$ k $BCAs as we showed in Theorem \ref{thm: z3}, whereas it is an open question whether $ \mathfrak{L}(M_3(\mathbb{Z})) $ includes the class of languages recognized by 1N3BCAs. 

Next we are going to demonstrate the power of 1NBHVA($ k $)s for $ k\geq 4 $. Recall that $ \mathfrak{L}(\mathbf{F}_2 \times \mathbf{F}_2) = \mathfrak{L}(SL(4,\mathbb{Z})) $ by Theorem \ref{thm: sl4z}. The proof idea of Theorem \ref{thm: sl4z} is to embed two copies of $ \mathbf{F}_2 $ into $ 4 \times 4 $ matrices. Combining this idea with the proof of Theorem \ref{thm: EFA-1NBHVA}, one can show that 1NBHVA(4)s recognize any recursively enumerable language. Since the computation of a 1NBHVA($ k $) involves vector matrix multiplications by computable numbers, it can be simulated by a Turing machine and it follows that any language recognized by a 1NBHVA($ k $) is recursively enumerable. The proof details are omitted here and can be found in \cite{SSD16}.

\begin{thm}\label{thm: EFA-1NBHVA4}
	$\mathsf{RE} = \mathfrak{L}\textup{(1NBHVA(k))} $ for $ k \geq 4 $.

\end{thm}

\section{Real-time Homing Vector Automata}\label{sec: hva-rt}
 
 In the previous section, we have seen that allowing 
 one-way access to the input tape raises nondeterministic blind homing vector automata of small vector dimension to Turing equivalence. For this reason, we will be focusing on real-time input in this section. 

Let us start with an example. We show that by allowing nondeterminism, it is possible to recognize an $\mathsf{NP}$-complete language in real-time and with matrices belonging to set $ S_5(1) $. $\mathtt{SUBSETSUM}$ is the $\mathsf{NP}$-complete language which is the collection of all strings of the form $t \#  a_1\#...\# a_n\#$, such
that $t$ and the $a_i$'s are numbers in binary notation $(1 \leq i \leq n)$, and there
exists a set $I \subseteq \{1, . . . , n\}$ satisfying $\sum_{i \in
	I}a_i=t$, where $n > 0$.  We define $$\mathtt{SUBSETSUM}_r=\{ t^r \#  a_1^r\#...\# a_n^r\#\ | \exists I \subseteq \{1, . . . , n\} \mbox{ s.t. } \sum_{i \in
	I}a_i=t\}$$ in which the binary numbers appear in reverse order. It is obvious that $\mathtt{SUBSETSUM_r} \in \mathsf{NP}$, since $ \mathtt{SUBSETSUM}  \in \mathsf{NP} $. It is possible to reduce $\mathtt{SUBSETSUM}$ to $ \mathtt{SUBSETSUM_r}$ in polynomial time by reversing the binary numbers that appear in the input. Therefore, we can conclude that $ \mathtt{SUBSETSUM_r}$ is $\mathsf{NP}$-complete.   

\begin{ex}
	$ \mathtt{SUBSETSUM_r} \in \mathfrak{L}(\textup{NBHVA(5)}) $.
\end{ex}

	We construct a NBHVA(5)$ _{M_5(1)} $ ${\mathcal{V}}$ recognizing 
	$\mathtt{SUBSETSUM_r}$. The state transition diagram of $ \mathcal{V} $ recognizing $ \mathtt{SUBSETSUM_r} $ is given in Figure $ \ref{fig: subsetsum} $.  
	\begin{figure}[h]
		\centering
		\includegraphics[width=1\linewidth]{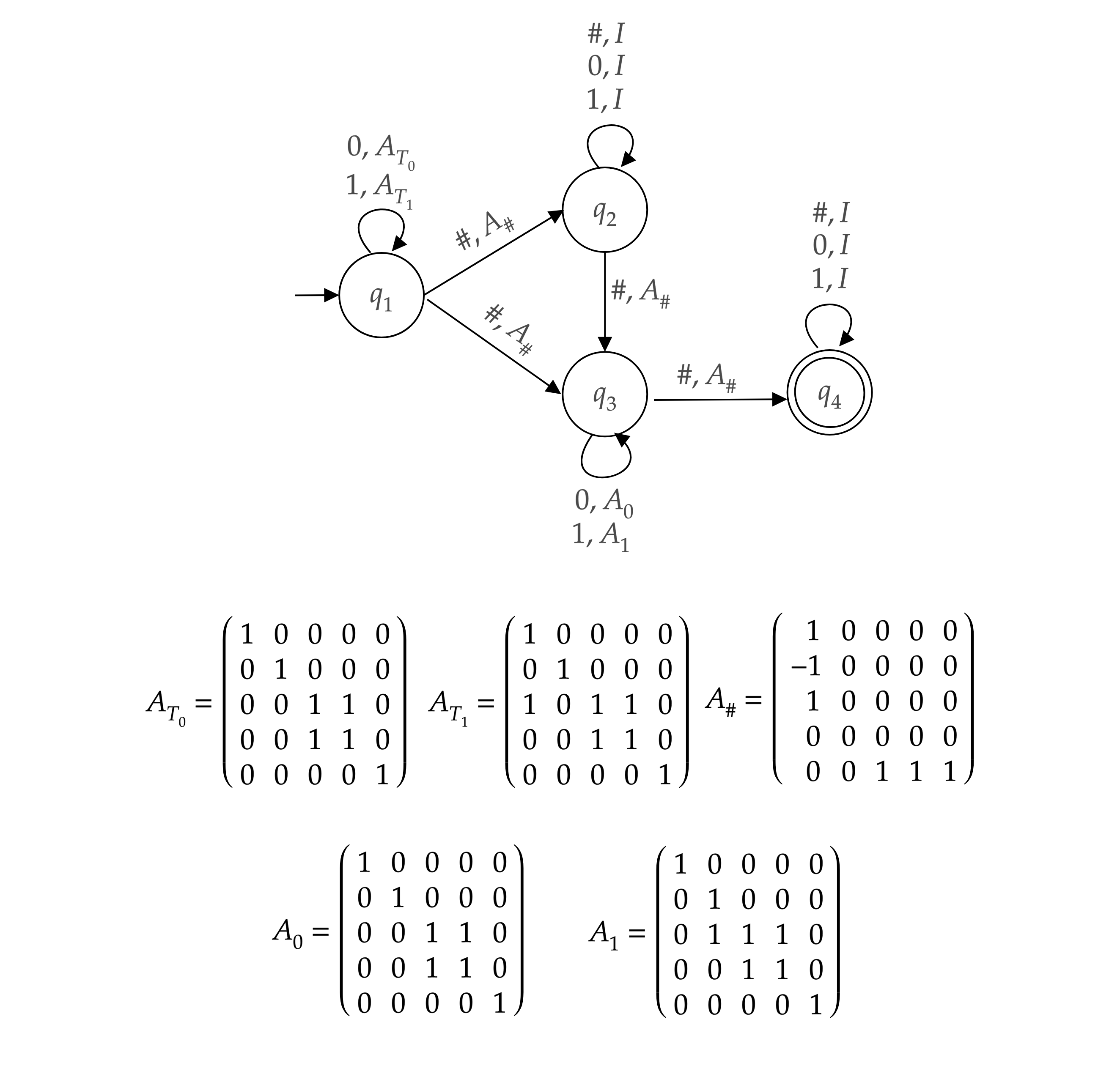}
		\caption{State transition diagram of $ \mathcal{V} $ recognizing $ \mathtt{SUBSETSUM_r} $ }
		\label{fig: subsetsum}
	\end{figure}
	
	The idea of this construction is to read the binary numbers in the string to entries of the vector, and to nondeterministically select the set of numbers that add up to $ t $.  We let the initial vector equal $ \mypar{0 ~~0~~1~~1~~1} $. We first encode $ t $ to the first entry of the vector as follows: While scanning the symbols of $t$, ${\mathcal{V}}$ multiplies the vector with the matrix $A_{T_0}$ (resp. $A_{T_1}$) for each scanned $0$
	(resp. $1$). The  powers of 2 required for the encoding are obtained by adding the third and fourth entries, which always contain identical numbers, to each other, creating the effect of multiplication by 2. When ${\mathcal{V}}$ reads a $\#$, ${\mathcal{V}}$ multiplies the vector with the matrix $ A_{\#} $ which subtracts the second entry from the first entry and resets the second entry back to 0, and the third and  fourth entries back to 1.   
%	\[
%	A_{T_0}=
%	\left[
%	\begin{array}{rrrrr}
%	1&0&0&0&0 \\
%	0&1&0&0&0\\
%	0&0&1&1&0\\
%	0&0&1&1&0\\
%	0&0&0&0&1
%	\end{array}
%	\right ]~~~~~
%	A_{T_1}=
%	\left[
%	\begin{array}{rrrrr}
%	1&0&0&0&0 \\
%	0&1&0&0&0\\
%	1&0&1&1&0\\
%	0&0&1&1&0\\
%	0&0&0&0&1
%	\end{array}
%	\right ]~~~~~
%	A_{\#}=
%	\left[
%	\begin{array}{rrrrr}
%	1&0&0&0&0 \\
%	-1&0&0&0&0\\
%	0&0&0&0&0\\
%	0&0&0&0&0\\
%	0&0&1&1&1
%	\end{array}
%	\right ] \]
%	
	In the rest of the computation, ${\mathcal{V}}$ nondeterministically decides which $a_i$'s to
	subtract from the first entry.
	Each selected $a_i$ is encoded using the same technique into the second entry of the vector. While scanning the symbols of $a_i$, ${\mathcal{V}}$ multiplies the vector with the matrix $A_{0}$ (resp. $A_{1}$) for each scanned $0$
	(resp. $1$).
	
%	$$
%	A_{0}=
%	\left[
%	\begin{array}{rrrrr}
%	1&0&0&0&0 \\
%	0&1&0&0&0\\
%	0&0&1&1&0\\
%	0&0&1&1&0\\
%	0&0&0&0&1
%	\end{array}
%	\right ]~~~~~
%	A_{1}=
%	\left[
%	\begin{array}{rrrrr}
%	1&0&0&0&0 \\
%	0&1&0&0&0\\
%	0&1&1&1&0\\
%	0&0&1&1&0\\
%	0&0&0&0&1
%	\end{array}
%	\right ]
%	.$$
%	
	
	${\mathcal{V}}$ chooses another
	$a_j$ if it wishes, and the same procedure is applied. At
	the end of the input, ${\mathcal{V}}$ accepts if the vector is equal to 
	$  \mypar{0 ~~0~~1~~1~~1}$, which requires that the first entry of the vector is equal to 0. This is possible iff there exists a set of $ a_i $'s whose sum add up to $ t $.

\subsection{Comparisons Among the Different Versions}
We start by comparing the deterministic blind and non-blind versions of our model. 

\begin{thm}\label{thm:blind}
	$ \bigcup_k \mathfrak{L} 
	\textup{(DBHVA(\textit{k}))} \subsetneq \bigcup_k \mathfrak{L} \textup{(DHVA(\textit{k}))}. $
\end{thm}
\begin{proof}
	It is obvious that any DBHVA($ k $) can be simulated by a DHVA($ k $). We are going to prove that the inclusion is proper by the witness language $ \mathtt{SUM}=\{a^nb^{a_1}a^{a_2}|n=a_1 \mbox{ or } n=a_1 + a_2, a_1 \geq 1, a_2 \geq 1\} $. Let us first construct a DHVA(2)$_{S_2(1)} $ $ {\mathcal{V}} $ recognizing $\mathtt{SUM} $. The state transition diagram of $ \mathcal{V} $ is given in Figure \ref{fig: na1a2}. The idea is to simulate a counter with the help of the matrices. Starting with the initial vector 
	$ \mypar{1~~1}	$, ${\mathcal{V}} $ multiplies the vector with the matrix $ A_+ $ for each  $ a $ it reads before the $b$'s, incrementing the first entry of the vector with each such multiplication. After finishing reading the first segment of $ a $'s, ${\mathcal{V}} $ multiplies the vector with the matrix $ A_- $, decrementing the first entry of the vector for each $ b $.
	\begin{figure}[h]
		\centering
		\includegraphics[width=1\linewidth]{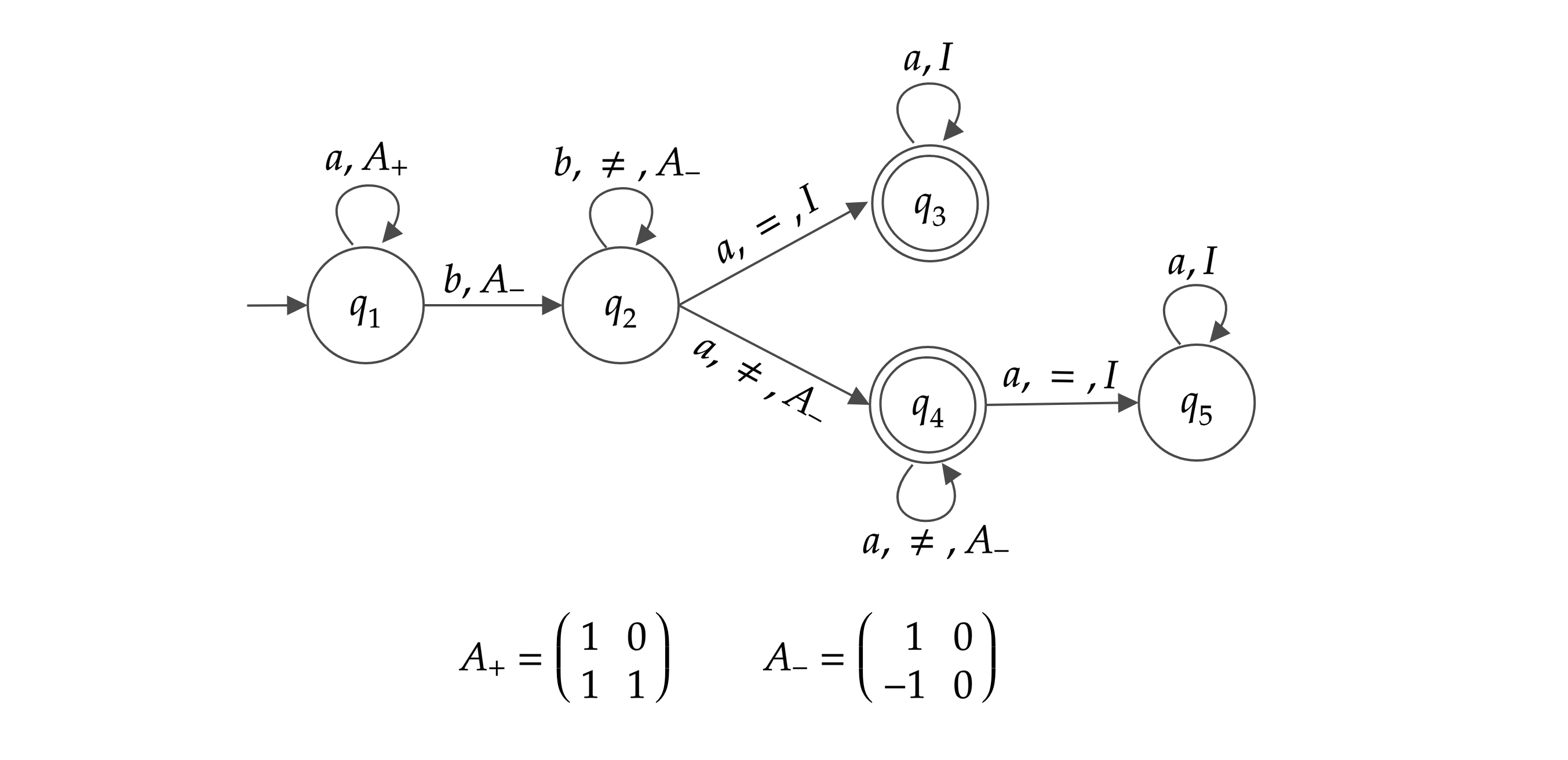}
		\caption{State transition diagram of $ \mathcal{V}  $ recognizing $ \mathtt{SUM} $ }
		\label{fig: na1a2}
	\end{figure}
	
%	$$A_{+}=
%	\left [
%	\begin{array} {rr}
%	1&0\\
%	1&1\\
%	\end{array}
%	\right ]~~~~~
%	A_{-}=
%	\left [
%	\begin{array}{rr}
%	1&0\\
%	-1&1\\
%	\end{array}
%	\right ]
%	$$ 
	
	$ {\mathcal{V}} $ checks the current value of the vector for equality to $ \mypar{1~~1}$ when reading the first $ a $. If the equality is detected, it is the case that $ n=a_1 $, and $ {\mathcal{V}} $ multiplies the vector with the identity matrix at each step for the rest of the computation. If that is not the case, ${\mathcal{V}} $ continues to multiply the vector with matrix $ A_- $ for each $ a $ after the $b$'s. The value of the vector will be equal to 
	$ \mypar{1~~1}$ at the end of the computation if and only if $ n=a_1 $ or $ n =a_1+a_2$. 
	
	Note that $\mathtt{SUM} $ can be also recognized by a DHVA(1) by using the one-dimensional matrices $ A_{+}=(2) $ and $ A_{-}=(\frac{1}{2}) $.
	
	Now we are going to show that $\mathtt{SUM} $ cannot be recognized by any DBHVA($ k $). Suppose for a contradiction that $ {L} $ is recognized by some DBHVA($ k $) $ {\mathcal{V}'} $. After reading a prefix of $ a $'s, the computation of $ {\mathcal{V}'} $ on a sufficiently long suffix of $b$'s will go through a sequence of states, followed by a state loop. Suppose that $ {\mathcal{V}'} $ is in the same state after reading two different strings $ a^nb^m $ and $ a^nb^n $, $ m<n $. Now consider the strings $u= a^nb^ma^{n-m} \in \mathtt{SUM} $ and $ w=a^nb^na^{n-m} \in  \mathtt{SUM}  $. After reading any one of these strings, $ {\mathcal{V}'} $ should be in the same accept state, and the  vector should be at its initial value. Assume that the strings in question are both extended with one more $ a $.  Since the same vector is being multiplied with the same matrix associated with the same state during the processing of that last $ a $, it is not possible for $ {\mathcal{V}'} $ to give different responses to $ 
	a^nb^na^{n-m+1}$ and $ a^nb^ma^{n-m+1}$. Noting that $ a^nb^na^{n-m+1} \in \mathtt{SUM}$, whereas $ a^nb^ma^{n-m+1} \notin \mathtt{SUM}$, we conclude that $ \mathtt{SUM} $ cannot be recognized by any DBHVA($ k $). 
\end{proof}

In the following theorem, we show that nondeterministic real-time homing vector automata are more powerful than their deterministic versions, both in the blind and nonblind cases.
\begin{thm}\label{thm:upow}
	\begin{enumerate}
		\item $ \bigcup_k \mathfrak{L}\textup{(DBHVA(\textit{k}))} \subsetneq \bigcup_k \mathfrak{L} \textup{(NBHVA(\textit{k}))}$. 
		\item $ \bigcup_k \mathfrak{L}\textup{(DHVA(\textit{k}))} \subsetneq \bigcup_k \mathfrak{L} \textup{(NHVA(\textit{k}))}$. 
	\end{enumerate}
\end{thm}
\begin{proof}
	
	It is obvious that the deterministic models can be simulated by the nondeterministic models. The inclusion is proper since $\mathtt{UPOW}' $ can be recognized by a NBHVA(2) by Theorem \ref{thm: bhva2} and every unary language recognized by a DHVA($ k $) is regular by Theorem $ \ref{thm:unary}$.
\end{proof}

A language $ {L} $ is in class $ \mathsf{TISP} $($ t(n),s(n) $) if there is a deterministic Turing machine that decides $ {L} $ within $ t(n) $ time and $ s(n) $ space where $ n $ is the length of the input. Since the numbers in the vector can grow by at most a fixed number of bits in each multiplication, a Turing machine simulating a DHVA($ k $) requires only linear space \cite{SYS13}. Since the numbers in the vector can have length $O(n)$, whereas the matrix dimensions and entries are independent of the input length $n$, multiplication of a vector and a matrix requires $ O(n) $ time for each input symbol. We can conclude that $ \bigcup_k \mathfrak{L}$(DHVA($ k $))$ \subseteq  \mathsf{TISP}( n^2,n )$.

\subsection{A Hierarchy Result}  
We will now establish a hierarchy result on real-time deterministic homing vector automata based on the dimension of the vector, when the matrix entries belong to a restricted set.

\begin{thm}\label{thm:hier}
	Then $ \mathfrak{L}\textup{(DHVA}(\textit{k}))_{S_k(m)} \subsetneq \mathfrak{L} \textup{(DHVA}(\textit{l}))_{S_l(m)}$ for $ l>(km)^k $. 
\end{thm}

\begin{proof} Using the generalized Stern-Brocot encoding, we showed that it is possible to recognize $ \mathtt{MPAL_l}=\{w\#w^r|w\in\{a_1,a_2,\dots,a_l\}^*\} $ by a DBHVA($ l $)$ _{S_l(1)} $ in Exercise \ref{ex: stern}.
	
	We are now going to show that $ \mathtt{MPAL_l}  \notin \mathfrak{L}(\textup{DHVA}(k)) _{S_k(m)} $ for $ l>(km)^k $. We first note that the value of any entry of a vector of size $ k $ can be at most $ m^{n+1}k^n $ after reading $ n $ symbols. This is possible by letting the initial vector  have $m$ in all entries, and multiplying the vector with the matrix with all entries equal to $ m $  at each step. Similarly, the smallest possible value of an entry is $ -m^{n+1}k^n  $, and so the number of possible different values for a single entry is $ 2m^{n+1}k^n+1 $.  If the machine has $ s $  states, $ s(2m^{n+1}k^n+1)^k $ is an upper bound for the number of different reachable configurations after reading $ n $ symbols. Since there are $ l^n $ strings of length $ n $ when the alphabet consists of $ l $ symbols, for large $ n $  and $ l >(km)^k $, the machine will end up in the same configuration after reading two different strings $ u $ and $ w $. This will cause the strings $ u\#w^r $ and $ w\#u^r $ which are not in  $ 
	\mathtt{MPAL_l}$ to be accepted by the machine. Therefore, we conclude that $ \mathtt{MPAL_l} \notin \mathfrak{L}(\textup{DHVA}(k))_{S_k(m)}$.
	
	Since a vector automaton with a larger vector size can trivially simulate a vector automaton with a smaller vector size, this result applies to our case.
\end{proof}

\subsection{Closure Properties}\label{sec: hva-closure}

In this section, we examine the closure properties of the class of languages recognized by real-time homing vector automata. We start with a lemma which will be useful in  our proofs. The languages mentioned below are from \cite{ISK76}.

\begin{lem}\label{lem: anb2n}
	\begin{enumerate}
		\item $ \mathtt{UNION}=\{a^nb^{n}|n \geq 0 \} \cup \{a^nb^{2n}|n \geq 0 \}  \notin \bigcup_k  \mathfrak{L} \textup{(DHVA(\textit{k}))} $.
		
		\item $ \mathtt{L_{bab}}=\{b^n(a^nb^n)^k|n,k \geq 1\} \notin \bigcup_k  \mathfrak{L} \textup{(DHVA(\textit{k}))}  $ .
		
		\item $ \mathtt{IJK}=\{a^ib^jc^k|i\neq j \mbox{ or } j > k\} \notin \bigcup_k  \mathfrak{L} \textup{(DHVA(\textit{k}))}  $ .
		
		\item $ \mathtt{UNION_c}=\{a^nb^n|n\geq 0\} \cup \{a^nb^{2n}c|n\geq 0\}  \notin \bigcup_k  \mathfrak{L} \textup{(DHVA(\textit{k}))}  $ .
	\end{enumerate}
\end{lem}
\begin{proof}
	We can show all these languages to be unrecognizable by DHVAs by applying the following common reasoning. Assume that the language $ {L} $ in question is recognized by some DHVA($ k $)  $ {\mathcal{V}} $. Since there are finitely many states, one of the states of $ {\mathcal{V}} $ will end up accepting more than one member of the language. For each language, we will focus on two such members $u$ and $v$. Note that $ {\mathcal{V}} $ is in the same configuration (since it has also returned to its initial vector) after reading both $u$ and $v$. We then append another string $x$ to both strings, selected so that $ ux \in {L} $ and $ vx \notin {L}$. The responses of $ {\mathcal{V}} $ to the  $ ux $ and $vx$ has to be identical, since it will have returned to the same configuration after processing both strings. We conclude that $ {\mathcal{V}} $ cannot distinguish between these two strings, and therefore that $ {L} \notin \bigcup_k  \mathfrak{L} \textup{(DHVA(\textit{k}))} $. All that 
	remains is to provide the strings $u$, $v$, and $x$ for the languages in the statement of the lemma. In the following, $i,j>1$ and $i\neq j$.
	
	\begin{enumerate}
		\item $u= a^ib^i $,  $v= a^jb^j $, and $ x=b^i $.
		
		\item  $u= b^ia^ib^i $, $v= b^ja^jb^j $ and $x= a^ib^i $.
		
		\item  $ u=a^ib^ic $, $ v=a^jb^jc $, and $x= c^{j-1}$ for $ i>j $.
		
		\item $u= a^ib^i $, $v= a^jb^j $, and $x= b^ic $.
	\end{enumerate}
\end{proof}
Let us note that it is possible to recognize the languages mentioned in the proofs with
DHVA($ k $)'s of smaller vector size when the vector entries are not restricted to be integers.
\begin{thm}
	\begin{enumerate}
		\item $ \bigcup_k \mathfrak{L} \textup{(DHVA(\textit{k}))} $  is closed under the following operations:
		\begin{enumerate}
			\item intersection with a regular set
		\end{enumerate}
		\item $ \bigcup_k \mathfrak{L} \textup{(DHVA(\textit{k}))} $ is not closed under the following operations:
		\begin{enumerate}
			\item union
			\item concatenation
			\item intersection
			\item star
			\item homomorphism	
			\item reversal
			\item complementation
		\end{enumerate}
	\end{enumerate}
\end{thm}	
\newpage
\begin{proof}
\hfill
	\begin{enumerate}
		\item 
		\begin{enumerate}
			\item Let $ {L}_{{\mathcal{V}}}	$ be recognized by a DHVA($ k $)  $ {\mathcal{V}}= (Q_1,\Sigma_1,\mathrm{M},\delta_1,q_1,Q_{a_1},v) $ and $
			{L}_{\mathcal{F}} $ be a language recognized by a finite state automaton $
			\mathcal{F}=(Q_2,\Sigma_2,\delta_2,$ $q_2,Q_{a_2}) $. Let us construct a DHVA($ k $) $
			\mathcal{\mathcal{V}'}=(Q,\Sigma,\mathrm{M},\delta,q_1,Q_a,v) $ recognizing $ {L}={L}_{{\mathcal{V}}} \cap
			{L}_{\mathcal{M}} $. $ \mathcal{V'} $ keeps track of the vector and the current state of $\cal  {V}
			$ as well as the current state of $ \mathcal{F}. $  Let $ Q' = Q_1 \times Q_2 $ be the state set of $
			\mathcal{\mathcal{V}'} $ and $ \Sigma=\Sigma_1 \cup \Sigma_2 $. For each $ (q_i,q_j) \in Q $, $ \sigma \in \Sigma $ and
			$\omega \in \Omega$, $ \delta((q_i,q_j),\sigma, \omega)= ((q_i',q_j'),A)$ where $\delta_1(q_i,\sigma,
			\omega)=(q_i',A)$ and $ \delta_2(q_j,\sigma)=q_j' $. $ q_1 $ is the pair $ (q_1,q_2) $ and $ Q_a $ is the set
			of pairs of states where both of the states are accept states of $ {\mathcal{V}} $ or $ \mathcal{F} $. We 
			obtain a DHVA($ k $) $ {\mathcal{V}}' $ recognizing $ {L} $.
		\end{enumerate}
		\item \begin{enumerate}
			\item Let $ {L}_1=\{a^nb^n | n \geq 0\} $ and $ {L}_2=\{a^nb^{2n}|n \geq 0\} $. $ {L}_1 $ and $ {L}_2 $ can be recognized by a DBHVA(2) which simulates a deterministic blind one-counter automaton whereas $  {L}_1 \cup  {L}_2 = \mathtt{UNION} $ cannot be recognized by any DHVA($ k $) for any $ k $ by Lemma \ref{lem: anb2n}.
			
			\item For the languages  $ {L}_1=\{a^nb^n | n \geq 0\} $ and $ {L}_2=\{a^nb^{2n}|n \geq 0\} $, ${L}_1{L}_2 \cap a^*b^*= \mathtt{UNION}  $,  which cannot be recognized by any DHVA($ k $) for any $ k $ by Lemma \ref{lem: anb2n} and Part i.a) of this theorem.
			
			\item Let $ \mathtt{L_1}=\{b^+(a^nb^n)^*| n \geq 1\} $ and $  \mathtt{L_2}=\{(b^na^n)^*b^+ | n \geq 1\} $. Both  $ {L}_1 $ and $ {L}_2 $ can be recognized by  DHVA(2)s which simulate deterministic one-counter automata, whereas  $ {L}_1 \cap \mathtt{L_2} = \mathtt{L_{bab}} = \{b^n(a^nb^n)^k|n,k \geq 1\}$ cannot be recognized by any DHVA($ k $) for any $ k $ by Lemma \ref{lem: anb2n}.
			
			\item Let $ {L}=\{a^nb^n|n \geq 0 \} \cup \{ca^nb^{2n}|n \geq 0 \}$. A  DBHVA($ 2$) $ {\mathcal{V}} $ recognizing ${L} $ branches into one of two computation paths depending on the first scanned symbol $ \sigma_1 $. If $ \sigma_1=a $, $ {\mathcal{V}} $ simulates a deterministic blind one-counter automaton recognizing $ \{a^{n-1}b^n | n \geq 0\} $ and if $ \sigma_1=c$, ${\mathcal{V}}$ simulates a deterministic blind one-counter automaton recognizing $\{a^nb^{2n}\}$. Now suppose $  {L}^* \in \bigcup_k \mathfrak{L} \textup{(DHVA(\textit{k}))} $. Then ${L}' = {L}^* \cap \{ca^ib^j|i,j \geq 0\}= \{ca^nb^n|n\geq 0\} \cup \{ca^nb^{2n}| n \geq 0\} \in \bigcup_k \mathfrak{L} \textup{(DHVA(\textit{k}))} $. A DHVA($ k $) recognizing $ {L}' $ can be easily modified to obtain a DHVA($ k $) recognizing the language $\mathtt{UNION} = \{a^nb^n|n\geq 0\} \cup \{a^nb^{2n}| n \geq 0\}   $, which is not in $ \mathfrak{L} \textup{(DHVA(\textit{k}))} $ by Lemma \ref{lem: anb2n}. 
			
			\item Let $ {L}=\{a^nb^n | n \geq 0\} \cup \{c a^{n-1}b^{2n}|n\geq 0\} $. A DBHVA($ k $) recognizing $ {L} $ works similarly to the one in part d). Now consider the homomorphism $ h $ such that $ h(a)=a $, $ h(b)=b $ and $ h(c)=a $. $ h({L})=\{a^nb^n | n \geq 0\} \cup \{a^{n}b^{2n}\} = \mathtt{UNION}  $, which cannot be recognized by any DHVA($ k $) for any $ k $ by Lemma \ref{lem: anb2n}.
			
			\item Let $ {L}=\{b^na^n | n \geq 0\} \cup \{cb^{2n}a^n|n\geq 0\} $. A DBHVA($ k $) recognizing $ {L} $ works similarly to the one in part d). Now consider the reverse of $ {L} $, $ \mathtt{UNION_c}= \{a^nb^n| n \geq 0\} \cup \{a^nb^{2n}c|n\geq 0\}$, which cannot be recognized by any DHVA($ k $) for any $ k $ by Lemma \ref{lem: anb2n}.

			\item Consider $ {L}=\{a^mb^mc^n| 0\leq m \leq n\} $, which can be recognized by a DHVA(3). $ \bar{{L}} \cap \{a^ib^jc^k |i,j,k \geq 0\} = \{a^ib^jc^k|i\neq j \mbox{ or } j > k\} =\mathtt{IJK}$ cannot be recognized by any DHVA($ k $) by Lemma \ref{lem: anb2n}.
		\end{enumerate}
	\end{enumerate}
\end{proof}

The set of languages recognized by real-time nondeterministic homing vector automata is closed under union, star and concatenation The constructions are fairly simple and omitted. 

\begin{thm}\label{thm: closed}
	\begin{enumerate}
		\item $  \bigcup_k \mathfrak{L} \textup{(DBHVA(\textit{k}))} $ is closed under the following operations:
		\begin{enumerate}
			\item intersection
		\end{enumerate}
		\item 
		$ \bigcup_k \mathfrak{L} \textup{(DBHVA(\textit{k}))} $ is not closed under the following operations:
		\begin{enumerate}
			\item union
			\item concatenation
			\item star
			\item homomorphism
			\item reversal	
			\item complementation
		\end{enumerate}
	\end{enumerate}
\end{thm}

\begin{proof}
	\hfill 
	\begin{enumerate}
		\item 
		\begin{enumerate}
			\item Let ${L}_{{\mathcal{V}}_1}$ and $ {L}_{{\mathcal{V}}_2}$ be recognized by DBHVA($ k_1$) $ {\mathcal{V}}_1=(Q_1,\Sigma_1,\delta_1,q_1,Q_{a_1},v_1) $ and DBHVA($ k_2$) $ {\mathcal{V}}_2=(Q_2,\Sigma_2,\delta_2,q_2,Q_{a_2},v_2) $, respectively. Let us construct a DBHVA($ k $) $ {\mathcal{V}}=(Q,\Sigma,\delta,q_1,Q_a,v) $ recognizing $ {L}={L}_{{\mathcal{V}}_1} \cap {L}_{{\mathcal{V}}_2} $ where $ k=k_1 + k_2 $. Let $ Q = Q_1 \times Q_2 $ be the state set of $ {V} $ and $ \Sigma=\Sigma_1 \cup \Sigma_2 $. For each $ (q_i,q_j) \in Q $ and $ \sigma \in \Sigma $, $ \delta((q_i,q_j),\sigma)= ((q_i',q_j'),A)$, where $\delta_1(q_i,\sigma)=(q_i',A_1)$, $ \delta_2(q_j,\sigma, \omega)=(q_j',A_2)$ and $ A $ is a $ k \times  k $ block diagonal matrix with $ A_1 $ and $ A_2 $ on its diagonal. $ q_1 $ is the pair $ (q_1,q_2) $, and $ Q_a $ is the set of pairs of states where both of the states are accept states of $ \mathcal{V}_1 $ or $ \mathcal{V}_2 $. The initial vector $ v 
			$ of $ {\mathcal{V}} $ is of the form $ (v_1~~v_2) $ and has dimension $ k $. $ {\mathcal{V}}$ keeps track of the current states and the current values of  both vectors by simultaneously multiplying its vector with the appropriate matrices. Since the computation is blind, the value of the vector is checked only at the end of the computation, and an input string is accepted if the vector is equal to its initial value.
		\end{enumerate}
		\item The proofs for the non-blind version also apply here. The proof for part (f) follows from the fact that 
		$ \bigcup_k \mathfrak{L} \textup{(DBHVA(\textit{k}))} $ is closed under intersection but not union.
	\end{enumerate}
\end{proof}

The set of languages recognized by real-time nondeterministic blind homing vector automata is closed under union and intersection. The construction for union is straightforward, and the construction for intersection is identical to the deterministic case.

\subsection{Stateless computation}\label{sec: hva-stateless}

Given two strings, a finite automaton is said to separate them if it accepts one and rejects the other. Introduced by Goral{\v{c}}{\'\i}k and Koubek \cite{GK86}, the string separation problem asks for the minimum number of states needed for accomplishing this task. String separation by homing vector automata and vector automata have been investigated in \cite{SYS19}. It turns out that a homing vector automaton needs at least two states to separate any pair of strings, regardless of the dimension of the vector. We are therefore motivated to examine the limitations imposed by statelessness on homing vector automata in more detail.

Stateless machines \cite{YDI08,IKO10,KMO09} have been investigated by many researchers, motivated by their relation to membrane computing and P systems  \cite{Pau00}, which are stateless models inspired from biology. While vector automata can simulate their classical states in their vectors by using longer vectors, this is not the case for homing vector automata.

Our study on stateless homing vector automata yields a characterization for the class of languages recognized by stateless real-time deterministic FAMs without equality (0-DFAMW) \cite{ISK76}. It turns out that a language is recognized by a 0-DFAMW iff it is commutative and its Parikh image is the set of nonnegative solutions to a system of linear homogeneous Diophantine equations.  When the computation is nondeterministic, then any language recognized by a stateless real-time nondeterministic FAM without equality is commutative. We conclude by providing some further examples and observations about language recognition power of stateless homing vector automata.  

\subsubsection{Observations}
The limitation of having a single state for homing vector automata leads to the acceptance of the string  $ xx $, whenever the string $ x $ is accepted. This is true since further repetition of the same input naturally carries the vector through a trajectory that ends up in its initial value. Based on this observation, we can list the following consequences:

For $ X \in \{\textup{DB,D,NB,N}\} $,
\begin{itemize} 
	\item 		If string $ x $ is accepted by a 0-XHVA $\cal V$, then any member of $ \{x\}^* $ is also accepted by $\mathcal{V}$.
	\item 	If all members of a language $L $ are accepted by a 0-XHVA $\mathcal{V}$, then any string in $ L^* $ is also accepted by $\mathcal{V}$.
	\item 	If language $L $ is recognized by a 0-XHVA, then $ L = L^* $.
	\item 	0-XHVAs cannot recognize any finite language except $ {L}_\varepsilon $.
\end{itemize}

We can further make the following observation for deterministic models.

\newpage
\begin{lem}
	\label{lemma: diff}
	If the strings $ w_1 $ and $ w_1w_2 $ are accepted by a \textup{0-XHVA} $ V$ where $ X \in \{\textup{DB,D}\} $, then the string $ w_2$ is also accepted by $ \mathcal{ V} $.
\end{lem}
\begin{proof}
	After reading $ w_1 $, the value of the vector is equal to its initial value. Since $ w_1w_2 $ is also accepted by $ \mathcal{V} $, reading $ w_2 $ results in acceptance when started with the initial vector.
\end{proof}

For the unary case we have the following.
%
%\begin{cor}\label{cor: diff}
%	If the strings $ a^i $ and $ a^j $ ($ 1<i<j $) are accepted by a \textup{0-DHVA} $ H$, then the string $ a^{j-i} $ is also accepted by $ H $.
%\end{cor}

\begin{thm}\label{thm: gcd}
	If the strings $ a^i $ and $ a^j $ ($ 1<i<j $) are accepted by a \textup{0-XHVA} $ \mathcal{V}$ where $ X \in  \{\textup{DB,D}\}   $, then the string $ a^{\gcd(i,j)} $ is also accepted by $ \mathcal{V} $.
\end{thm}
\begin{proof}
	It is well known that for any positive integers $ i,j $, there are two integers $ l_i $ and $ l_j $ such that $ i l_i + j l_j = \gcd(i,j) $. Assume that $ l_i $ is positive and $ l_j $ is non-positive. (The other case is symmetric.) Note that $ il_i \geq j(-l_j) $. The strings $ a^{j(-l_j)} $ and $ a^{i l_i} $ are accepted by $ H $. By Lemma \ref{lemma: diff}, the string  $ a^{i l_i-j(-l_j)} $,  which is $ a^{\gcd(i,j)} $, is also accepted by $ \mathcal{V} $. 
	
\end{proof}

\begin{cor}
	\label{cor: gcd}
	If the strings $ a^i $ and $ a^j $ ($ 1<i<j $) are accepted by a \textup{0-XHVA} $ H$ where $ X \in  \{\textup{DB,D}\}  $ and $ \gcd(i,j)=1 $, then $ \mathcal{V} $ recognizes $ a^* $.
\end{cor}

Let us now investigate the case where the set of matrices is commutative.

Let $ L \in \Sigma^* $ be a language. The \textit{commutative closure} of $ L $ is defined as $ com(L)=\{x \in \Sigma^* | \phi(x) \in \phi(L) \} $. A language is \textit{commutative} if $ com(L)=L $.

\begin{thm}\label{thm: comm}
	If a language $L $ is recognized by a \textup{0-XBHVA} $ \mathcal{V} $ where $ X\in  \{\textup{D,N}\}  $ with a commutative set of matrices, then $L$ is commutative.
\end{thm}
\begin{proof}
	Let $w \in L  $ and suppose that the string $w=w_{[1]}w_{[2]}\cdots w_{[n]}$ is accepted by $ \mathcal{V} $. Let $A_1 A_2\cdots A_n$ be the product of the matrices labeling the computation such that $$v A_1 A_2 \cdots  A_n=v$$ where $v$ is the initial vector of $ \mathcal{V} $. Since the matrices are commutative, then for any permutation $\tau$, $$A_1A_2 \cdots  A_n=A_{\tau(1)} A_{\tau(2)} \cdots  A_{\tau(n)}.$$
	This leads to the acceptance of the string $w'=w_{[\tau(1)]} w_{[\tau(2)]} \cdots w_{[\tau(n)]} $ since $$vA_{\tau(1)} A_{\tau(2)} \cdots A_{\tau(n)}=v.$$ Hence, if $w$ is accepted by $ \mathcal{V} $, then any string obtained from $ w $ by permuting its letters is also accepted by $ \mathcal{V} $. Any string $ x $ with $ \phi(x)=\phi(w) $ is in $ L $ and we conclude that $ L $ is commutative. 	
\end{proof}

When the computation is not blind, then the class of languages recognized is no longer commutative. The language of balanced strings of brackets $ \mathtt{DYCK} $ can be recognized by 0-DHVA(1) as follows. Starting with the initial vector $ (1) $, for each left bracket the vector is multiplied by $ (2) $. As long as the vector is not equal to (1), for each right bracket, the vector is multiplied by ($\frac{1}{2}$). If the vector is equal to $ (1) $ and the next symbol is a right bracket, then the vector is set to (0).

\begin{cor}\label{cor: commreverse}
	If a language $L $ is recognized by a \textup{0-XBHVA} $ \mathcal{V} $ where $ X \in  \{\textup{D,N}\}  $  with a commutative set of matrices, then $L=L^r$.	
\end{cor}
\begin{proof}
	Suppose that $w \in L$. Then it is clear that $w^r$ will be also accepted by $ \mathcal{V} $ by Theorem \ref{thm: comm} and $w^r \in L$. Since for every string $ w $ it is true that  $w \in L $ if and only if $ w^r \in L$, we conclude that $L=L^r$.
\end{proof}

\subsubsection{Regular languages}

Let $ \cal F$ be an $ n $-state deterministic finite automaton. Without loss of generality, we can enumerate its states as $ q_1,\ldots,q_n $ where $q_1$ is the initial state. Moreover we can denote $ q_i $ by $ e_i $, which is the basis vector in dimension $ n $ having 1 in its $ i $'th entry, and 0 otherwise. Besides, for each symbol $\sigma $, we can design a zero-one matrix, say $ A_\sigma $, that represents the transitions between the states, i.e. $ A_\sigma[i,j] = 1 $ if and only if $ \cal F $ switches from $ q_i $ to $ q_j $ when reading symbol $\sigma $. Thus, the computation of $  \cal F $ on an input, say $ w $, can be traced by matrix-vector multiplication:  
\[
e_j =  e_1 A_{w[1]} A_{w[2]} \cdots A_{w[|w|]}
\]
if and only if $ \cal F $ ends its computation in $ q_j $ after reading $ w $.

Based on this representation, we can easily observe that if a language $ L $ is recognized by an $ n $-DFA whose initial state is the single accept state, then $ L $ is recognized by a 0-DBHVA($n$).

Let us give some examples of stateless HVAs recognizing regular languages.

\begin{ex} \label{ex: 1}
For $ k >1 $, $\mathtt{AB_k}^*=\{a^kb^k\}^* \in \mathfrak{L}(\textup{0-DBHVA(2k)})$.
	\end{ex}
 Let us construct a 0-DBHVA(2$ k $)  $ \mathcal{V}_k $ recognizing $\mathtt{AB_k}^*$. The initial vector of $ \mathcal{V}_k  $ is $ v=(1 ~~ 0 ~~ \cdots ~~ 0) $. 
For each $ a $, the value in the $ i $'th entry of the vector is transferred to the $ (i+1) $'st entry when $ 1 \leq i \leq k $, and, the rest of the entries are set to 0. For each $ b $, the value in the $ (i+k) $'th entry of the vector is transferred to the $ (i+k+1 \mod 2k) $'th entry, and the rest of the entries are set to 0, ($ 1 \leq i \leq k $). Thus, we return  to the initial vector if and only if after some number of $ (a^kb^k) $ blocks have been read.

\begin{ex} \label{ex: 2}
$ \modm = \{ a^i \mid i \mod m \equiv 0 \} \in \mathfrak{L}(	\textup{0-DBHVA($m$)}) $.
\end{ex}
It is easy to observe that any unary $n$-state DFA whose initial state is the single accept state can recognize either $ L_\varepsilon $ or $ \modm $ for some $m\leq n$. Hence, for any $ m>0 $, the language $ \modm  $ is recognized by a \textup{0-DBHVA($m$)}. 

Note that if it is allowed to use arbitrary algebraic numbers in the transition matrices, then for every $ m>0 $, the language $ \modm $ is recognized by a \textup{0-DBHVA(2)} with initial vector $ v = \mypar{1 ~~ 0} $ that multiplies its vector with the matrix $$ A_m = \mymatrix{rr}{ \cos  \frac{2 \pi}{m} & -\sin  \frac{2 \pi}{m} \\ \sin  \frac{2 \pi}{m} & \cos  \frac{2 \pi}{m} } $$ for each scanned symbol. (The entries of $A_m$ are algebraic for any $m$ \cite{Le93}.)    

One may ask whether all regular languages $ L $ satisfying $ L= L^* $ can be recognized by stateless HVAs. We provide a negative answer for the deterministic model. The language $ \modtwothree $ is defined as $ \{ a^2,a^3 \}^* = a^* \setminus \{a\} $, satisfying that $ \modtwothree = \modtwothree^* $. $ \modtwothree $ cannot be recognized by any \textup{0-DHVA} since a 0-DHVA which accepts $ a^2 $ and $ a^3 $, also accepts $ a $ by Lemma \ref{lemma: diff}.

\subsubsection{Stateless 1-dimensional HVAs}

We now focus on stateless HVAs whose vectors have dimension 1 and demonstrate some results on stateless FAMWs. Note that stateless FAMs do not process the end-marker  $\dollar$ by definition, since their single state is also an accept state and the computation ends once $\dollar$ is scanned in an accept state. 

We start by comparing the class of languages recognized by stateless versions of $ k $BCAs, FAMWs, and BHVA(1)s. The equivalence between $ k $BCA and FAMWs and BHVA(1)s does not hold immediately in the stateless case. The reason is that the counters can only be updated by the set $ \{-1,0,1\} $ in a single step since additional states are needed to update the counters by arbitrary integers (see Fact \ref{fact: counter}). Furthermore, the capability of multiplication with negative rational numbers brings additional power to the stateless DBHVA(1)s. 

\begin{thm}\label{thm: famhva}
	$  \mathfrak{L} \textup{(0-XFAMW)} \subsetneq \mathfrak{L}\textup{(0-XBHVA(1))}$ where $ X  \in \{\textup{D,N}\} $.
\end{thm}
\begin{proof}
	Let $ \mathtt{EVENAB}=\{a^nb^n |~n = 2k \mbox{ for some }k\geq 0 \}$. The following 0-DBHVA(1) recognizes $ \mathtt{EVENAB} $: The register is multiplied by $ (-2) $ and $ (1/2) $ when the machine reads an $ a $ and $ b $ respectively. 
	
	Suppose that there exists some  0-XFAMW  recognizing $ \mathtt{EVENAB} $. Let $ m_a $ and $ m_b $ be the positive rational numbers that are multiplied by the register upon reading $ a $ and $ b $. Since $ aabb\in \mathtt{EVENAB} $, it is true that $ m_a^2m_b^2=1 $. Since both $ m_a $ and $ m _b$ are positive, it is not possible that $ m_am_b=-1 $. It follows that $ m_am_b=1 $, in which case the string $ ab $ is accepted and we get a contradiction. Hence we conclude that $ \mathtt{EVENAB} $ cannot be recognized by any 0-XFAMW.
\end{proof}

When the register is multiplied with only positive rational numbers, then we have $  \mathfrak{L} $(0-XFAMW)=$\mathfrak{L}$(0-XBHVA(1))$ _\mathbb{Q^+} $.

For real-time and blind machines, we can state the following result.

\begin{cor}\label{thm: commu}
	If $ L \in \mathfrak{L}(\textup{0-XFAMW})$ where $X \in  \{ \textup{D,N} \} $, then $ L $ is commutative.
\end{cor}
\begin{proof}
By Theorem \ref{thm: famhva}, $ L $ is accepted by a 0-XBHVA(1). Since multiplication in dimension 1 is commutative, the result follows by Theorem \ref{thm: comm}.
\end{proof}

Let us recall Fact \ref{fact: bounded}, which states that a bounded language is accepted by a 1NFAMW iff it is semilinear \cite{ISK76}. In the next theorem, we prove a similar result and characterize the class of languages recognized by 0-DFAMWs. We show that any language recognized by a 0-DFAMW is commutative and semilinear. Furthermore, any commutative language whose Parikh image is the set of nonnegative solutions to a system of linear homogenous Diophantine equations can be recognized by a 0-DFAMW. 

\begin{thm}
$ L \in \mathfrak{L}(\textup{0-DFAMW}) $ iff $ L $ is commutative and $\phi(L) $ is the set of nonnegative integer solutions to a system of linear homogeneous Diophantine equations.
\end{thm}
\begin{proof} Let  $ L $ be a language over the alphabet $ \Sigma=\{\sigma_1,\dots,\sigma_n\}  $ recognized by a 0-DFAMW $ \mathcal{V} $. Let $A=\{a_1,a_2,\dots,a_n\}$ be the set of rational numbers such that the register is multiplied with $a_i$ upon reading $\sigma_i$. Let $\{p_1,p_2,\dots ,p_k\}$ be the set of prime factors of the denominators and the numerators of the rational numbers in $A$. Then each $ a_i $ can be expressed as 
$$
a_i=\frac{p_1^{x_{1_i}}p_2^{x_{2_i}}\cdots p_k^{x_{k_i}}}{p_1^{y_{1_i}}p_2^{y_{2_i}}\cdots p_k^{y_{k_i}}} .
$$  

If a string $w$ is accepted by $\mathcal{V}$, then the value of the register should be equal to 1 after reading $ w $, which is possible only if
$$
a_1^{w_{|\sigma_1|}}a_2^{w_{|\sigma_2|}}\cdots a_n^{w_{|\sigma_n|}}=1.
$$
This leads to the following system of linear Diophantine equations in $ n$ variables.

\begin{align*}
(x_{1_1}-y_{1_1})w_{|\sigma_1|}+(x_{1_2}-y_{1_2})w_{|\sigma_2|}+\dots + (x_{1_n}-y_{1_n})w_{|\sigma_n|}&=0\\
(x_{2_1}-y_{2_1})w_{|\sigma_1|}+(x_{2_2}-y_{2_2})w_{|\sigma_2|}+\dots + (x_{2_n}-y_{2_n})w_{|\sigma_n|}&=0\\
&\vdots\\
(x_{k_1}-y_{k_1})w_{|\sigma_1|}+(x_{k_2}-y_{k_2})w_{|\sigma_2|}+\dots + (x_{k_n}-y_{k_n})w_{|\sigma_n|}&=0\\
\end{align*}

%	together with the restriction that $ w_{|\sigma_1|}(x_{l_1})+w_{|\sigma_2|}(x_{l_2})+\dots + w_{|\sigma_n|}(x_{l_n})$ is an even number.

For $ j=1,\dots ,k $, the $ j $'th equation is stating that the exponent of $ p_j $ is equal to 0 after reading $ w$. 

\begin{comment}
In matrix form, we can state it as

$$
\mymatrix{llll}{
x_{1_1}-y_{1_1}&x_{1_2}-y_{1_2}&\dots & x_{1_n}-y_{1_n}\\
x_{2_1}-y_{2_1}&x_{2_2}-y_{2_2}&\dots & x_{2_n}-y_{2_n} \\
\vdots&&& \\
x_{k_1}-y_{k_1}&x_{k_2}-y_{k_2}&\dots & x_{k_n}-y_{k_n} \\  
}\mymatrix{c}{
w_{|\sigma_1|}\\
w_{|\sigma_2|}\\
\vdots \\
w_{|\sigma_n|}\\
}=\mymatrix{c}{
0\\
0\\
\vdots \\
0\\
}, 
$$

\noindent or $ As=0 $. 

\end{comment}

Hence, we can conclude that the Parikh images of the accepted strings are the nonnegative solutions to a system of linear homogeneous Diophantine equations. $ L $ is commutative by Theorem \ref{thm: commu}. (One can further conclude that $ L $ is semilinear.)  

\newpage

For the converse, suppose that we are given a commutative language $ L $ over the alphabet $ \Sigma = \{\sigma_1,\dots,\sigma_n\} $. Let $ T $ be the set of Parikh images of the strings in $ L $. $ T $ is the set of nonnegative solutions to a system of, say, $ k $ linear homogeneous Diophantine equations in $ n $ unknowns,

\begin{align*}
b_{11}t_1 + b_{12}t_2+\dots +b_{1n}t_n&=0\\
b_{21}t_2 + b_{22}t_2+\dots +b_{2n}t_n&=0\\
&\vdots\\
b_{k1}t_1+b_{k2}t_2+\dots + b_{kn}t_n&=0\\
\end{align*} 
where $ \mypar{t_1 ~~~t_2 ~~~\dots~~~ t_n} \in T$. 
%\begin{align*}
%A[11]w_{|\sigma_1|}+A[12]w_{|\sigma_2|}+\dots +A[1n]w_{|\sigma_n|}&=0\\
%A[21]w_{|\sigma_1|}+A[22]w_{|\sigma_2|}+\dots +A[2n] w_{|\sigma_n|}&=0\\
%&\vdots\\
%A[k1]w_{|\sigma_1|}+A[k2]w_{|\sigma_2|}+\dots + A[kn]w_{|\sigma_n|}&=0\\
%\end{align*} 

We construct a 0-DFAMW $ \mathcal{V} $ recognizing $ L $ as follows. We choose a set of $ k $ distinct prime numbers, $\{p_1,p_2,\dots,p_k\} $. When $ \mathcal{V}$ reads $ \sigma_i $, the register is multiplied by $$ a_i= p_1^{b_{1i}}p_2^{b_{2i}}\cdots p_k^{b_{ki}}.$$ Suppose that a string $ w $ is accepted by $ \mathcal{V} $. Then 
$$
a_1^{w_{|\sigma_1|}}a_2^{w_{|\sigma_2|}}\cdots a_n^{w_{|\sigma_n|}}=1.
$$
The product is equal to 1 if all of the exponents of the prime factors are equal to 0, that is when $ \mypar{w_{|\sigma_1|} ~~~w_{|\sigma_2|}~~~ \dots ~~~ w_{|\sigma_n|} }\in T$. Hence we see that the set of accepted strings are those with Parikh image in $ T $. Since $ L $ is commutative, any string $ w $ with $ \phi(w) \in T $ belongs to $ L $ and we conclude that $ \mathcal{V} $ recognizes $ L $.

\end{proof}

Note that $ \mathfrak{L} $(0-DFAMW)=$\mathfrak{L}$(0-1DFAMW), since a 0-1DFAMW that has an instruction to stay on some input symbol cannot read the rest of the string. The reasoning of Lemma \ref{lem: detrt} applies. 

\subsubsection{Additional Results on Stateless Homing Vector Automata}

One sees that the nonregular language $\mathtt{AB} = \{ a^nb^n | n \geq 0 \} $ cannot be recognized by any 0-NHVA($ k $) for any $ k $, since $ \mathtt{AB} \neq \mathtt{AB}^* $. On the other hand, $ \Leq = \{ x \in \{a,b\}^* \mid |x|_a = |x|_b \} $ can be recognized by a 0-DBHVA(1) with initial vector $  (1) $, and transition matrices $ A_a = (2) $ and $ A_b = (1/2) $. It is possible to recognize the star of $\mathtt{AB} $, even in the deterministic and blind computation mode with a stateless homing vector automaton. Note that $ \mathtt{AB}^* $ cannot be recognized by any 1NFAMW \cite{ISK76}. The proof is due to Abuzer Yakaryılmaz and can be found in \cite{SYS19}.

\begin{thm}
	\label{thm:0HVA-upal}
	The language $\mathtt{AB}^*= \{ \{a^nb^n\}^* \mid n \geq 0 \}$ is recognized by a \textup{0-DBHVA(10)}.
\end{thm}

Nondeterministic HVAs are more powerful than their deterministic variants in terms of language recognition in general. In the next theorem, we show that this is also true for the stateless models.

\begin{thm}
	\begin{enumerate}
		\item $\mathfrak{L} 
		\textup{(0-DBHVA)} \subsetneq  \mathfrak{L} \textup{(0-NBHVA)}. $
		\item $  \mathfrak{L} 
		\textup{(0-DHVA)} \subsetneq  \mathfrak{L} \textup{(0-NHVA)}. $
	\end{enumerate}
\end{thm}
\begin{proof}
	Let us construct a 0-NBHVA(1) $ \mathcal{V} $ recognizing $\mathtt{LEQ}=\{x \in \{a,b\}^* \mid |x|_a \leq |x|_b\}$. Starting with the initial vector $(1) $, $ \cal V $ multiplies its vector with $ A=(2) $ for each $a$ and with $B_1=(1/2)$ or $B_2=(1) $ for each $b$ nondeterministically.
	
	Suppose that $\mathtt{LEQ}$ can be recognized by a 0-DHVA($k$) $ \cal V' $. The strings $ w_1=b $ and $ w_2=ba $ are accepted by $ \cal V' $. By Lemma \ref{lemma: diff}, the string $w_3=a$ is also accepted by $ \cal V' $. We obtain a contradiction, and conclude that $\mathtt{LEQ}$ cannot be recognized by any 0-DHVA($k$).
\end{proof}

\newpage
Now let us compare the language recognition power of 1 and 2 dimensional stateless HVAs.

\begin{thm}\label{thm:12}
\begin{enumerate}
\item $ \mathfrak{L} 
\textup{(0-DBHVA(1))} \subsetneq \mathfrak{L} \textup{(0-DBHVA(2))}. $
\item $ \mathfrak{L} 
\textup{(0-NBHVA(1))} \subsetneq  \mathfrak{L} \textup{(0-NBHVA(2))}. $
\end{enumerate}
\end{thm}
\begin{proof}
Note that the language $ \mathtt{AB_1}^*=(ab)^* $ can be recognized by a 0-DBHVA(2) by Example \ref{ex: 2}. Assume that $ \mathtt{AB_1}^* $ is recognized by a 0-NBHVA($1$). Since $\mathtt{AB_1}^*$ is not equal to its reverse, by Corollary \ref{cor: commreverse} we get a contradiction. We conclude that $\mathtt{AB_1}^*$ cannot be recognized by any 0-NBHVA(1). 	
\end{proof}

Now, we compare the blind and non-blind versions of one dimensional HVAs.

\begin{thm}
\begin{enumerate}
\item 	$ \mathfrak{L} 
\textup{(0-DBHVA(1))} \subsetneq  \mathfrak{L} \textup{(0-DHVA(1))}. $
\item 	$ \mathfrak{L} 
\textup{(0-NBHVA(1))} \subsetneq  \mathfrak{L} \textup{(0-NHVA(1))}. $

\end{enumerate}
\end{thm}
\begin{proof}
Let us construct a 0-DHVA(1) $ \mathcal{V} $ recognizing $ \mathtt{AB_1}^*=(ab)^* $. The initial vector is equal to (1). For each scanned $ a $, $ \mathcal{V} $ multiplies its vector with $ A_==(2) $ if it is equal to its initial value, and with $ A_{\neq}=(0) $ otherwise. For each scanned $ b $, $ \mathcal{V} $ multiplies its vector with $ B_==(0) $ if it is equal to its initial value, and with $ B_{\neq}=(1/2) $ otherwise. $\mathtt{AB_1}^*$ cannot be recognized by any 0-NBHVA(1) as we saw in the proof of Theorem \ref{thm:12}.
\end{proof}

Let us look at some closure properties for the stateless models. All of the classes associated with the stateless models are closed under the star operation since for any language recognized by a stateless homing vector automaton, it is true that $ L=L^* $.

\begin{thm}
	\begin{enumerate}
		\item $\bigcup_k \mathfrak{L}\textup{(0-DBHVA($ k $))} $ is closed under the following operation:
		\begin{enumerate}
			\item intersection
		\end{enumerate}
		\item $\bigcup_k  \mathfrak{L}\textup{(0-DBHVA($ k $))} $ and $ \mathfrak{L}\textup{(0-DHVA)} $ are not closed under the following operations:
		\begin{enumerate}
			\item union
			\item complement
			\item concatenation
		\end{enumerate}
	\end{enumerate}
\end{thm}
\begin{proof} 
	\hfill 
	\begin{enumerate}
		\item 	\begin{enumerate}
			\item The proof of Theorem \ref{thm: closed} that $\bigcup_k \mathfrak{L}\textup{(DBHVA($ k $))} $ is  closed under intersection is also valid for stateless models.
		\end{enumerate}	

				\item 	\begin{enumerate}
			\item The languages $ \mathtt{MOD_2}$ and $ \mathtt{MOD_3}$ are recognized by 0-DBHVA(2) and 0-DBHVA(3) respectively, by Example \ref{ex: 2}. Their union cannot be recognized by any 0-DHVA($ k $), since a 0-DHVA($ k $) accepting the strings $a^2$ and $a^3$ should also accept the non-member string $a$ by Lemma \ref{lemma: diff}.
\item The complement of the language $ \modm $ ($m>1$)is not recognized by any 0-NHVA($ k $), since $ \overline{\modm} $ contains $ a $ and any 0-NHVA($ k $) accepting $ a $ accepts any member of $a^*$.
\item The concatenation of  $ \mathtt{MOD_2}$ and $ \mathtt{MOD_3}$, $\modtwothree$, cannot be recognized by any 0-DHVA.
\end{enumerate}
\end{enumerate}
\end{proof}

$\bigcup_k  \mathfrak{L}\textup{(0-NBHVA(k))} $ and 
$ \bigcup_k \mathfrak{L}\textup{(0-NHVA(k))} $ are not closed under complement. $ \bigcup_k \mathfrak{L}\textup{(0-NBHVA(k))} $ is closed under intersection. The proofs are identical.

\section{Open Questions}\label{sec: hva-end}

What can we say about the relationship between real-time homing vector automata and one-way homing vector automata? We conjecture that one-way nondeterministic blind homing vector automata are more powerful than their real-time versions. Our candidate language is $ \mathtt{UPOW}=\{a^{2^n}|n\geq 0\} $, which can be recognized by a 1NBHVA(2). Note that when the machine in consideration is deterministic and blind, the real-time and one-way versions are equivalent in power. 

Can we show a separation result between the class of languages recognized based on the set of matrices used during the transitions of a homing vector automaton? 

Can we show a hierarchy result between the classes of languages recognized by deterministic homing vector automata of dimensions $ k $ and $ k+1 $ for some $ k>1 $, maybe when the matrix entries are restricted to the set $ \{-1,0,1\} $? Consider the family of languages $ \mathtt{POW}(k)=\{ a^nb^{k^n}|n \geq 0\} $. We conjecture that it is not possible to recognize $ \mathtt{POW}(k)$ with a homing vector automaton of dimension less than $ k+1 $ with the restricted set of matrices.

Let $ G $ be a group of $ k\times k $ matrices. Can we always construct a $ \textup{1NBHVA($ k $)}_G $ recognizing the same language as a given $ G $-automaton? (Note that we have proven that this is the case for $\textup{1NBHVA(2)}_{\mathbf{F}_2} $ and $ \mathbf{F}_2 $-automaton.)

Suppose that one can always find a suitable initial vector $ {v} $ such that for every $ A\in G $ except the identity matrix, $ {v}A\neq v $. Then one could construct the required  $ \textup{1NBHVA($ k $)}_G $ from the given $G $-automaton directly. For which groups $ G $ is it always possible to find such a vector?

What can we say about the reverse direction? For instance, is every language recognized by some $ \textup{1NBHVA(2)}_{\mathbf{F}_2} $ context-free?

We proved that 1NBHVA($ k $)s are more powerful than extended finite automata when both are defined over $ 2 \times 2 $ integer matrices. Is this result still true when both models are defined over $ 3 \times 3  $ integer matrices?

Do 0-NHVA($ k $)s recognize every regular language $L$ satisfying $ L=L^*$? Is there any nonregular language $L$ satisfying $L=L^*$ that cannot be recognized by any stateless HVA?

We proved that any language recognized by a 0-NFAMW is commutative. What can we say about the non-blind case?

We gave a characterization for the class of languages recognized by 0-DFAMWs. Can we give a similar characterization for the non-blind and nondeterministic models?

\chapter{VALENCE GRAMMARS AND VALENCE AUTOMATA}\label{chapter: valence}
In this chapter, we are going to focus on context-free valence grammars, valence automata and valence pushdown automata. 

We start by stating the formal definitions in Section \ref{sec: val-def}. In Section \ref{sec: val-pda}, we show that valence pushdown automata and finite valence automata are equivalent in terms of language recognition power, using the well known equivalence between pushdown automata and finite automata over polycyclic monoids. Then, we extend some results proven for valence automata in \cite{Re10} to context-free valence grammars and give the properties of the set of languages generated by context-free valence grammars in Section \ref{sec: val-new}. We state some open questions in Section \ref{sec: val-open}.

\section{Definitions}\label{sec: val-def}

	Let $ \mathcal{G}=(N,T,P,S) $ be a  \textit{context-free grammar}  where $ N$ is the set of variables, $ T $ is the terminal alphabet, % such that $ N \cap T = \emptyset $, 
$ P \subseteq N \times (N \cup T)^* $ is the set of rules or productions, and $ S \in N $ is the start symbol.  We will denote by $ \Rightarrow $ and $ \Rightarrow^* $ the step
derivation relation and its regular closure respectively.
$ {L}(\mathcal{G})$ denotes the language $\{w\in T^*: S \Rightarrow^* w\} $ of words generated by $G$. 
  
Given a monoid $M$, a \textit{context-free valence grammar over  $ M $} \cite{FS02} is a five-tuple  $ \mathfrak{G}=(N,T,P,S,M) $, where $ N,T,S  $ are defined as before and $ P \subseteq N \times (N \cup T)^* \times M $ is a finite set of objects called \textit{valence rules}. Every valence rule can be thus described as an ordered pair
$ p=(A \rightarrow \alpha, m) $, where $(A \rightarrow \alpha)\in N \times (N \cup T)^*$ and  $ m\in M $. The element $m$ is called the valence of $ p $. 

The step derivation $(\Rightarrow)$ of the valence grammar is defined as follows:  $ (w,m) \Rightarrow (w',m') $ if there exists a valence rule $ (A \rightarrow \alpha,n) $ such that $ w=w_1 Aw_2 $ and $ w'=w_1 \alpha w_2 $ and $ m'=mn$. 
The regular closure of $ \Rightarrow$ will be denoted by $ \Rightarrow^*$.
A derivation of $\mathfrak{G}$ will be said {\em successful} or {\em valid} if it is of the form $(S,1) \Rightarrow^* (w,1)$, that is, it
transforms the pair $(S,1)$ into the pair $(w,1)$, after finitely many applications of the step derivation relation. 
%	Then $ \Rightarrow^*$ is defined as: if $ (w,m), (w',m')\in T^*\times M$, either if $ (w,m)=(w',m')$ or there exists a finite sequence of
%	elements of $T^*\times M$, $(w_0, m_0), \ldots, (w_k, m_k)$, such that $ (w_0, m_0)=(w,m)$, $(w_k, m_k)= (w',m')$, and, for 
%	every $i=0,\ldots, k-1$, $(w_i, m_i)\rightarrow (w_{i+1}, m_{i+1})$. In particular, one has that $m' = m$
The language generated by $\mathfrak{G}$ is the set of all the words $w$ of $T^*$ such that $(S,1) \Rightarrow^* (w,1)$. 
%A derivation is called valid if the valence of the derivation is mapped to the identity element of the monoid at the end of the computation.

A context-free valence grammar is said to be  \textit{regular} if all of its rules are right-linear, that is, every valence rule $ (A \rightarrow \alpha,n) $ is such that
$\alpha = uX,$ where $u\in T^*$ and $X\in N$. 
The language families generated by context-free and
regular valence grammars over $ M $ are denoted by $ \mathfrak{L}(\textup{Val}, \textup{CF},M)$ and  $ \mathfrak{L}(\textup{Val}, \textup{REG},M)$, respectively.

Let $  {\cal A}=(Q, \Sigma, \Gamma, \delta, q_1, Q_a) $ be a one-way nondeterministic pushdown automaton. Given a monoid $M$, the \textit{nondeterministic valence pushdown automaton} \textup{(\textit{valence} PDA) over $ M $} is the model of computation obtained from $ \cal A$ as follows:  
with every transition of $ \cal A$  is assigned an element of $M$, called the {\em valence} of the transition. Then the valence of an arbitrary computation 
is defined as the product of the valences of all the transitions of the computation (taken in the obvious order). % is equal to the identity element of $ M $. 
A word of $\Sigma ^*$ is said to be {\em accepted} by the model if there exists an accepting computation for the word whose valence is the identity of $M$.
The set of all the accepted words is defined as the language accepted by the valence pushdown automaton. 
%	 The family of languages accepted by valence automata over $M $ and valence \textup{PDA} over $ M $ are denoted as  $ \mathfrak{L}(\textup{Val}, \textup{NFA},M)$ and  $ \mathfrak{L}(\textup{Val},\textup{PDA},M)$ respectively. 
The family of languages accepted by valence \textup{PDA} over $ M $ is denoted  $ \mathfrak{L}(\textup{Val},\textup{PDA},M)$.
It is worth noticing that the equivalence between valence pushdown automata and valence context-free grammars does not hold for an arbitrary monoid. 
However a result of \cite{FS02} shows that the equivalence is true %$ \mathfrak{L}(\textup{Val},\textup{PDA},M) = \mathfrak{L}(\textup{Val}, \textup{CF},M)$
if $ M $ is a commutative monoid.

In the case that pushdown automaton $  \mathcal{A}$ is a finite state automaton, the corresponding valence model is called the
\textit{valence automaton over $ M $}. This model coincides with the $ M $-automaton, and the family $ \mathfrak{L}(\textup{Val}, \textup{NFA},M)$ of languages accepted by valence automata is exactly $\mathfrak{L}(M)$. The equivalence between valence automata and regular valence grammars is proven in \cite{FS02}.
%one can prove that $ \mathfrak{L}(\textup{Val}, \textup{NFA},M) = \mathfrak{L}(\textup{Val}, \textup{REG},M)$ \cite{FS02}. 
%The correspondence between finite automata and regular grammars still holds for valence automata and regular valence grammars and $ \mathfrak{L}(\textup{Val}, \textup{NFA},M) = \mathfrak{L}(\textup{Val}, \textup{REG},M)$ \cite{FS02}. 
%Hence, the theory of regular valence grammars and valence automata coincide with the theory of monoid automata. 
%The equivalence between pushdown automata and context-free grammars does not hold for the valence case and the equality $ \mathfrak{L}(\textup{Val},\textup{PDA},M) = \mathfrak{L}(\textup{Val}, \textup{CF},M)$ is true when $ M $ is a commutative monoid \cite{FS02}. 

Introduced by \cite{Pa80}, the valence grammars have been studied by various authors including \cite{FS97,Ho01,FS02}. 
A through study of several remarkable structural properties of 
the languages generated by the corresponding valence grammars has been done in \cite{FS02},
over arbitrary monoids and in particular over commutative groups. In \cite{FS02}, the following fact is proven.

\begin{fact}\textup{\cite{FS02}}\label{fact: fs}
Context-free valence grammars over finite or commutative monoids are not stronger
than context-free valence grammars over finite or commutative groups.		
\end{fact}

It is left open in \cite{FS02} whether context-free grammars with valences from finite monoids are more powerful than ordinary context-free grammars and the question is answered in \cite{Ze11}.

\begin{fact} \textup{\cite{Ze11}}
Context-free grammars over finite monoids recognize exactly the class of context-free languages. 
\end{fact}

Now we are going to present some definitions about monoids and semigroups.

Let $ S $ be a semigroup. Given subsets $ A $ and $ B $ of semigroups $ S $, $ AB $ is the set $\{ab | a\in A, b\in B\} $. An \textit{ideal} $ \mathtt I $ of a semigroup $ S $ is a subset of $ S $ with the property that $ S\mathtt{I}S \subseteq \mathtt{I} $.

The binary relation $ \rho_{\mathtt I} $ defined by 
\[ 
a \rho_{\mathtt I} b \iff \mbox{ either $ a=b $ or both $ a $ and $ b $ belong to } {\mathtt I}
\]
is a congruence. The equivalence classes of $ S \textup{ mod }\rho_{\mathtt I}$ are $ {\mathtt I} $ itself and every one-element set $ \{x\} $ with $ x \in S\setminus {\mathtt I} $. The quotient semigroup $ S/\rho_{\mathtt I} $ is written as $ S / {\mathtt I} $ and is called the \textit{Rees quotient semigroup}   \cite{Ho95}. 
\[ 
S /{\mathtt I} = \{{\mathtt I}\} \cup \{\{x\}| x \in S \setminus {\mathtt I} \}
\]

A semigroup is called \textit{simple} if it contains no proper ideal. A semigroup $ S $  with a zero element is called 0-$ simple $ if the only ideals of $ S $ are $ \{0\} $ and $ S $ itself, and $ SS \neq \{0\} $.

\section{Equivalence of Finite and Pushdown Automata with Valences}\label{sec: val-pda}
In this section, we are going to prove that valence pushdown automata are only as powerful as valence automata. 

Let us recall the definition of polycyclic monoid which we have introduced in Subsection \ref{sec: linear}. Let $ X $ be a finite alphabet and let $X^*$ be the free monoid of words over $X$. For each symbol $ x\in X $, let $P_x$ and $Q_x$ be functions from $X^*$ into $X^*$
defined as follows:  for every $u\in X^*$,
\begin{align*}
P_x(u)= ux,\quad %P_x:&~ X^* \rightarrow X^* ~~~~~w \mapsto wx \\
Q_x(ux)=u.%Q_x:&~ X^*x \rightarrow X^* ~~~~ wx \mapsto w.
\end{align*}
The monoid of all partial functions on $X^*$ generated by the
set of functions $\{P_x,\, Q_x \ | x \in X   \}$ is called the \textit{polycyclic monoid} on $X$. We will focus on the polycyclic monoid of rank 2, which will be denoted by $ \textup{P}_2 $, since it contains every polycyclic monoid of countable rank.

\begin{thm}\label{thm: pdap2} For any monoid $ M $, 
	$\mathfrak{L}(\textup{Val},\textup{PDA},M) = \mathfrak{L}(\textup{Val},\textup{NFA},\textup{ P}_2 \times M ) $.
\end{thm}
\begin{proof}Let $ {L}\in \mathfrak{L}(\textup{Val},\textup{PDA},M)$ and $\mathcal{\mathfrak{P}}=\{Q,\Sigma,X,\delta,q_1, Q_a,M\} $ be a valence PDA recognizing $ {L}$. We know that a PDA with stack alphabet $ X $ is equivalent to a valence automaton over $ P(X) $. Hence, $ \mathfrak{P} $ can be seen as an \textup{NFA} where two distinct valences (one in $P(X) $ and one in $ M $) are assigned to each transition. An equivalent valence automaton $ \mathcal{E}= \{Q,\Sigma,P(X) \times M,\delta',q_1, Q_a\} $ can be constructed, where a valence from the monoid $P(X) \times M  $ is assigned to each transition. Recall that the partial functions $ Q_a $ and $ P_b $ model the operations of popping $ a $ and pushing $ b $ respectively. A transition of $ \mathfrak{P} $ of the form $ (q',b,m) \in \delta(q,\sigma,a) $ where $ a,b \in X_{\varepsilon} $, $ q,q' \in Q $, $ \sigma \in \Sigma_{\varepsilon} $ and $ m\in M $ can be expressed equivalently as $(q', ( Q_aP_b , m ) ) \in \delta'(q,\sigma) $ where $ ( Q_aP_b , m )  \in P(X) \times M  $. 
	
	A string is accepted by $ \mathcal{E} $ if and only if the product of the valences labeling the transitions in $ \mathcal{E} $ is equal to $ ( 1,1 )  $, equivalently when the product of the valences labeling the transitions in $ \mathfrak{P} $ is equal to the identity element of $ M $ and the stack is empty. Since any polycyclic monoid is embedded in $ \textup{P}_2 $, we conclude that $ {L}\in \mathfrak{L}(\textup{Val},\textup{NFA},\textup{P}_2\times M ) $.
	
	Conversely, let $ {L}\in \mathfrak{L}(\textup{Val},\textup{NFA},\textup{P}_2\times M ) $ and let $ \mathcal{E} = \{Q,\Sigma,\textup{P}_2\times M,\delta,q_1, Q_a\}$ be a valence automaton over $ \textup{P}_2\times M $ recognizing $ {L} $. Suppose that $ ( p, m ) \in \textup{P}_2\times M  $ is a valence labeling a transition of $ \mathcal{E} $. The product of the labels of a computation which involves a transition labeled by the zero element of $ \textup{P}_2 $ cannot be equal to the identity element. Hence we can remove such transitions. Any nonzero element $ p $ of $ \textup{P}_2 $ can be written as $$ Q_{x_{1}}Q_{x_{2}}\dots Q_{x_{n}} P_{y_{1}}P_{y_{2}}\dots P_{y_{o}} $$ for some $ n,o \in \mathbb{N} $ and $ x_{i}, y_{i} \in X_{\varepsilon} $, after canceling out elements of the form $ P_aQ_a $ and $ P_bQ_b $, where $ X=\{a,b\} $ is the generator set for $ \textup{P}_2 $. The product can be interpreted as a series of pop operations followed by a series of push operations performed by a PDA, without consuming any input symbol. Hence, an equivalent valence PDA $ \mathfrak{P}=\{Q',\Sigma,X,\delta',q_1, Q_a,M\} $ can be constructed where a valence from $ M $ is assigned to each transition. Let $ (q', ( p, m )) \in \delta(q,\sigma) $ where  $ q,q' \in Q $, $ \sigma \in \Sigma_{\varepsilon} $, $( p, m ) \in \textup{P}_2\times M  $ and $ p=Q_{x_{1}}Q_{x_{2}}\dots Q_{x_{n}} P_{y_{1}}P_{y_{2}}\dots P_{y_{o}} $  be a transition in $ \mathcal{E} $. In $ \mathfrak{P} $, we need additional $ n+o $ states $ \{q_1, \dots ,q_{n+o}\} \notin Q $ and the following transitions to mimic that specific transition of $ \mathcal{E} $. 
	\begin{align*}
	\vspace{-0.1in}
	(q_{1}, \varepsilon,m) &\in \delta'(q,\sigma,x_{1})  \\
	(q_{2}, \varepsilon,1) &\in \delta'(q_{1},\varepsilon,x_{2})\\
	&~\vdots\\
	(q_{n+1}, \varepsilon,1) &\in \delta'(q_{n},\varepsilon,x_{n})\\
	(q_{n+2}, y_{1},1) &\in \delta'(q_{n+1},\varepsilon,\varepsilon )\\
	(q_{n+3}, y_{2},1) &\in \delta'(q_{n+2},\varepsilon,\varepsilon )\\
	&~\vdots\\
	(q', y_{o},1) &\in \delta'(q_{n+o},\varepsilon,\varepsilon )\\
	\end{align*}

	A string is accepted by $ \mathfrak{P} $ if and only if the product of the valences labeling the transitions in $ \mathfrak{P} $ is equal to the identity element of $M $ and the stack is empty, equivalently when the product of the valences labeling the transitions in $ \mathcal{E} $ is equal to $ ( 1,1 ) $. We conclude that $ {L}\in \mathfrak{L}(\textup{Val},\textup{PDA},M)$.
\end{proof}

Note that when $ M $ is commutative, the equality $\mathfrak{L}(\textup{Val},\textup{CF},M) = \mathfrak{L}(\textup{Val},\textup{NFA},\textup{P}_2\times M ) $ also holds. 
\begin{cor}
	Let $ M $ be a polycyclic monoid of rank 2 or more. Then  $\mathfrak{L}(\textup{Val},\textup{PDA},M) $ is the class of recursively enumerable languages.
\end{cor}
\begin{proof}
	It is known that $ \mathfrak{L}(\textup{Val},\textup{NFA}, M \times M  ) $ is the class of recursively enumerable languages \cite{Ka09} when $ M $ is a polycyclic monoid of rank 2 or more. Since $\mathfrak{L}(\textup{Val},\textup{PDA},M) = \mathfrak{L}(\textup{Val},\textup{NFA},\textup{P}_2\times M ) $, by Theorem \ref{thm: pdap2}, the result follows. 
\end{proof}

\section{Context-free Valence Languages}\label{sec: val-new}

It is known that the class of languages generated by regular valence grammars and the class of languages recognized by valence automata coincide \cite{FS02}. In this section, we are going to prove that the results proven in \cite{Re10} which hold for valence automata and therefore regular valence grammars, also hold for context-free valence grammars. Although the proofs are almost identical, they are presented here for completeness. Note that the same proofs can be also adapted to valence PDA.

In \cite{Re10} Proposition 4.1.1, it is shown that the elements belonging to a proper ideal of a monoid do not have any use in the corresponding monoid automaton. We show that the same result holds for context-free valence grammars.

\begin{prop}\label{prop: ideal}
	Let $ {\mathtt I} $ be a proper ideal of a monoid $ M $. Then $ \mathfrak{L}(\textup{Val},\textup{CF},M) = \mathfrak{L}(\textup{Val},\textup{CF},M/{\mathtt I}).$
\end{prop}
\newpage 
\begin{proof}	Let $ {L} \in \mathfrak{L}(\textup{Val},\textup{CF},M)  $ and let $ \mathfrak{G} $ be a context-free valence grammar over the monoid $ M $ such that $ {L}(\mathfrak{G})={L} $. The product of the valences which appear in a derivation containing a rule with valence $ x \in {\mathtt I}  $, will itself belong to $ {\mathtt I} $. Since $ {\mathtt I} $ is a proper ideal and $ 1\notin {\mathtt I} $, such a derivation is not valid. Hence any such rules can be removed from the grammar and we can assume that $ \mathfrak G $ has no such rules. For any $ x_1,x_2,\dots,x_n \in M \setminus {\mathtt I} $, it follows that $ x_1\dots x_n=1 $ in $ M $  if and only if $ \{x_1\}\{x_2\}\dots \{x_n\}=\{1\} $ in $ M/{\mathtt I} $. Let $\mathfrak{G}' $ be the context-free grammar with valences in $ M/{\mathtt I} $, obtained from $ \mathfrak{G} $ by replacing each valence $ x \in M$ with $ \{x\} $. It follows that a string $ w $ has a valid derivation in $ \mathfrak{G} $ if and only if the product of the valences is mapped to $ \{1\} $ in $ \mathfrak{G}' $. Hence  $ {L}(\mathfrak{G}')={L} $.
	
	Conversely let $ {L} \in \mathfrak{L}(\textup{Val},\textup{CF},M/{\mathtt I})  $ and let $ \mathfrak{G}' $ be a context-free grammar over the monoid $ M/{\mathtt I} $ such that $ {L}(\mathfrak{G}')={L} $. Suppose that there exists a valid derivation consisting of a rule with $ {\mathtt I}$ as the valence. Then the product of the valences of the whole derivation will be $ {\mathtt I} $, which is not possible. Let $ \mathfrak G $ be the context-free grammar with valences in $ M$, obtained from $ \mathfrak{G}' $ by replacing each valence $ \{x\} \in M/{\mathtt I}$ with $ x $. Since  $ \{x_1\}\{x_2\}\dots \{x_n\}=\{1\} $ in $ M/{\mathtt I} $ if and only if  $ x_1\dots x_n=1 $ in $ M $, a string $ w $ has a valid derivation in $ \mathfrak{G} $ if and only if the product of the valences is mapped to $ \{1\} $ in $ \mathfrak{G}' $. Hence  $ {L}(\mathfrak{G})={L} $. 
\end{proof}

Let $ S $ be a semigroup. $ S $ is the \textit{null semigroup} if it has an absorbing element zero and if the product of any two elements in $ S $ is equal to zero. A null semigroup with two elements is denoted by $ \textup{O}_2 $.

The following corollary is analogous to  \cite{Re10} Cor. 4.1.2.
\begin{cor}\label{cor: simple} 
	For every monoid $ M $, there is a simple or 0-simple monoid $ N $ such
	that $ \mathfrak{L}(\textup{Val},\textup{CF},M) =   \mathfrak{L}(\textup{Val}, \textup{CF},N) $.
\end{cor}
\begin{proof} 
	If $ M $ has no proper ideals then it is simple. Otherwise, let $ \mathtt{I} $ be the union of all proper ideals of $ M $ and let $ N=M/ \mathtt{I} $. We can conclude from the proof of Cor. 4.1.2 \cite{Re10} 
	\begin{comment}	
	``	Otherwise, let $ I $
	be the union of all the proper ideals of $ M $. Then $ I $ is an ideal and, since the identity  element 1 does not lie in any proper ideal, $ 1 \notin I $ and $ I $ is a proper ideal of $ M $. Set $ N = M/I $ and assume for a contradiction that there exists some $ J \subseteq  N, $ a proper
	non-zero ideal. But if $ J $ is an ideal of $ N $, $ J′ = \{x \in M |\{x\} \in J\} \cup I  $ must be a
	proper ideal of $ M $. But this contradicts our assumption that $ N $ was exactly the result
	of removing all proper ideals from $ M $, and so $ J = N $ and $  N$ has no proper non-zero ideals.''
	
	\end{comment} 
	that $ N^2 = {0} $ or $ N $ is 0-simple. If $ N^2 = {0} $, then $  N $ is  $ \textup{O}_2 $ and the semigroup  $ \textup{O}_2 $ does not add any power to the grammar since it does not even contain the identity element. Hence, $ \mathfrak{L}(\textup{Val},\textup{CF},  \textup{O}_2) = \mathfrak{L}(\textup{Val},\textup{CF},\{1\}) $ where $ \{1\} $ is the trivial monoid which is simple. In the latter case $ N $ is 0-simple and by Proposition \ref{prop: ideal}, $ \mathfrak{L}(\textup{Val},\textup{CF},M) = \mathfrak{L}(\textup{Val},\textup{CF},M/I) =  \mathfrak{L}(\textup{Val}, \textup{CF},N) $.
\end{proof}

Proposition 4.1.3 of \cite{Re10} states that a finite automaton over a monoid with a zero element is no more powerful then a finite automaton over a version of the same monoid from which the zero element has been removed, in terms of language recognition. The result is still true for context-free valence grammars since the same proof idea applies. The following notation is used: $ M^0 = M \cup \{0\} $ if $ M $ has no zero element and $ M^0=M $ otherwise.  
\begin{prop}\label{prop: zero}
	Let $ M $ be a monoid. Then $ \mathfrak{L}(\textup{Val},\textup{CF},M^0) =\mathfrak{L}(\textup{Val},\textup{CF},M) $.
\end{prop}
\begin{proof}
	Since $ M \subseteq M^0 $, it follows that $ \mathfrak{L}(\textup{Val},\textup{CF},M) \subseteq \mathfrak{L}(\textup{Val},\textup{CF},M^0) $. Suppose $ {L} \in \mathfrak{L}(\textup{Val},\textup{CF},M^0) $ and let $ \mathfrak{G} $ be a context-free valence grammar with valences in $ M^0 $ and $ {L}(\mathfrak{G})={L} $. Note that a valid derivation cannot contain a rule with a zero valence since otherwise the product of the valences would be equal to zero. Any such rules can be removed from $ \mathfrak{G} $ to obtain $ \mathfrak{G}' $, a context-free grammar with valences in $ M $, without changing the language, and $ {L} \in {L}(\mathfrak{G}') $.
\end{proof}

In the case $ |X|=1 $, the monoid $ P(X) $ coincides with the well-known structure of \textit{bicyclic monoid} which will be denoted by B.

\begin{fact}\textup{\cite{Re10}}\label{fact: simple}
	A simple (0-simple) monoid with identity $ 1 $ is either a group (respectively, a group with 0 adjoined) or contains a copy of the bicyclic monoid as a
	submonoid having 1 as its identity element.
\end{fact}

\begin{thm}\label{thm: mg}
Let $ M $ be a monoid. Then either $ \mathfrak{L}(\textup{Val},\textup{CF},M) = \mathfrak{L}(\textup{Val},\textup{CF},G) $ for some group
$ G $, or $ \mathfrak{L}(\textup{Val},\textup{CF},N) \subseteq \mathfrak{L}(\textup{Val},\textup{CF},\textup{B}) $.	
\end{thm}
\begin{proof}
Let $ M $ be a monoid. By Corollary \ref{cor: simple}, then  $ \mathfrak{L}(\textup{Val},\textup{CF},M) =   \mathfrak{L}(\textup{Val},\textup{CF},N)$ for some simple or 0-simple monoid. By Fact \ref{fact: simple}, $ N $ is either a group (or a group with zero adjoined) or contains a copy of the bicyclic monoid. If $ N $ is a group, then the result follows. If $ N $ is a group with zero adjoined, then by \ref{prop: zero}, $\mathfrak{L}(\textup{Val},\textup{CF},N) = \mathfrak{L}(\textup{Val},\textup{CF},G)$ for some group $ G $. Otherwise, it should be the case that $ \mathfrak{L}(\textup{Val},\textup{CF},N) \subseteq \mathfrak{L}(\textup{Val},\textup{CF},\textup{B}) $. 
\end{proof}

Now we are ready to prove the main theorem of the section which will allow us to determine the properties of the set of languages generated by context-free valence grammars. We need the following proposition which is the grammar analogue of Proposition 1 of \cite{Ka06}.

\begin{prop}\label{prop: kambites}
	Let $ M $ be a monoid, and suppose that $ {L} $ is accepted by a context-free valence grammar over $ M $. Then there exists a finitely generated submonoid $  N $ of $ M $ such that $ {L} $ is accepted by a context-free valence grammar over $ N $.
\end{prop}
\begin{proof}
	There are only finitely many valences appearing in the rules of a grammar since the set of rules of a grammar is finite. Hence, the valences appearing in derivations are from the submonoid $ N $ of $ M $ generated by those elements. So the grammar can be viewed as a context-free valence grammar over $ N $.
\end{proof}

Recall that a group $ G $ is \textit{locally finite} if every finitely generated subgroup of $ G $ is finite and \textit{periodic} if every element of the group has finite order.

\begin{thm}\label{thm: characterize}
	Let $ M $ be a monoid. Then $ \mathfrak{L}(\textup{Val},\textup{CF},M) $ either
	\begin{enumerate}[(i)]
		\item equals $  \mathsf{CF} $,
		\item contains $ \mathfrak{L}(\textup{Val},\textup{CF},\mathbb{Z}) $ ,
		\item contains  $ \mathfrak{L}(\textup{Val},\textup{CF},\textup{B}) $ or
		\item is equal to $ \mathfrak{L}(\textup{Val},\textup{CF},G) $ for $ G $ an infinite periodic group which is not locally finite.		
	\end{enumerate}
\end{thm}
\newpage
\begin{proof}
	Let $ M $ be a monoid. Then either $ \mathfrak{L}(\textup{Val},\textup{CF},M) = \mathfrak{L}(\textup{Val},\textup{CF},G) $ for some group
	$ G $, or $ \mathfrak{L}(\textup{Val},\textup{CF},N) \subseteq \mathfrak{L}(\textup{Val},\textup{CF},\textup{B}) $. Then in the former case $ (iii) $ holds. In the latter case $ \mathfrak{L}(\textup{Val},\textup{CF},M) = \mathfrak{L}(\textup{Val},\textup{CF},G) $ for some group $ N $.
	
%	Let $ M $ be a monoid. By Corollary \ref{cor: simple}, $ \mathfrak{L}(\textup{Val},\textup{CF},M) =   \mathfrak{L}(\textup{Val}, \textup{CF},N) $ for some simple or 0-simple monoid $ N $. By Fact \ref{fact: simple}, $ N $ either contains a copy of the bicyclic monoid as a submonoid or $ N $ is a group (a group with 0 adjoined). In the former case \textit{(ii)} holds. 
	
	In the latter case, if $ N $ is a group with zero adjoined, then by Proposition \ref{prop: zero} we know that $ \mathfrak{L}(\textup{Val}, \textup{CF},N) =\mathfrak{L}(\textup{Val}, \textup{CF},G) $ for some group $ G $. If $ G $ is not periodic, then it has an element of infinite order which generates a subgroup isomorphic to $ \mathbb{Z} $ and hence \textit{(iii)} follows.
	Otherwise, suppose that $ G $ is locally finite. By Proposition \ref{prop: kambites}, every language in  $ \mathfrak{L}(\textup{Val},\textup{CF},G) $ belongs to  $ \mathfrak{L}(\textup{Val},\textup{CF},H) $ for some finitely generated subgroup $ H $ of $ G $. Since $ G $ is locally finite, $ H $ is finite. Any language $ \mathfrak{L}(\textup{Val},\textup{CF},H) $ is context-free by a result from \cite{Ze11} and hence $ (i) $ holds. 
	The only remaining case is that $ G $ is a periodic group which is not locally finite, in which case $ (iv) $ holds.
\end{proof}

The result about valence grammars over commutative monoids which we have stated in Fact \ref{fact: fs}, now follows as a corollary of Theorem \ref{thm: characterize}.

\begin{cor}
	Let $ M $ be a commutative monoid. Then $ \mathfrak{L}(\textup{Val},\textup{CF},M) = \mathfrak{L}(\textup{Val},\textup{CF},G) $ for some group $ G $.
\end{cor}
\begin{proof}
	Since no commutative monoid $ M $ can contain a copy of the bicyclic monoid as a submonoid, the result follows by the proof of Theorem  \ref{thm: characterize}.
\end{proof}

\section{Open Questions}\label{sec: val-open}
We proved that a valence PDA over $ M $ is equivalent to a valence automaton over $ \textup{P}_2 \times M$. Can we prove a similar equivalence result for context-free valence grammars?

In Theorem \ref{thm: characterize}, we conclude that $ \mathfrak{L}(\textup{Val},\textup{CF},M) $ contains the class
$ \mathfrak{L}(\textup{Val},\textup{CF},\textup{B}) $ when $ M $ is a monoid that contains B. Since B is not commutative, no correspondence with valence PDA exists, and little is known about the class $ \mathfrak{L}(\textup{Val},\textup{CF},\textup{B}) $, except that it contains the set of partially blind one counter languages. What can we say further about $ \mathfrak{L}(\textup{Val},\textup{CF},\textup{B}) $?

\chapter{CONCLUSION} \label{chap: con}

The main purpose of this thesis is to explore various computational models with storage. We mainly examine extended finite automata and homing vector automata, investigating the language recognition power of these models under different restrictions. We also present several results about valence pushdown automata and context-free valence grammars.

Matrices have a fundamental role in this study. Focusing on finite automata over matrix groups, both models can be regarded as a finite automaton equipped with a storage mechanism that is modified through matrix multiplications. 

 We first examine the class of languages recognized by finite automata over rational and integer matrix groups and compare them with the previously known language classes. Our findings can be summarized as follows:
\begin{itemize}
	\item Finite automata over the group of $ 2 \times 2 $ integer matrices recognize exactly the class of context-free languages.
	\item Finite automata over the group of $ 2 \times 2$ rational matrices with determinant 1 can recognize some non-context-free languages.
\end{itemize}

Together with the previous results, the overall picture is given in Figure \ref{fig:efa}. The question mark indicates that the existence of a language in the particular subset is unknown.
%which state the existence of an Heisenberg automaton recognizing a context-free language and the fact that the class of languages recognized by a finite automaton over the direct product of two free groups of rank 2 is exactly the class of Turing recognizable languages, 
\begin{figure}[h]
	\centering
	\includegraphics[width=1\linewidth]{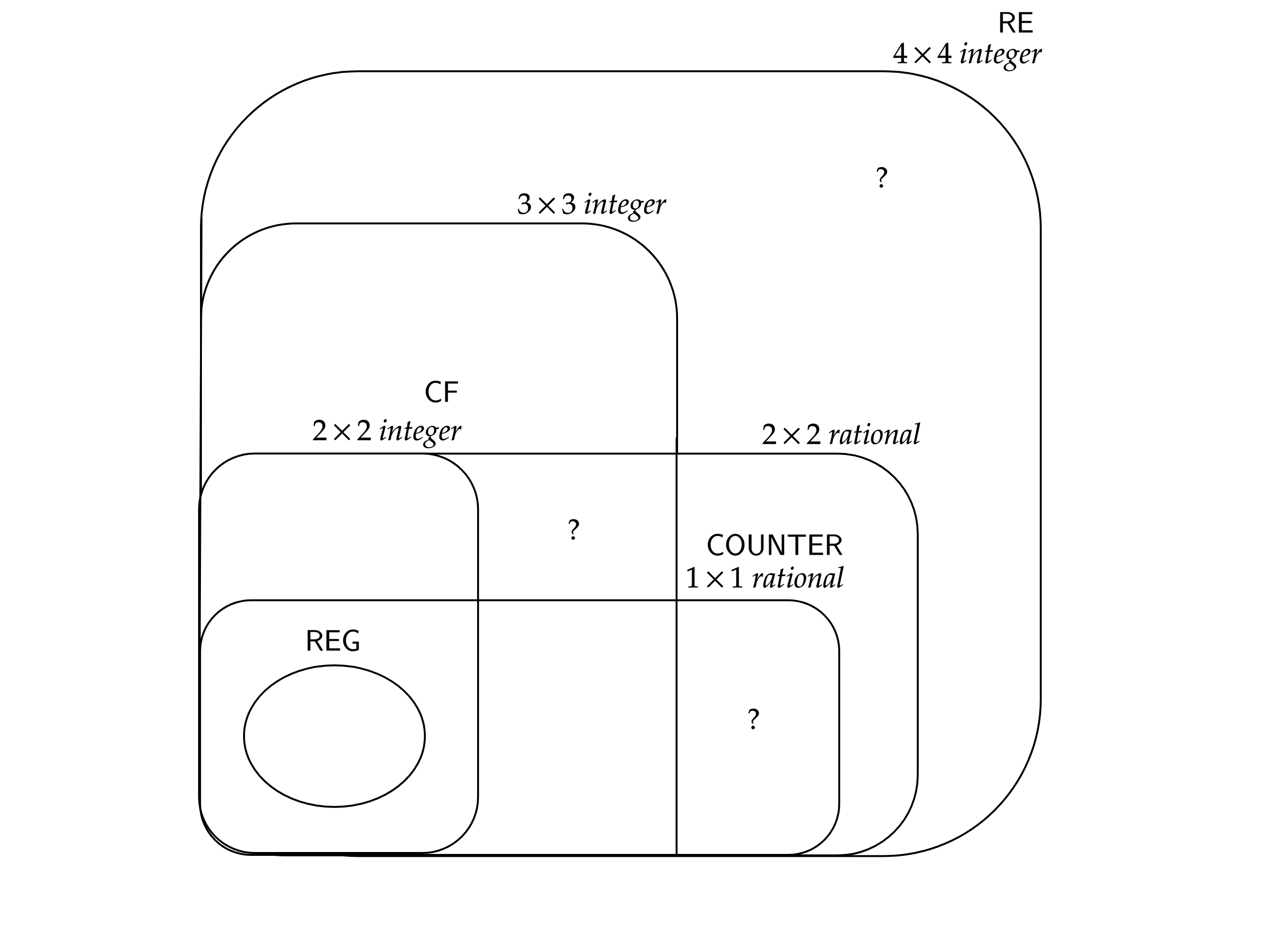}
	\caption{The hierarchy of the classes of languages recognized by extended finite automata over integer and rational matrix groups}
	\label{fig:efa}
\end{figure}

Next we analyze the language recognition power of extended finite automata under time restriction. We prove that the growth rate of the group is effective in the language recognition power of the corresponding group automaton, when a time bound is imposed on the machine. Using this result we prove that there exists a context-free language which cannot be recognized by any $ G $-automata in polynomial time if the group $ G $ has polynomial growth. Furthermore, we investigate the class of languages recognized by group automata in linear-time and obtain the following results:
\begin{itemize}
	\item Free group automata recognize exactly the class of context-free languages in linear time. 
	\item Heisenberg group automata cannot recognize some context-free languages in linear time.
	\item The class of languages recognized by finite automata over the group of $ 3 \times 3$ integer matrices is a proper superclass of the class of languages recognized by Heisenberg group automata under the restriction of linear time.
 \end{itemize}	

We investigate the link between the decidability problems for matrix groups and the corresponding group automata. We provide alternative proofs for the decidability of the subsemigroup membership problem for the group of $ 2 \times 2 $ integer matrices and the identity problem for the monoid of $ 2 \times 2 $ matrices with integer entries. We also prove the undecidability of the emptiness problem for finite automata over $ 4 \times 4 $ integer matrix groups.

Another way we focus on matrices is through the study of the homing vector automaton model, a finite automaton equipped with a vector. Like in many classical models where the acceptance condition is ending with the initial register value, a string is accepted by a homing vector automaton if the vector is equal to its initial value after a series of multiplications with some rational valued matrices. We define different variants of the machine such as real-time, one-way, deterministic, nondeterministic, blind, and non-blind and we add further restrictions on the machine by restricting the matrices multiplied with the register to a specific set. We investigate the relationship between homing vector automata and counter automata, proving that one-dimensional homing vector automata are equivalent to blind counter automata and the two models are incomparable when the computation is not blind. One-way nondeterministic blind homing vector automata are closely linked to extended finite automata over matrix groups and this link extends our knowledge on homing vector automata. We visualize our findings in Figure \ref{fig:hva} and summarize our findings as follows:

\begin{itemize}
	\item The class of languages recognized by one-way nondeterministic blind homing vector automata over the group of $ 2 \times 2 $ integer matrices with determinant 1 contains the class of context-free languages.
	\item  Over the monoid of $ 2 \times 2 $ integer matrices, the class of languages recognized by one-way nondeterministic blind homing vector automata is a proper superset of the class of languages recognized by extended finite automata.
	\item Over the monoid of $ 3 \times 3 $ integer matrices, the class of languages recognized by one-way nondeterministic blind homing vector automata is a proper superset of the class of languages recognized by one-way nondeterministic blind counter automata.
		\item The class of languages recognized by one-way nondeterministic blind homing vector automata over the group of $ 4 \times 4 $ integer matrices with determinant 1 is the class of recursively enumerable languages.
\end{itemize}
\begin{figure}
	\centering
	\includegraphics[width=1\linewidth]{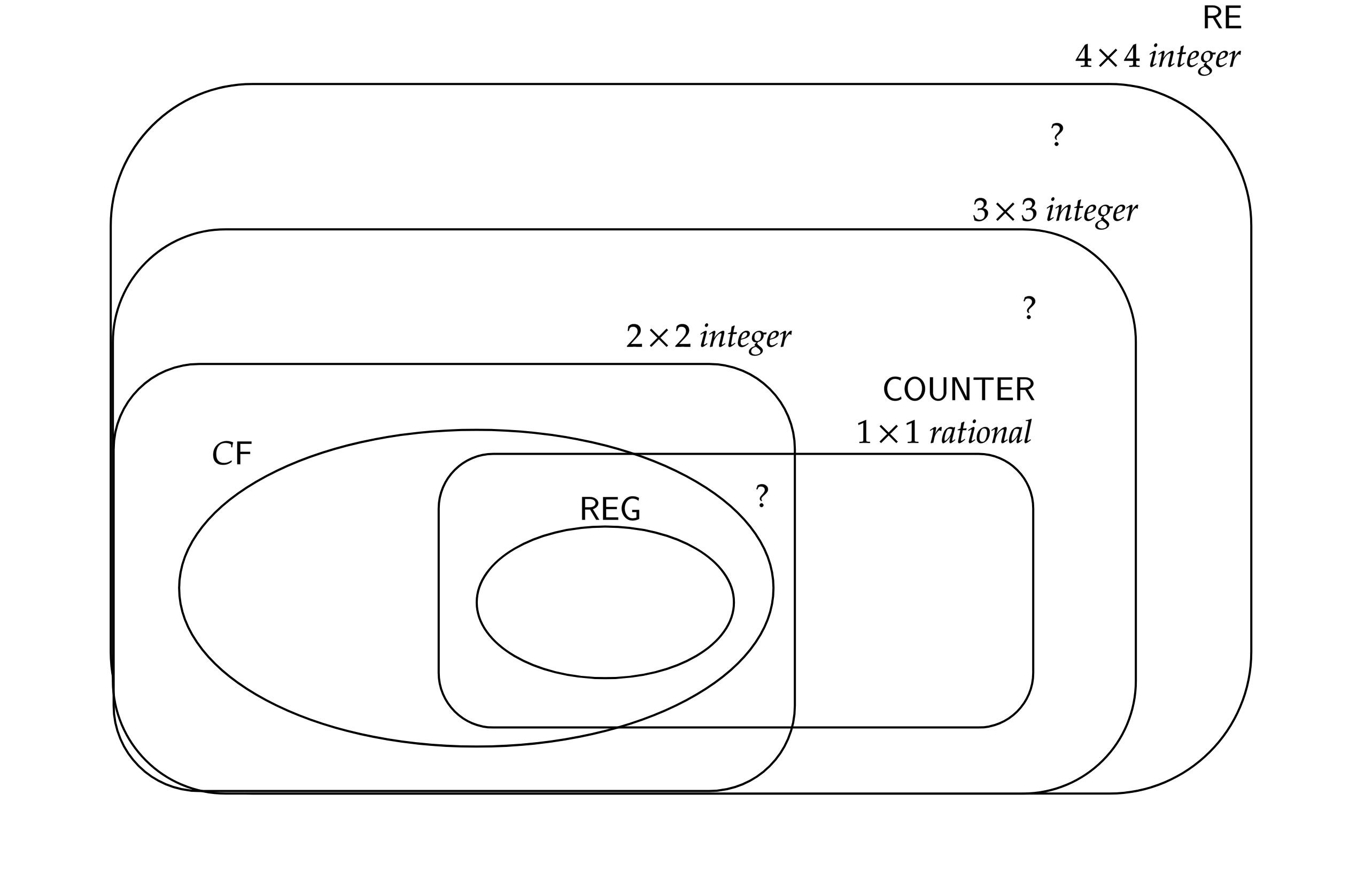}
	\caption{The hierarchy of languages recognized by one-way nondeterministic blind homing vector automata with rational and integer entries}
	\label{fig:hva}
\end{figure}

We propose a new method called the generalized Stern-Brocot encoding, which enables encoding any string over an alphabet of size $ k $ into a $ k $-dimensional vector. The encoding can be carried out in real-time with $ k \times k $ matrices whose entries belong to the set $ \{-1,0,1\} $. 
%This encoding allows recognition of certain languages by real-time blind deterministic homing vector automata and we prove a hierarchy result for deterministic homing vector automata whose matrices belong to that particular set. 
We further explore real-time homing vector automata, by establishing some separation results among the various variants and discussing some closure properties. We prove that any unary language recognized by a real-time deterministic homing vector automaton is regular. Continuing our study on real-time homing vector automata, we investigate their stateless versions. We examine the language recognition power of different versions of these machines and make some observations. Our study of the stateless real-time deterministic blind one-dimensional homing vector automata yields the following characterization for the class of languages recognized by stateless deterministic finite automata with multiplication without equality.
\begin{itemize}
	\item A language is recognized by a stateless deterministic finite automata with multiplication without equality iff it is commutative and its Parikh image is the set of nonnegative integer solutions to a system of linear homogeneous Diophantine equations. 
\end{itemize}   
Finally we look at the problem of string separation by homing vector automata and vector automata, observing that these models can separate any pair of words by using 2-dimensional vectors in the real-time, deterministic and blind computation mode.

Other computational models we investigate are the valence pushdown automata and context-free valence grammars. We prove that valence pushdown automata are equivalent to valence automata defined over some specific monoid. We generalize some results proven in the context of extended finite automata to context-free valence grammars.

%, obtaining the following
%
%\begin{itemize}
%	\item 	Let $ M $ be a monoid. Then the the class of languages generated by context-free valence grammars over the monoid $ M $ either 
%	\begin{enumerate}[(i)]
%		\item equals the class of context-free languages
%		\item contains the class of languages generated by context-free valence grammars over $ \mathbb{Z} $
%		\item contains the class of languages generated by context-free valence grammars over $ \mathbb{Z} $
%		\item is equal to the class of languages generated by context-free valence grammars over the group $ G $ where $ G $ is an infinite periodic group which is not locally finite.		
%	\end{enumerate}
%\end{itemize}  

To conclude, this research opens up a new perspective for analyzing the decision problems of matrix groups and various automaton models, and any computational model than can be traced by matrix multiplications. The study of extended finite automata over matrix groups and under time restriction contributes to the current literature on the topic. The classes of languages recognized by extended finite automata and homing vector automata over integer matrices of dimensions 2, 3, and 4 provides a different point of view for classifying languages. 

Many questions in need of further investigation are presented throughout the thesis. The future research on the subject should concentrate on the language recognition power of extended finite automata and homing vector automata over $ 3 \times 3 $ integer matrices. We conjecture that these classes are the proper subsets of the class of recursively enumerable languages. Note that for the group of $ 3 \times 3 $ integer matrices, the decidability of the membership problem is also open.

\bibliographystyle{fbe_tez_v11}
\bibliography{thesis_arxiv}

\begin{thebibliography}{10}
\newcommand{\enquote}[1]{``#1''}
\expandafter\ifx\csname url\endcsname\relax
  \def\url#1{{\tt #1}}\fi
\expandafter\ifx\csname urlprefix\endcsname\relax\def\urlprefix{}\fi

\bibitem{Si06}
Sipser, M., {\em Introduction to the Theory of Computation, 2nd edition\/},
  Thomson Course Technology, United States of America, 2006.

\bibitem{Ch62}
Chomsky, N., \enquote{Context-free grammars and pushdown storage}, {\em M. I.
  T. Res. Lab. Electron. Quart. Prog. Report.\/}, Vol.~65, pp. 187--194, 1962.

\bibitem{FMR67}
Fischer, P.~C., A.~R. Meyer and A.~L. Rosenberg, \enquote{Real time counter
  machines}, {\em Proceedings of the 8th Annual Symposium on Switching and
  Automata Theory, {SWAT'67} {(FOCS)}\/}, pp. 148--154, 1967.

\bibitem{Tu37}
Turing, A.~M., \enquote{On computable numbers, with an application to the
  Entscheidungsproblem}, {\em Proceedings of the London mathematical
  society\/}, Vol.~2, No.~1, pp. 230--265, 1937.

\bibitem{MS97}
Mitrana, V. and R.~Stiebe, \enquote{The accepting power of finite automata over
  groups}, {\em New Trends in Formal Languages\/}, pp. 39--48, Springer-Verlag,
  1997.

\bibitem{ISK76}
Ibarra, O.~H., S.~K. Sahni and C.~E. Kim, \enquote{Finite automata with
  multiplication}, {\em Theoretical Computer Science\/}, Vol.~2, No.~3, pp. 271
  -- 294, 1976.

\bibitem{Pa80}
P\u{a}un, G., \enquote{A new generative device: valence grammars}, {\em Rev.
  Roumaine Math. Pures Appl\/}, Vol.~25, No.~6, pp. 911--924, 1980.

\bibitem{FS02}
Fernau, H. and R.~Stiebe, \enquote{Sequential grammars and automata with
  valences}, {\em Theoretical Computer Science\/}, Vol. 276, No. 1–2, pp. 377
  -- 405, 2002.

\bibitem{SYS13}
Salehi, {\"{O}}., A.~Yakary{\i}lmaz and A.~C.~C. Say, \enquote{Real-time vector
  automata}, {\em Proceedings of the 19th International Conference on
  Fundamentals of Computation Theory, {FCT'13}\/}, Vol. 8070 of {\em {LNCS}\/},
  pp. 293--304, Springer, 2013.

\bibitem{SDS16}
Salehi, {\"{O}}., F.~D'Alessandro and A.~C.~C. Say, \enquote{Language classes
  associated with automata over matrix groups}, {\em Proceedings of the Eighth
  Workshop on Non-Classical Models of Automata and Applications, {NCMA'16}\/},
  Vol. 321 of {\em books@ocg.at\/}, pp. 287--300, {\"{O}}sterreichische
  Computer Gesellschaft, 2016.

\bibitem{SDS18}
Salehi, {\"O}., F.~D’Alessandro and A.~C. Say, \enquote{Language classes
  associated with automata over matrix groups}, {\em RAIRO-Theoretical
  Informatics and Applications\/}, Vol.~52, No. 2-3-4, pp. 253--268, 2018.

\bibitem{SS18}
Salehi, {\"{O}}. and A.~C.~C. Say, \enquote{Extended finite automata and
  decision problems for matrix semigroups.}, {\em Tenth Workshop on
  Non-Classical Models of Automata and Applications (Short Papers)\/}, Vol.
  321, pp. 45--52, {\"{O}}sterreichische Computer Gesellschaft, 2018.

\bibitem{Tu69}
Turakainen, P., \enquote{Generalized automata and stochastic languages}, {\em
  Proceedings of the American Mathematical Society\/}, Vol.~21, pp. 303--309,
  1969.

\bibitem{LR14}
Lipton, R.~J. and K.~W. Regan, {\em Quantum Algorithms via Linear Algebra\/},
  MIT Press, 2014.

\bibitem{SS15}
Salehi, {\"{O}}. and A.~C.~C. Say, \enquote{Homing vector automata}, {\em
  Proceedings of the Seventh Workshop on Non-Classical Models of Automata and
  Applications, {NCMA'15}\/}, Vol. 318 of {\em books@ocg.at\/}, pp. 193--205,
  {\"{O}}sterreichische Computer Gesellschaft, 2015.

\bibitem{SSD16}
Salehi, {\"{O}}., A.~C.~C. Say and F.~D'Alessandro, \enquote{Homing vector
  automata}, {\em RAIRO - Theoretical Informatics and Applications\/}, Vol.~50,
  No.~4, pp. 371--386, 2016.

\bibitem{SYS19}
Salehi, {\"O}., A.~Yakaryılmaz and A.~C.~C. Say, \enquote{New results on
  vector automata and homing vector automata}, {\em International Journal of
  Foundations of Computer Science\/}, 2019, accepted.

\bibitem{SDS17}
Salehi, {\"{O}}., F.~D'Alessandro and A.~C.~C. Say, \enquote{Generalized
  results on monoids as memory}, {\em Proceedings of the 15th International
  Conference on Automata and Formal Languages,{AFL'17}\/}, Vol. 252 of {\em
  {EPTCS}\/}, pp. 234--247, 2017.

\bibitem{Mi60}
Minsky, M., {\em Recursive Unsolvability of Post's Problem of ``tag'': And
  Other Topics in Theory of Truing Machines\/}, Massachusetts Institute of
  Technology, Lincoln Laboratory, 1960.

\bibitem{FMR68}
Fischer, P.~C., A.~R. Meyer and A.~L. Rosenberg, \enquote{Counter machines and
  counter languages}, {\em Mathematical Systems Theory\/}, Vol.~2, No.~3, pp.
  265--283, 1968.

\bibitem{Gr76}
Greibach, S.~A., \enquote{Remarks on the complexity of nondeterministic counter
  languages}, {\em Theoretical Computer Science\/}, Vol.~1, No.~4, pp. 269 --
  288, 1976.

\bibitem{Pe11}
Petersen, H., \enquote{Simulations by time-bounded counter machines}, {\em
  International Journal of Foundations of Computer Science\/}, Vol.~22, pp.
  395--409, 2011.

\bibitem{Gr78}
Greibach, S.~A., \enquote{Remarks on blind and partially blind one-way
  multicounter machines}, {\em Theoretical Computer Science\/}, Vol.~7, pp.
  311--324, 1978.

\bibitem{Fr03}
Fraleigh, J. and V.~Katz, {\em A First Course in Abstract Algebra\/},
  Addison-Wesley world student series, Addison-Wesley, 2003.

\bibitem{LS77}
Lyndon, R.~C. and P.~E. Schupp, {\em Combinatorial Group Theory\/},
  Springer-Verlag, 1977.

\bibitem{Sc27}
Schreier, O., \enquote{Die Untergruppen der freien Gruppen}, {\em Abhandlungen
  aus dem Mathematischen Seminar der Universit{\"a}t Hamburg\/}, Vol.~5, No.~1,
  pp. 161--183, 1927.

\bibitem{Ni21}
Nielsen, J., \enquote{Om Regning med ikke-kommutative Faktorer og dens
  Anvendelse i Gruppeteorien}, {\em Matematisk Tidsskrift. B\/}, pp. 77--94,
  1921.

\bibitem{Gi96}
Gilman, R., \enquote{Formal languages and infinite groups}, {\em Geometric and
  computational perspectives on infinite groups\/}, pp. 27--51, 1996.

\bibitem{ES69}
Eilenberg, S. and M.~Schützenberger, \enquote{Rational sets in commutative
  monoids}, {\em Journal of Algebra\/}, Vol.~13, No.~2, pp. 173 -- 191, 1969.

\bibitem{Co05}
Corson, J.~M., \enquote{Extended finite automata and word problems}, {\em
  International Journal of Algebra and Computation\/}, Vol.~15, No.~03, pp.
  455--466, 2005.

\bibitem{DM00}
Dassow, J. and V.~Mitrana, \enquote{Finite automata over free groups}, {\em
  International Journal of Algebra and Computation\/}, Vol.~10, No.~06, pp.
  725--737, 2000.

\bibitem{MS01}
Mitrana, V. and R.~Stiebe, \enquote{Extended finite automata over groups}, {\em
  Discrete Applied Mathematics\/}, Vol. 108, No.~3, pp. 287--300, 2001.

\bibitem{Ka09}
Kambites, M., \enquote{Formal languages and groups as memory}, {\em
  Communications in Algebra\/}, Vol.~37, No.~1, pp. 193--208, 2009.

\bibitem{RK09}
Render, E. and M.~Kambites, \enquote{Rational subsets of polycyclic monoids and
  valence automata}, {\em Inf. Comput.\/}, Vol. 207, No.~11, pp. 1329--1339,
  2009.

\bibitem{Re10}
Render, E., {\em Rational Monoid and Semigroup Automata\/}, Ph.D. Thesis,
  University of Manchester, 2010.

\bibitem{EO04}
Elston, G.~Z. and G.~Ostheimer, \enquote{On groups whose word problem is solved
  by a~counter automaton}, {\em Theoretical Computer Science\/}, Vol. 320, No.
  2–3, pp. 175 -- 185, 2004.

\bibitem{Ka06}
Kambites, M., \enquote{Word problems recognisable by deterministic blind monoid
  automata}, {\em Theoretical Computer Science\/}, Vol. 362, No.~1, pp.
  232--237, 2006.

\bibitem{EKO08}
Elder, M., M.~ and G.~Ostheimer, \enquote{On groups and counter automata}, {\em
  International Journal of Algebra and Computation\/}, Vol.~18, No.~08, pp.
  1345--1364, 2008.

\bibitem{CR15}
Corson, J.~M. and L.~L. Ross, \enquote{Automata with counters that recognize
  word problems of free products}, {\em International Journal of Foundations of
  Computer Science\/}, Vol.~26, No.~01, pp. 79--98, 2015.

\bibitem{RCR17}
Bishop-Ross, R., J.~M. Corson and J.~L. Ross, \enquote{Context-sensitive
  languages and G-automata}, {\em International Journal of Algebra and
  Computation\/}, Vol.~27, No.~02, pp. 237--249, 2017.

\bibitem{Ze16}
Zetzsche, G., {\em Monoids as Storage Mechanisms\/}, Phd thesis, Technische
  Universit{\"a}t Kaiserslautern, 2016.

\bibitem{CEO06}
Cleary, S., M.~Elder and G.~Ostheimer, \enquote{The word problem distinguishes
  counter languages}, {\em arXiv preprint math/0606415\/}, 2006.

\bibitem{KM79}
Kargapolov, M.~I. and J.~I. Merzljakov, {\em Fundamentals of the Theory of
  Groups\/}, Springer-Verlag, 1979.

\bibitem{BO08}
Brown, N.~P. and N.~Ozawa, {\em C*-Algebras and Finite-Dimensional
  Approximations\/}, Vol.~88, American Mathematical Soc., 2008.

\bibitem{Gi90}
Grigorchuk, R.~I., \enquote{On growth in group theory}, {\em Proceedings of the
  International Congress of Mathematicians\/}, Vol.~1, pp. 325--338, 1990.

\bibitem{NP70}
Nivat, M. and J.~F. Perrot, \enquote{Une g{\'e}n{\'e}ralisation du mono\"{\i}de
  bicyclique}, {\em Comptes Rendus de l’Acad{\'e}mie des Sciences de
  Paris\/}, Vol. 271, pp. 824--827, 1970.

\bibitem{Ni70}
Nivat, M., \enquote{Sur les automates a m{\'e}moire pile}, {\em Proc. of
  International Computing Symposium\/}, pp. 221--225, 1970.

\bibitem{Ze13}
Zetzsche, G., \enquote{Silent transitions in automata with storage}, {\em
  International Colloquium on Automata, Languages, and Programming\/}, pp.
  434--445, Springer Berlin Heidelberg, 2013.

\bibitem{Ha00}
La~Harpe, P.~D., {\em Topics in Geometric Group Theory\/}, The University Of
  Chicago Press, Chicago, 2000.

\bibitem{GS98}
Glaister, I. and J.~Shallit, \enquote{Automaticity III: polynomial automaticity
  and context-free languages}, {\em Computational Complexity\/}, Vol.~7, No.~4,
  pp. 371--387, 1998.

\bibitem{Zak83}
\.{Z}ak, S., \enquote{A {T}uring machine time hierarchy}, {\em Theoretical
  Computer Science\/}, Vol.~26, No.~3, pp. 327 -- 333, 1983.

\bibitem{De11}
Dehn, M., \enquote{{\"U}ber unendliche diskontinuierliche Gruppen}, {\em
  Mathematische Annalen\/}, Vol.~71, No.~1, pp. 116--144, 1911.

\bibitem{Mi58}
Mikhailova, K.~A., \enquote{The occurrence problem for direct products of
  groups}, {\em Dokl. Akad. Nauk SSSR\/}, Vol. 119, No.~6, pp. 1103--1105,
  1958.

\bibitem{Pa70}
Paz, A., \enquote{Events which are not representable by a probabilistic
  automaton}, {\em SIGACT News\/}, , No.~4, pp. 8--11, 1970.

\bibitem{COSW19}
Colcombet, T., J.~Ouaknine, P.~Semukhin and J.~Worrell, \enquote{On
  reachability problems for low dimensional matrix semigroups}, {\em arXiv
  preprint arXiv:1902.09597\/}, 2019.

\bibitem{CK05}
Choffrut, C. and J.~Karhum{\"a}ki, \enquote{Some decision problems on integer
  matrices}, {\em RAIRO-Theoretical Informatics and Applications\/}, Vol.~39,
  No.~1, pp. 125--131, 2005.

\bibitem{PS17p}
Potapov, I. and P.~Semukhin, \enquote{Membership problem in GL(2, $\mathbb{Z}$)
  extended by singular matrices}, {\em LIPIcs-Leibniz International Proceedings
  in Informatics\/}, Vol.~83, Schloss Dagstuhl-Leibniz-Zentrum fuer Informatik,
  2017.

\bibitem{PS17}
Potapov, I. and P.~Semukhin, \enquote{Decidability of the membership problem
  for 2$\times$ 2 integer matrices}, {\em Proceedings of the Twenty-Eighth
  Annual ACM-SIAM Symposium on Discrete Algorithms\/}, pp. 170--186, SIAM,
  2017.

\bibitem{KNP17}
Ko, S.-K., R.~Niskanen and I.~Potapov, \enquote{On the identity problem for the
  special linear group and the Heisenberg group}, {\em 45th International
  Colloquium on Automata, Languages, and Programming (ICALP 2018)\/}, Vol. 107
  of {\em Leibniz International Proceedings in Informatics (LIPIcs)\/}, pp.
  132:1--132:15, 2018.

\bibitem{BP10}
Bell, P.~C. and I.~Potapov, \enquote{On the undecidability of the identity
  correspondence problem and its applications for word and matrix semigroups},
  {\em International Journal of Foundations of Computer Science\/}, Vol.~21,
  No.~06, pp. 963--978, 2010.

\bibitem{KSS07}
Kambites, M., P.~V. Silva and B.~Steinberg, \enquote{On the rational subset
  problem for groups}, {\em Journal of Algebra\/}, Vol. 309, No.~2, pp.
  622--639, 2007.

\bibitem{Lo13}
Lohrey, M., \enquote{The rational subset membership problem for groups: a
  survey}, {\em Groups St Andrews\/}, Vol. 422, pp. 368--389, 2013.

\bibitem{Ze15}
Zetzsche, G., \enquote{The emptiness problem for valence automata or: Another
  decidable extension of Petri nets}, {\em International Workshop on
  Reachability Problems\/}, pp. 166--178, Springer, 2015.

\bibitem{BR84}
Baumslag, G. and J.~E. Roseblade, \enquote{Subgroups of direct products of free
  groups}, {\em Journal of the London Mathematical Society\/}, Vol.~2, No.~1,
  pp. 44--52, 1984.

\bibitem{St58}
Stern, M.~A., \enquote{\"{U}ber eine zahlentheoretische \mbox{F}unktion}, {\em
  Journal f\"{u}r die reine und angewandte Mathematik\/}, Vol.~55, pp.
  193--220, 1858.

\bibitem{Br61}
Brocot, A., \enquote{{Calcul des rouages par approximation, nouvelle
  m\'{e}thode}}, {\em Revue Chronom\'{e}trique\/}, Vol.~3, pp. 186--194, 1861.

\bibitem{GKP89}
Graham, R., D.~Knuth and O.~Patashnik, {\em Concrete Mathematics: A foundation
  for computer science\/}, Addison-Wesley, 1989.

\bibitem{Ga13}
Garrity, T., \enquote{A multidimensional continued fraction generalization of
  Stern’s diatomic sequence}, {\em Journal of Integer Sequences\/}, Vol.~16,
  No.~2, p.~3, 2013.

\bibitem{RK10}
Render, E. and M.~Kambites, \enquote{Semigroup automata with rational initial
  and terminal sets}, {\em Theoretical Computer Science\/}, Vol. 411, No. 7-9,
  pp. 1004--1012, 2010.

\bibitem{GK86}
Goral{\v{c}}{\'\i}k, P. and V.~Koubek, \enquote{On discerning words by
  automata}, {\em International Colloquium on Automata, Languages, and
  Programming, {ICALP'86}\/}, Vol. 226 of {\em LNCS\/}, pp. 116--122, Springer,
  1986.

\bibitem{YDI08}
Yang, L., Z.~Dang and O.~H. Ibarra, \enquote{On stateless automata and {P}
  Systems}, {\em Int. J. Found. Comput. Sci.\/}, Vol.~19, No.~5, pp.
  1259--1276, 2008.

\bibitem{IKO10}
Ibarra, O.~H., J.~Karhum{\"{a}}ki and A.~Okhotin, \enquote{On stateless
  multihead automata: Hierarchies and the emptiness problem}, {\em Theoretical
  Computer Science\/}, Vol. 411, No.~3, pp. 581--593, 2010.

\bibitem{KMO09}
Kutrib, M., H.~Messerschmidt and F.~Otto, \enquote{On stateless deterministic
  restarting automata}, {\em 35th Conference on Current Trends in Theory and
  Practice of Computer Science, {SOFSEM'09}\/}, Vol. 5404 of {\em {LNCS}\/},
  pp. 353--364, Springer, 2009.

\bibitem{Pau00}
P{\u{a}}un, G., \enquote{Computing with membranes}, {\em Journal of Computer
  and System Sciences\/}, Vol.~61, No.~1, pp. 108--143, 2000.

\bibitem{Le93}
Lehmer, D.~H., \enquote{A note on trigonometric algebraic numbers}, {\em The
  American Mathematical Monthly\/}, Vol.~40, No.~3, pp. 165--166, 1933.

\bibitem{FS97}
Fernau, H. and R.~Stiebe, \enquote{Regulation by valences}, {\em International
  Symposium on Mathematical Foundations of Computer Science\/}, pp. 239--248,
  Springer, 1997.

\bibitem{Ho01}
Hoogeboom, H.~J., \enquote{Context-free valence grammars-revisited}, {\em
  International Conference on Developments in Language Theory\/}, pp. 293--303,
  Springer, 2001.

\bibitem{Ze11}
Zetzsche, G., \enquote{On the capabilities of grammars, automata, and
  transducers controlled by monoids}, {\em Proceedings of the 38th
  International Conference on Automata, Languages and Programming - Volume Part
  II\/}, ICALP'11, pp. 222--233, Springer-Verlag, 2011.

\bibitem{Ho95}
Howie, J.~M., {\em Fundamentals of Semigroup Theory\/}, Clarendon Oxford, 1995.

\end{thebibliography}
\end{document}